\documentclass[
  aps,
  pra,
  reprint,
  superscriptaddress,
]{revtex4-2}

\usepackage[tbtags]{mathtools}
\usepackage{amsthm}
\usepackage{newtx}
\usepackage{bm}

\usepackage{derivative}
\usepackage{mleftright}

\usepackage{microtype}
\usepackage[dvipsnames]{xcolor}

\usepackage[
  colorlinks=true,
  linkcolor=blue,
  citecolor=blue,
  urlcolor=MidnightBlue,
]{hyperref}

\NewCommandCopy{\mc}{\mathcal}
\NewCommandCopy{\mb}{\mathbb}

\NewDocumentCommand{\id}{}{I}

\NewDocumentCommand{\C}{}{\mb C}
\NewDocumentCommand{\R}{}{\mb R}
\NewDocumentCommand{\Z}{}{\mb Z}

\NewDocumentCommand{\F}{}{\mc F}
\NewDocumentCommand{\G}{}{\mc G}
\NewDocumentCommand{\K}{}{\mc K}

\let\Re\relax

\DeclareMathOperator{\Pf}{Pf}
\DeclareMathOperator{\Re}{Re}
\DeclareMathOperator{\SR}{SR}
\DeclareMathOperator{\Tr}{Tr}
\DeclareMathOperator{\Var}{Var}
\DeclareMathOperator{\diag}{diag}
\DeclareMathOperator{\im}{im}
\DeclareMathOperator{\rank}{rank}
\DeclareMathOperator{\spn}{span}

\DeclarePairedDelimiter{\set}{\lbrace}{\rbrace}
\DeclarePairedDelimiter{\abs}{\lvert}{\rvert}
\DeclarePairedDelimiter{\norm}{\lVert}{\rVert}
\DeclarePairedDelimiter{\ceil}{\lceil}{\rceil}
\DeclarePairedDelimiter{\floor}{\lfloor}{\rfloor}
\DeclarePairedDelimiter{\evaluate}{.}{\rvert}
\DeclarePairedDelimiter{\bra}{\langle}{\rvert}
\DeclarePairedDelimiter{\ket}{\lvert}{\rangle}
\DeclarePairedDelimiterX{\braket}[2]{\langle}{\rangle}{#1\delimsize\vert#2}

\NewDocumentCommand{\chiG}{}{\chi_{\mathrm{G}}}
\NewDocumentCommand{\bchiG}{}{\underline{\chi}_{\mathrm{G}}}
\NewDocumentCommand{\achiG}{}{\chi_{\mathrm{G}}^{(\delta)}}
\NewDocumentCommand{\abchiG}{}{\underline{\chi}_{\mathrm{G}}^{(\delta)}}

\theoremstyle{plain}
\newtheorem{theorem}{Theorem}
\newtheorem{lemma}[theorem]{Lemma}
\newtheorem{corollary}[theorem]{Corollary}
\newtheorem{proposition}[theorem]{Proposition}

\theoremstyle{definition}
\newtheorem{definition}[theorem]{Definition}

\theoremstyle{remark}

\NewDocumentCommand{\wgr}{m}{\textcolor{MidnightBlue}{(wgr:\space#1)}}
\NewDocumentEnvironment{zwblock}{}{\color{BrickRed}zw:\space\ignorespaces}{\ignorespacesafterend}
\NewDocumentEnvironment{wfcblock}{}{\color{ForestGreen}wfc:\space\ignorespaces}{\ignorespacesafterend}
\NewDocumentEnvironment{wgrblock}{}{\color{MidnightBlue}wgr:\space\ignorespaces}{\ignorespacesafterend}

\RenewDocumentCommand{\wgr}{m}{}
\RenewDocumentEnvironment{wgrblock}{+b}{}{}

\begin{document}

\title{Robust exponential lower bounds for fermionic and bosonic Gaussian ranks}

\author{Fuchuan Wei}
\thanks{These authors contributed equally to this work.}
\affiliation{Yau Mathematical Sciences Center, Tsinghua University, Beijing 100084, China}

\author{Kong-Wing Wu}
\thanks{These authors contributed equally to this work.}
\affiliation{Qiuzhen College, Tsinghua University, Beijing 100084, China}

\author{Zhengwei Liu}
\email{liuzhengwei@mail.tsinghua.edu.cn}
\affiliation{Yau Mathematical Sciences Center, Tsinghua University, Beijing 100084, China}
\affiliation{Department of Mathematics, Tsinghua University, Beijing 100084, China}
\affiliation{Yanqi Lake Beijing Institute of Mathematical Sciences and Applications, Beijing 100407, China}

\author{Zi-Wen Liu}
\email{zwliu0@tsinghua.edu.cn}
\affiliation{Yau Mathematical Sciences Center, Tsinghua University, Beijing 100084, China}

\date{\today}

\begin{abstract}
  The power and limitations of classical simulation are central to understanding quantum computational advantages.
  A leading simulation paradigm is based on coherent decomposition into classically tractable free states where the  decomposition rank determines the simulation complexity.
  Proving strong lower bounds on this number is a notoriously difficult and mathematically rich problem, as exemplified by the qubit stabilizer rank problem.
  Here, we study the Gaussian version of this problem in both bosonic and fermionic systems and establish robust exponential lower bounds on Gaussian rank.
  In particular, we prove that for every pure non-Gaussian state on finitely many modes, with definite parity in the fermionic case, the approximate border Gaussian rank of its tensor powers grows at least exponentially at any fixed norm error below one.
  Our proofs combine reduction to four modes with Majorana spectral bounds for fermions, and Gaussian postselection with an entropy-based rank bound for bosons.
  As concrete examples, we derive explicit exponential lower bounds for the four-mode fermionic GHZ state and the bosonic single-photon state.
  The bosonic results require no assumption on the energy of the target state.
  Our results show that non-Gaussianity universally entails exponential Gaussian decomposition complexity, setting fundamental limitations on decomposition-based classical simulation of bosonic and fermionic systems.

\end{abstract}

\maketitle

\section{Introduction}
\label{sec:unified-framework}

Understanding and characterizing the advantage of quantum over classical computation is a central quest of quantum information science.
This motivates the extensive study of  the efficiency of classically simulating quantum systems and dynamics: the more efficient general classical simulation methods can be, the less room there is for quantum speedups.
A standard framework for classical simulation is based on decomposing quantum states or operations into combinations of classically tractable components whose properties can be computed efficiently and then combining their contributions to achieve the overall simulation~\cite{bravyi2015tradingclassicalquantumcomputational,bravyi2019simulationquantumcircuitslowrank,dias2025samplingcomplexitycoherentsuperpositions}. In this formalism, the free decomposition rank, namely the minimum number of free components needed in the decomposition, is a key parameter governing the  simulation cost.

A prominent example is the longstanding \emph{stabilizer rank} problem, which comes down to the question of how the minimum number of stabilizer states in a coherent decomposition of $T$-state tensor powers grows with the number of copies, which controls the cost of stabilizer-decomposition-based simulation of universal quantum circuits~\cite{bravyi2015tradingclassicalquantumcomputational,bravyi2019improvedclassicalsimulationquantum,bravyi2019simulationquantumcircuitslowrank}.
In particular, subexponential stabilizer rank growth, namely the existence of subexponential-size stabilizer decompositions, would enable subexponential-time classical simulation of universal quantum computation and thereby rules out any exponential quantum speedup.
However, despite its evident importance and sustained efforts over the past decade~\cite{bravyi2015tradingclassicalquantumcomputational,bravyi2019improvedclassicalsimulationquantum,bravyi2019simulationquantumcircuitslowrank,huang2018approximatestabilizerrankimproved,qassim2019cliffordrecompilationfasterclassical,kocia2022improvedstrongsimulationuniversal,huang2021feynmanpathtypesimulationusing,qassim2021improvedupperboundsstabilizer,labib2022stabilizerrankhigherorderfourier,peleg2022lowerboundsstabilizerrank,lovitz2022newtechniquesboundingstabilizer,kissinger2022classicalsimulationquantumcircuits,desilva2024basesoptimisingstabiliserdecompositions,mehraban2024quadraticlowerboundsapproximate,mehraban2024improvedboundstestinglow,arunachalam2025learningstabilizerstructurequantum,kjelstrøm2026efficientsimulationhighlevelquantum,gosset2026quantumstatepreparationoptimal,kalra2026stabilizerranksbarneswall,labib2026stabilizerrankboundsmagicstate}, 
proving good lower bounds has remained notoriously difficult, with the best bounds for $T$-state tensor powers remaining only nearly quadratic both for exact rank and at fixed approximation error~\cite{mehraban2024quadraticlowerboundsapproximate,gosset2026quantumstatepreparationoptimal}, leaving a substantial gap to the conjectured exponential lower bound.

For fermionic and bosonic systems, the corresponding classically tractable free states are Gaussian states.
Fermionic linear optics with Gaussian inputs and occupation-number measurements admits efficient classical simulation~\cite{valiant2002quantumcircuitssimulatedclassicallypolynomialtime,terhal2001classicalsimulationnoninteractingfermionquantum,bravyi2004lagrangianrepresentationfermioniclinear,jozsa2008matchgatesclassicalsimulationquantum}, as do bosonic Gaussian circuits with Gaussian inputs, homodyne or heterodyne detection, and adaptive feed-forward~\cite{bartlett2002efficientclassicalsimulationcontinuous}.
As with magic in the stabilizer setting, non-Gaussian resources extend these models to universal computation: pure non-Gaussian parity eigenstates enable universal matchgate computation, while cubic-phase resources supply the nonlinearity for continuous-variable universality~\cite{hebenstreit2019purefermionicnongaussianstates,lloyd1998quantumcomputationcontinuousvariables,gottesman2001encodingqubitoscillator}. This motivates the study of Gaussian rank and related measures of non-Gaussianity, which underlie decomposition-based simulation algorithms in both settings~\cite{cudby2025gaussiandecompositionmagicstates,dias2024classicalsimulationnongaussianfermionic,reardonsmith2024improvedsimulationquantumcircuits,dias2025classicalsimulationnongaussianbosonic,dias2026erratumclassicalsimulationnongaussianbosoniccircuits,hahn2026classicalsimulationquantumresource,dias2026optimalimprovedgatedecompositions}.

In this work, we resolve the exponential lower bound problem for Gaussian rank in both fermionic and bosonic systems, proving unconditional exponential lower bounds on the approximate Gaussian rank of tensor powers of the standard four-mode fermionic GHZ state and the bosonic single-photon state, and more generally of every fixed pure non-Gaussian state on finitely many modes.
Notably, the Gaussian setting entails subtleties and difficulties absent in the stabilizer setting.
In particular, unlike stabilizer states, Gaussian states form continuous families. 
A sequence of superpositions of $r$ Gaussian states can converge to a target that has no exact decomposition with $r$ terms: the Gaussian states may approach one another while their coefficients diverge and cancel.
We therefore consider the \emph{border Gaussian rank} given by the smallest $r$ for which such $r$-term superpositions approximate the target arbitrarily accurately.
Our bounds allow both these limits and any constant norm error below one, also without input energy assumptions in the bosonic case.

Also note that a recent work proves a superpolynomial coherent-state-rank bound for single-photon tensor powers in the vanishing error limit~\cite{cottier2026lowerboundscoherentstate}.
In comparison, our bounds allow the larger family of all pure Gaussian states including squeezed states, and hold at nonvanishing error, so the coherent-state-rank lower bounds do not imply bounds for the full Gaussian rank studied here. Bosonic simulation has also been studied through stellar structure~\cite{chabaud2021classicalsimulationgaussianquantum}, whose associated rank measures polynomial degree rather than the number of Gaussian summands.



\section{Gaussian ranks and the main result}
\label{sec:gaussian-rank-setting}


In this section, we formally define Gaussian ranks and its border and approximate variants, and summarize our main results.
All logarithms are natural, and entropies are measured in nats.

Throughout, pure fermionic states are assumed to have definite parity.
For a fermionic or bosonic state $\ket{\psi}$, the Gaussian rank $\chiG(\ket{\psi})$ is the smallest number of pure Gaussian states needed to express $\ket{\psi}$ as a linear combination.
In the fermionic case, Gaussian states of either parity are allowed.
The border Gaussian rank $\bchiG(\ket{\psi})$ is the smallest $r$ for which $\ket{\psi}$ can be approximated arbitrarily closely in norm by vectors of Gaussian rank at most $r$.
If no finite $r$ suffices, the corresponding rank is $+\infty$.


The single-photon state $\ket{1}$ illustrates the difference between exact and border Gaussian rank.
A  coherent state of a single bosonic mode is $\ket{\alpha} \coloneqq e^{-\abs{\alpha}^2 / 2} \sum_{m = 0}^{\infty} \frac{\alpha^m}{\sqrt{m!}} \ket{m}$, where $\alpha \in \C$ and $\ket{m}$ is the state with exactly $m$ photons.
These Gaussian states satisfy $(\ket{\varepsilon} - \ket{0}) / \varepsilon \longrightarrow \ket{1}$ in norm as real $\varepsilon \to 0$, so the single-photon state $\ket{1}$ has border Gaussian rank two, although its exact Gaussian rank is infinite~\cite{hahn2026classicalsimulationquantumresource}.
Appendix~\ref{appendix:gaussian-border-gap} proves this statement and also gives a fermionic separation for border and non-border rank.

For a normalized target and $0\le\delta<1$, define
\begin{equation}
  \abchiG(\ket{\psi})
  \coloneqq
  \min_{\norm{\ket{\psi}-\ket{\eta}}_2\le\delta}
       \bchiG(\ket{\eta}),
  \label{eq:unified-approx-border-rank}
\end{equation}
where the approximant $\ket{\eta}$ need not be normalized.
A lower bound on this quantity also bounds exact, border, and approximate Gaussian rank.
The restriction $\delta<1$ excludes the zero approximant.
Replacing $\bchiG$ with $\chiG$ defines the approximate Gaussian rank $\achiG$.
For $0<\delta<1$, these quantities satisfy
\begin{equation}
  \abchiG(\ket{\psi})
  \le\achiG(\ket{\psi})
  \le\bchiG(\ket{\psi})
  \le\chiG(\ket{\psi}).
  \label{eq:setting-rank-hierarchy}
\end{equation}
At $\delta=0$, the approximate ranks reduce to their corresponding exact or border versions.
Therefore, a lower bound on $\abchiG$ controls all four notions.
Full definitions and derivations are given in in Appendices~\ref{appendix:gaussian-rank-definitions} and~\ref{appendix:gaussian-border-gap}.

Non-Gaussianity alone forces exponential rank growth:

\begin{theorem}[Robust exponential Gaussian rank]
  \label{thm:unified-dichotomy}
  Let $\ket{\psi}$ be a fixed normalized pure non-Gaussian state on finitely many fermionic or bosonic modes.
  There exists $c_\psi>0$ such that, for every fixed $0\le\delta<1$,
  \begin{equation}
    \abchiG(\ket{\psi}^{\otimes k})\ge e^{c_\psi k}
    \label{eq:unified-exponential-rank}
  \end{equation}
  for all sufficiently large $k$.
\end{theorem}

The positive rate $c_\psi$ is independent of $\delta$, although how large $k$ must be may depend on $\delta$.
Since Gaussian states satisfy $\abchiG (\ket{\psi}^{\otimes k}) = 1$, this theorem gives a dichotomy for the rank behavior of pure states.

Specifically, for the four-mode fermionic state $\ket{M} = (\ket{0000} + \ket{1111}) / \sqrt{2}$ and the bosonic single-photon state $\ket{1}$, Proposition~\ref{proposition:fermion-balanced-main} and Corollary~\ref{thm:boson-single-photon-main} give exponential lower bounds with bases $343 / 243$ and $11 / 10$, respectively.
Table~\ref{tab:canonical-resource-scaling} presents finite-copy lower bounds for exact border rank, obtained by optimizing the fermionic estimate at each copy number and applying a stronger counting bound for the bosonic state.


\begin{table}[t]
  \begin{ruledtabular}
    \begin{tabular}{rrr}
      $k$ & Fermion: $\ket{M}^{\otimes k}$ & Boson: $\ket{1}^{\otimes k}$ \\
      \colrule
      $1$  & $2$ & $2$ \\
      $2$  & $3$ & $2$ \\
      $3$  & $4$ & $3$ \\
      $4$  & $6$ & $3$ \\
      $5$  & $8$ & $4$ \\
      $10$ & $42$ & $11$ \\
      $20$ & $1{,}298$ & $98$ \\
      $50$ & $40{,}267{,}004$ & $78{,}435$
    \end{tabular}
  \end{ruledtabular}
  \caption{
    Finite-copy lower bounds on the border Gaussian rank of tensor powers of the four-mode fermionic state $\ket{M} = (\ket{0000} + \ket{1111}) / \sqrt{2}$ and the bosonic single-photon state $\ket{1}$.
    The fermionic bounds follow from Eq.~\eqref{eq:short-certificate-main}, optimized over $p$ at $\delta = 0$.
    The bosonic bounds follow from Eq.~\eqref{equation:single-photon-border-count} at $p = k$.
    All values are rounded up to integers.
    Both border Gaussian ranks are upper bounded by $2^k$.
  }
  \label{tab:canonical-resource-scaling}
\end{table}


We next present quantitative bounds and the mechanisms behind them for fermions in Sec.~\ref{sec:fermionic-results} and bosons in Sec.~\ref{sec:bosonic-results}.

\section{Fermionic quantitative results}
\label{sec:fermionic-results}


We now ask how efficiently Gaussian superpositions can represent many copies of a fermionic resource state. Even when the Gaussian components may correlate all copies, every non-Gaussian target requires exponentially many terms.

\begin{theorem}[Exponential fermionic Gaussian rank bound]
  \label{thm:fermion-universal-main}
  Let $\ket{\psi}$ be a normalized pure non-Gaussian state on finitely many fermionic modes.
  There exist $p_\psi, t \in (0, 1]$, depending only on $\ket{\psi}$, such that for every fixed $0 \le \delta < 1$ and $0 < c < p_\psi \log B_t$,
  \begin{equation}
    \abchiG (\ket{\psi}^{\otimes k}) \ge e^{c k}
  \end{equation}
  for all sufficiently large $k$.
  Here $B_t > 1$ is defined in Theorem~\ref{thm:weighted-spectral-bound}.
\end{theorem}


To prove this theorem, we first reduce the target state $\ket{\psi}$ to four modes, then establish a lower bound on Gaussian rank there, and transfer the bound back with controlled approximation error.

Four modes are the smallest system supporting a pure non-Gaussian fermionic state of definite parity.
To reduce the target state to four modes, we apply a Gaussian unitary, measure all but four modes in the occupation (computational) basis, and postselect the vacuum outcome.
With probability $p_\psi > 0$, the output is $\ket{M_t} \coloneqq (\ket{0000} + t \ket{1111}) / \sqrt{1 + t^2}$, with $t \in (0, 1]$, up to an overall phase and a Majorana factor for odd target states, both of which preserve Gaussian rank and norm.
Proposition~\ref{prop:postselected-four-mode-reduction} gives the full reduction.

We first consider the case $t = 1$, corresponding to the state
\begin{equation}
  \ket{M} \coloneqq \frac{\ket{0000} + \ket{1111}}{\sqrt2}.
  \label{eq:M-main}
\end{equation}

\begin{proposition}
  \label{proposition:fermion-balanced-main}
  For every $0 \le \delta < 1$, there exists a constant $c_\delta > 0$ such that, for all $k \ge 1$,
  \begin{equation}
    \abchiG (\ket{M}^{\otimes k})
    \ge c_\delta \mleft(\frac{343}{243}\mright)^k.
  \end{equation}
\end{proposition}

We now sketch the argument for exact rank.
For an $n$-mode vector $\ket{\psi}$ and an integer $0 \le p \le n$, let $F_p (\ket{\psi})$ be the matrix whose columns are $\gamma_{j_1} \cdots \gamma_{j_p} \ket{\psi}$ for $1 \le j_1 < \cdots < j_p \le 2 n$.
By linearity of $F_p$ and subadditivity of matrix rank, any decomposition into $r$ Gaussian states $\ket{G_i}$ satisfies
\begin{equation}
  \rank F_p (\ket{\psi}) \le \sum_{i = 1}^r \rank F_p (\ket{G_i}) = r C_{n, p},
\end{equation}
where $C_{n, p}$ is the common  rank of $F_p$ on pure Gaussian states.
For the target $\ket{M}^{\otimes k}$, a suitable local family of Majorana monomials produces orthogonal states, whose tensor products give $D_{k, p}$ orthonormal columns of $F_p (\ket{M}^{\otimes k})$.
Comparing the two matrix ranks gives, for integers $k \ge 1$ and $0 \le p \le 2 k$,
\begin{equation}
  \chiG (\ket{M}^{\otimes k}) \ge \frac{\rank F_p (\ket{M}^{\otimes k})}{C_{4 k, p}} \ge \frac{D_{k, p}}{C_{4 k, p}}.
\end{equation}
Taking $p = \floor{8 k / 9}$ makes this ratio grow as $(343 / 243)^k$ up to positive constant factors.

The above arguments extends to border Gaussian rank, since the matrix rank bound persists under norm limits.
For $0 \le \delta < 1$, orthogonality of the selected columns and norm preservation allow us to apply matrix anlysis to control the error.
Appendix~\ref{appendix:formal-main} gives the construction and counting estimates.

For $0 < t < 1$, the selected excitations are no longer orthogonal, so we use the spectrum of a weighted Gram matrix to control approximation error.
This yields an asymptotic exponential lower bound on approximate border Gaussian rank for every fixed $0 < t \le 1$ and error below one (Theorem~\ref{thm:weighted-spectral-bound}).
The base $B_t > 1$ is continuous, with $B_1 = 343 / 243$ and $B_t \to 1$ as $t \to 0$, where the state approaches the Gaussian vacuum.

Combining the above procedures establishes Theorem~\ref{thm:fermion-universal-main}.
The full argument, including control of the approximation error under Gaussian postselection, is given in Appendix~\ref{appendix:theorem-two-proof}.

\section{Bosonic quantitative results}
\label{sec:bosonic-results}

For bosons, the infinite-dimensional Fock space requires control of photon-number tails as well as rank.
Nevertheless, every fixed pure non-Gaussian target state on finitely many modes has a positive exponential lower bound, without any assumption on its energy.

\begin{theorem}[Exponential bosonic Gaussian rank bound]
  \label{thm:boson-universal-main}
  Let $\ket{\psi}$ be a normalized non-Gaussian pure state on $n$ bosonic modes, with $n \in \Z_{> 0}$.
  There exist $p_\psi \in (0, 1]$ and $B_\psi > 1$, such that for every fixed $0 \le \delta < 1$ and $0 < c < p_\psi \log B_\psi$,
  \begin{equation}
    \abchiG (\ket{\psi}^{\otimes k}) \ge e^{c k}
    \label{eq:main-quantitative}
  \end{equation}
  for all sufficiently large $k$.
\end{theorem}

As a corollary, for the single-photon state $\ket{1}$, we obtain

\begin{corollary}
  \label{thm:boson-single-photon-main}
  For every fixed $0 \le \delta < 1$ and all sufficiently large integers $k$,
  \begin{equation}
    \abchiG (\ket{1}^{\otimes k})
    \ge \mleft(\frac{11}{10}\mright)^k.
    \label{eq:boson-single-photon-main}
  \end{equation}
\end{corollary}

At zero error, a direct Schmidt-rank count gives the stronger border-rank bound
\begin{equation}
  \bchiG (\ket{1}^{\otimes k})
  > (5/4)^k
  \label{equation:single-photon-border-bound}
\end{equation}
for every integer $k \ge 1$, as proved in Appendix~\ref{appendix:boson-single-photon-rate}.

To prove Theorem~\ref{thm:boson-universal-main}, we first reduce the target by Gaussian postselection to states near a nonvacuum Fock state, then establish an exponential rank bound for their products and transfer it back with controlled approximation error.
This reduction succeeds with probability $p_\psi > 0$.

Now we bound the Gaussian rank of products of these reduced states. We first consider exact decompositions.
Couple a normalized $n$-mode  $\ket{\phi}$ to a finite-mode environment by a Gaussian isometry, and denote the environment state by $\rho_\phi$.
A cutoff at $p \in \Z_{\ge 0}$ environment photons bounds each Gaussian term's Schmidt rank by $\binom{n + p}{p}$ (Lemma~\ref{lem:photon-cap}).
An exact decomposition of $\ket{\phi}$ into $r$ Gaussian vectors therefore gives a truncated output of Schmidt rank at most $r \binom{n + p}{p}$.
If the cutoff discards probability at most $\varepsilon$, with $0 \le \varepsilon < 1$, the largest $r \binom{n + p}{p}$ squared Schmidt coefficients of the full output must sum to at least $1 - \varepsilon$.
The eigenvalues of $\rho_\phi$ are the squared Schmidt coefficients of the joint output.
For fixed R\'enyi order $s > 1$, a large entropy prevents a small number of Schmidt coefficients from carrying a fixed positive weight.
Comparing the Schmidt rank required to retain this weight with the contribution of the Gaussian terms gives
\begin{equation}
  \log \chiG (\ket{\phi})
  \ge H_s (\rho_\phi) - \log \binom{n + p}{p} - O_{s, \varepsilon} (1),
  \label{equation:bosonic-certificate-main}
\end{equation}
where $H_s (\rho) \coloneqq \log \Tr (\rho^s) / (1 - s)$ is the R\'enyi entropy.
For error $0 \le \delta < 1$, the same bound holds for the approximate border Gaussian rank $\abchiG (\ket{\phi})$, with a bounded correction depending also on $\delta$ (Appendix~\ref{appendix:certificate}).

Uniform control of the photon tails allows products of successful outputs to be truncated while retaining asymptotically all their weight.
Since the environment entropy adds across outputs, it can force the required Schmidt rank to grow exponentially faster than a single Gaussian term can supply.
\eqref{equation:bosonic-certificate-main} then requires exponentially many Gaussian terms (Proposition~\ref{prop:finite-copy-certificate}).
Gaussian postselection and a suitable coupling ensure both the tail control and this separation uniformly, giving a common base $B_\psi > 1$ even for products of different successful outputs (Proposition~\ref{prop:local-core}).

Finally, we transfer this exponential bound to the original tensor powers.
Completeness and typicality select a record with a fraction $p_\psi + o (1)$ of successful outputs and controlled approximation error.
Since Gaussian postselection does not increase border rank, the lower bound transfers to the original tensor powers with rate $p_\psi \log B_\psi$ (Theorem~\ref{thm:bosonic-exponential-transfer}).
An explicit coupling gives the single-photon base $11 / 10$.

\section{Discussion and outlook}

In this work, we established robust exponential lower bounds on suitably defined Gaussian ranks for both fermionic and bosonic systems, showing that non-Gaussianity alone forces Gaussian decomposition complexity to grow exponentially with the number of copies.
This resolves the Gaussian analogue of the longstanding exponential lower bound problem for free rank decomposition. By revealing fundamental limitations of decomposition-based classical simulation, our results offer new insights into the classical–quantum computational boundary. 

Our results suggest that the difficulty of proving stabilizer rank lower bounds may partly reflect the different correlation structures available for free in the qubit stabilizer setting. Stabilizer states need not satisfy the Gaussian constraints underpinning our proofs, so some key techniques developed here do not extend directly to stabilizer rank.

Several interesting questions are worth further studying. For $\ket{M}^{\otimes k}$ and $\ket{1}^{\otimes k}$, determining the optimal exponential growth rates of approximate border Gaussian rank requires matching upper and lower bounds and understanding whether these rates are independent of approximation error.
An important direction is to develop structural certificates analogous to our Majorana matrices or environment spectra for proving lower bounds on stabilizer rank for tensor powers of standard qubit magic states.
More generally, it remains to identify additional structural properties of a family of free states that ensure exponential decomposition complexity for tensor powers of every fixed pure state outside the family.
Finally, tensor powers exhibit exponential Gaussian rank without correlations between copies, this motivates distinguishing the accumulation of independent local resources from collective non-Gaussianity that cannot be removed by local operations, and investigating whether the latter provides robust signatures of quantum phases, critical phenomena, or many-body dynamics.




\begin{acknowledgments}
  F.W.\ and Z.-W.L.\ are supported in part by NSFC under Grant No.~12475023, Dushi Program, and a startup funding from YMSC.
  Z.L. acknowledges support from the Beijing Natural Science Foundation (Grant No. Z220002).
  F.W.\ acknowledges support from the Shuimu Tsinghua Scholar Program.
\end{acknowledgments}


%

\onecolumngrid

\appendix

\section{Gaussian states and ranks}

Mode numbers, copy numbers, occupation numbers, and monomial degrees are integers throughout.
We write $[k] = \{1, \ldots, k\}$ for $k \ge 1$ and $[0] = \varnothing$.
The identity operator is $I$, and $\ket{0^n}$ denotes the vacuum on $n$ modes.
The norms $\norm{\cdot }_2$, $\norm{\cdot }$, and $\norm{\cdot }_1$ denote the vector norm induced by the inner product, the operator norm, and the trace norm, respectively.
For a complex matrix $A$, its Frobenius norm satisfies $\norm{A}_{\mathrm F}^2 = \Tr (A^\dagger A)$.
A density operator is a positive trace-class operator of trace one.
For a finite or countable probability vector $q$, define
\begin{equation}
  H (q) \coloneqq -\sum_j q_j \log q_j, \quad 0 \log 0 \coloneqq 0.
  \label{eq:shannon-entropy}
\end{equation}
Binary entropy is $h (x) \coloneqq H (x, 1 - x)$, for $0 \le x \le 1$.
The von Neumann entropy $H (\rho)$ is the Shannon entropy of the eigenvalues of $\rho$.
Its R\'enyi entropy of order $s > 1$ is
\begin{equation}
  H_s (\rho) \coloneqq \frac{\log \Tr \rho^s}{1 - s} = -\frac{s}{s - 1} \log \norm{\rho}_s, \quad \norm{\rho}_s \coloneqq (\Tr \rho^s)^{1 / s}.
  \label{eq:bosonic-renyi-entropy}
\end{equation}
For probability vectors, set $H_s (q) \coloneqq H_s (\diag q)$.

\subsection{Fermionic Gaussian states}
\label{appendix:conventions}

The creation and annihilation operators of $n$ fermionic modes satisfy
\begin{equation}
  \{a_i, a_j^\dagger\} = \delta_{i j}
  I, \quad \{a_i, a_j\} = \{a_i^\dagger, a_j^\dagger\} = 0.
\end{equation}
The occupation-number basis of the Fock space $\mc H_n$ is
\begin{equation}
  \ket{x} \coloneqq (a_1^\dagger)^{x_1}\cdots(a_n^\dagger)^{x_n} \ket{0^n},
  \quad x \in \{0, 1\}^n.
\end{equation}
The order of the creation operators fixes the signs; in tensor products, modes are ordered by tensor factor.
Writing $\abs{x} = \sum_j x_j$, the parity sectors are
\begin{equation}
  \mc H_n = \mc H_n^+ \oplus \mc H_n^-, \quad
  \mc H_n^\pm = \spn \{\ket{x} : (-1)^{\abs{x}} = \pm 1\}.
\end{equation}
The parity superselection rule permits only pure states of definite parity.

We use the Majorana operators
\begin{equation}
  \gamma_{2 j - 1} = a_j + a_j^\dagger, \quad
  \gamma_{2 j} = i (a_j^\dagger - a_j),
\end{equation}
so that $\gamma_\mu^\dagger = \gamma_\mu$ and $\{\gamma_\mu, \gamma_\nu\} = 2 \delta_{\mu \nu}
I$~\cite{bravyi2004lagrangianrepresentationfermioniclinear}.
Up to an overall phase, a fermionic Gaussian unitary has the form
\begin{equation}
  U_A = \exp \mleft(\frac14 \sum_{\mu, \nu = 1}^{2 n} A_{\mu \nu} \gamma_\mu \gamma_\nu\mright), \quad A^{\mathsf T} = -A, \quad A \in \R^{2 n \times 2 n}.
\end{equation}
These unitaries preserve parity.
A pure fermionic Gaussian state is obtained, up to phase, by applying $U_A$ to the vacuum in the even sector or to a one-particle reference state in the odd sector~\cite{bravyi2004lagrangianrepresentationfermioniclinear,hackl2024bosonicfermionicgaussianstates}.
For $n \le 3$, every pure fermionic state is Gaussian.
Four modes are the first nontrivial case.
Appendix~\ref{appendix:theorem-two-proof} reduces general pure non-Gaussian states to this case by Gaussian postselection.

\subsection{Bosonic Gaussian states}
\label{appendix:statement}

\paragraph{Fock space and photon number.}
The creation and annihilation operators of $n$ bosonic modes satisfy
\begin{equation}
  [a_i, a_j^\dagger] = \delta_{i j}
  I, \quad [a_i, a_j] = [a_i^\dagger, a_j^\dagger] = 0.
\end{equation}
The Fock space $\F_n$ has the orthonormal basis
\begin{equation}
  \ket{x} \coloneqq \prod_{j = 1}^n \frac{(a_j^\dagger)^{x_j}}{\sqrt{x_j!}} \ket{0^n},
  \quad x \in \Z_{\ge 0}^n,
\end{equation}
where $a_j \ket{0^n} = 0$.
For $j \in [n]$, let $e_j \in \Z^n$ be the standard coordinate vector.
Then
\begin{equation}
  a_j \ket{x} = \sqrt{x_j} \ket{x - e_j}, \quad
  a_j^\dagger\ket{x} = \sqrt{x_j + 1} \ket{x + e_j},
\end{equation}
with the first expression equal to zero when $x_j = 0$.
We call each excitation a photon and write
\begin{equation}
  N_{\mathrm{tot}} \coloneqq \sum_{j = 1}^n a_j^\dagger a_j, \quad N_{\mathrm{tot}} \ket{x} = \abs{x} \ket{x}, \quad \abs{x} \coloneqq \sum_j x_j.
\end{equation}
To control second photon-number moments, we use
\begin{equation}
  \norm{\ket{u}}_D^2 \coloneqq \sum_{x \in \Z_{\ge 0}^n} (1 + \abs{x})^2 \abs{u_x}^2,
  \quad \ket{u} = \sum_x u_x \ket{x},
  \label{eq:photon-moment-norm}
\end{equation}
allowing the value $ + \infty$.
This norm is finite precisely when the second moment is finite, and then equals $\norm{(N_{\mathrm{tot}} + I) \ket{u}}_2$.
For zero modes, set $\F_0 = \C$.

\paragraph{Gaussian states and coherent states.}
A bosonic Gaussian unitary is a finite product of operators $e^{-i t H}$, where $t \in \R$ and $H$ is a symmetric polynomial of degree at most two in the ladder operators.
For finitely many modes, such polynomials are essentially self-adjoint on finite occupation sums; their self-adjoint closures define the exponentials.
Linear and constant terms include displacements and global phases.
A normalized pure state is Gaussian if it is obtained from the vacuum by a Gaussian unitary~\cite{hackl2024bosonicfermionicgaussianstates}.
We call a nonzero scalar multiple of such a state a \emph{pure Gaussian vector} and write $\G_n$ for the set of these vectors; in particular, $\G_0 = \C\setminus\{0\}$.

For $\bm\alpha \in \C^n$, the displacement and coherent state are
\begin{equation}
  D (\bm\alpha) \coloneqq \exp \mleft[ \sum_{j = 1}^n (\alpha_j a_j^\dagger - \overline{\alpha_j} a_j) \mright], \quad \ket{\bm\alpha} \coloneqq D (\bm\alpha) \ket{0^n}.
\end{equation}
Normal ordering gives
\begin{equation}
  \ket{\bm\alpha} = e^{-\abs{\bm\alpha}^2 / 2}
  \sum_{x \in \Z_{\ge 0}^n} \frac{\bm\alpha^x}{\sqrt{x!}} \ket{x},
  \label{eq:coherent-occupation-expansion}
\end{equation}
For $z \in \C^n$ and $x \in \Z_{\ge 0}^n$, we use $z^x \coloneqq \prod_j z_j^{x_j}$, $x! \coloneqq \prod_j x_j!$, and the Euclidean norm $\abs{z}$.
The measure $\odif[order={2 n}]{z}$ is Lebesgue measure on $\C^n \cong \R^{2 n}$.

The coherent states resolve the identity:
\begin{equation}
  \int_{\C^n} \ket{\bm\alpha} \bra{\bm\alpha}\,
  \frac{\odif[order={2n}]{\bm\alpha}}{\pi^n} = I_{\F_n},
  \label{eq:coherent-resolution}
\end{equation}
which follows by Gaussian integration of the occupation-basis matrix elements in Eq.~\eqref{eq:coherent-occupation-expansion}.

\paragraph{Bargmann representation.}
For every $\ket{\psi} = \sum_x \psi_x \ket{x} \in \F_n$, its \emph{Bargmann function} $F_\psi : \C^n \to \C$ is
\begin{equation}
  F_\psi (z) \coloneqq \sum_{x \in \Z_{\ge 0}^n} \psi_x \frac{z^x}{\sqrt{x!}} = e^{\abs{z}^2 / 2} \braket{\overline z}{\psi}.
  \label{eq:bargmann-function}
\end{equation}
The overlap identity follows from Eq.~\eqref{eq:coherent-occupation-expansion}.
Cauchy--Schwarz gives
\begin{equation}
  \sum_x \abs{\psi_x} \frac{\abs{z^x}}{\sqrt{x!}}
  \le \norm{\ket{\psi}}_2
  \mleft(\sum_x \frac{\abs{z^x}^2}{x!}\mright)^{1 / 2}
  = \norm{\ket{\psi}}_2 e^{\abs{z}^2 / 2}.
  \label{eq:bargmann-growth}
\end{equation}
Applying the same estimate to the tails of the series gives uniform convergence on compact sets, so $F_\psi$ is entire.
The Bargmann representation preserves the norm:
\begin{equation}
  \int_{\C^n} \abs{F_\psi (z)}^2 e^{-\abs{z}^2}
  \frac{\odif[order={2 n}]{z}}{\pi^n}
  = \norm{\ket{\psi}}_2^2.
  \label{eq:bargmann-norm}
\end{equation}
The norm identity follows by Gaussian integration: the monomials $z^x / \sqrt{x!}$ are orthonormal.

The ladder operators act on Bargmann functions as follows~\cite{chabaud2022holomorphicrepresentationquantumcomputations}:
\begin{equation}
  \ket{x} \longleftrightarrow \frac{z^{x}}{\sqrt{x!}},
  \quad
  a_j \longleftrightarrow \partial_{z_j},
  \quad
  a_j^\dagger \longleftrightarrow z_j,
  \quad
  N_{\mathrm{tot}} \longleftrightarrow \sum_{j = 1}^n z_j \partial_{z_j}.
  \label{eq:bargmann-dictionary}
\end{equation}
Pure Gaussian vectors correspond exactly to
\begin{equation}
  F (z) = C_0 \exp \mleft(\frac12 z^{\mathsf T} A z + \bm b^{\mathsf T} z\mright), \quad C_0 \neq 0, \quad A = A^{\mathsf T}, \quad \norm{A} < 1.
  \label{eq:pure-gaussian-bargmann}
\end{equation}
Here $z \in \C^n$, $C_0 \in \C \setminus \{0\}$, $A \in \C^{n \times n}$, and $\bm b \in \C^n$; the condition $\norm{A} < 1$ ensures normalizability.

To verify this characterization, use a unitary change of modes to put the complex symmetric matrix $A$ in Takagi form, with diagonal entries $d_j \ge 0$.
The Gaussian integral in Eq.~\eqref{eq:bargmann-norm} is finite precisely when every $d_j < 1$: its real quadratic coefficients are $ - (1 - d_j)$ and $ - (1 + d_j)$.
In these coordinates, differentiating the exponential gives $(a_j - d_j a_j^\dagger - b_j) \ket{\psi} = 0$.
Displacements satisfying $\alpha_j - d_j\overline{\alpha_j} = b_j$ and squeezes with $\tanh r_j = d_j$ reduce these equations to the vacuum equations.
Conversely, a Gaussian unitary transforms the vacuum annihilation equations into affine linear equations in the ladder operators, whose normalized solution has the form in Eq.~\eqref{eq:pure-gaussian-bargmann}.

\subsection{Gaussian ranks}
\label{appendix:gaussian-rank-definitions}

The following definitions apply to both fermionic and bosonic vectors.
In either case, a pure Gaussian vector is a nonzero scalar multiple of a normalized pure Gaussian state.
Gaussian rank measures the number of terms in a Gaussian decomposition; in the bosonic setting, it was introduced as a simulation resource in Ref.~\cite{hahn2026classicalsimulationquantumresource}.

\begin{definition}[Gaussian rank]
  For $\ket{\psi} \ne 0$, define
  \begin{equation}
    \chiG (\ket{\psi}) \coloneqq \inf \set*{r \in \Z_{> 0} : \ket{\psi} = \sum_{j = 1}^r \ket{G_j}, \quad \ket{G_j} \text{ pure Gaussian vectors}},
    \label{eq:gaussian-rank-definition}
  \end{equation}
  with value $ + \infty$ if no finite decomposition exists, and set $\chiG (0) = 0$.
\end{definition}

\begin{definition}[Border Gaussian rank]
  For $\ket{\psi} \ne 0$, define
  \begin{equation}
    \bchiG (\ket{\psi}) \coloneqq \inf \set*{r \in \Z_{> 0} : \text{there are } \ket{\psi_h} \to \ket{\psi} \text{ with } \chiG (\ket{\psi_h}) \le r \text{ for all } h},
    \label{eq:border-gaussian-rank-definition}
  \end{equation}
  where convergence is in $\norm{\cdot }_2$.
  Set $\bchiG (0) = 0$, and take the infimum of the empty set to be $ + \infty$.
\end{definition}

For every integer $r \ge 0$, the vectors of border Gaussian rank at most $r$ form the norm closure of those of Gaussian rank at most $r$.
Both ranks are unchanged by multiplication by a nonzero scalar.
In a sequence realizing border rank, individual Gaussian summands may diverge while their sum converges.
The examples in Appendix~\ref{appendix:gaussian-border-gap} show that allowing such limits can lower the rank.

Fix a normalized target $\ket{\psi}$ and $0 \le \delta < 1$.
Approximants below lie in the same Fock space as the target and need not be normalized.

\begin{definition}[Approximate Gaussian rank]
  The $\delta$-approximate Gaussian rank is
  \begin{equation}
    \achiG (\ket{\psi}) \coloneqq \min_{\norm{\ket{\psi} - \ket{\eta}}_2 \le \delta} \chiG (\ket{\eta}).
    \label{eq:approx-gaussian-rank}
  \end{equation}
\end{definition}

\begin{definition}[Approximate border Gaussian rank]
  The $\delta$-approximate border Gaussian rank is
  \begin{equation}
    \abchiG (\ket{\psi}) \coloneqq \min_{\norm{\ket{\psi} - \ket{\eta}}_2 \le \delta} \bchiG (\ket{\eta}).
    \label{eq:approx-border-rank}
  \end{equation}
\end{definition}

The restriction $\delta < 1$ excludes the zero approximant.
At $\delta = 0$, these definitions reduce to Gaussian rank and border Gaussian rank, respectively.
For $\delta > 0$, both minima are finite: fermionic occupation states span the finite-dimensional Fock space, and finite coherent-state superpositions are dense in bosonic Fock space by Eq.~\eqref{eq:coherent-resolution}.
The minima are attained because the set of attainable finite ranks is a nonempty set of nonnegative integers.
For $0 < \delta < 1$,
\begin{equation}
  \abchiG (\ket{\psi}) \le \achiG (\ket{\psi}) \le \bchiG (\ket{\psi}) \le \chiG (\ket{\psi}).
  \label{eq:gaussian-rank-comparison}
\end{equation}
For the middle inequality, when the border rank is finite, a sequence realizing it eventually enters the approximation ball; when it is infinite, the inequality is immediate.
Hence a lower bound on approximate border Gaussian rank bounds all four quantities.
For a fermionic target, all four definitions may be restricted to its parity sector: parity projection removes opposite-parity Gaussian summands, preserves norm convergence, and cannot increase the approximation error.

\subsection{Bosonic Gaussian measurements and photon cutoffs}
\label{appendix:bosonic-measurements-and-cutoffs}

\paragraph{Gaussian operators and instruments.}
The bosonic Gaussian operators used below are compositions of Gaussian unitaries, vacuum insertions, and coherent-state contractions.
Each is bounded, and their tensor products are Gaussian.
Input and output mode numbers may differ, including zero-mode outputs in $\F_0 = \C$.
A bounded Gaussian operator $K : \F_n \to \F_{n'}$ preserves pure Gaussian vectors, with zero allowed as an output, even in the presence of unmeasured modes:
\begin{equation}
  (K \otimes I_{\F_\ell})\G_{n + \ell} \subseteq \G_{n' + \ell} \cup \{0\}, \quad \ell \ge 0.
  \label{eq:complete-gaussian-preservation}
\end{equation}
This includes Gaussian vectors entangled across several copies of the input.
A Gaussian isometry appends vacuum modes and applies a Gaussian unitary.

A Gaussian instrument is specified here by a weakly measurable family of bounded Gaussian Kraus operators $K_\omega : \F_n \to \F_{n_\omega}$, indexed by a measurable outcome space $\Omega$ with measure $\nu$.
The input mode number $n$ is fixed; the output mode number $n_\omega \ge 0$ may depend on the outcome.
It is complete when
\begin{equation}
  \int_\Omega K_\omega^\dagger K_\omega\,\odif{\nu} (\omega) = I
  \label{eq:gaussian-instrument-completeness}
\end{equation}
in the weak operator sense.
For a normalized input $\ket{\psi}$, the unnormalized output at $\omega$ is $K_\omega \ket{\psi}$, and its outcome law $\mb P$ satisfies
\begin{equation}
  \mb P (A) = \int_A \norm{K_\omega \ket{\psi}}_2^2\,\odif{\nu} (\omega)
\end{equation}
for every measurable set $A \subseteq \Omega$.
We also call $(K_\omega)_{\omega \in \Omega}$ a Kraus density relative to $\nu$.
For heterodyne detection, the outcome is $\bm\alpha \in \C^n$, the Kraus operator is $\bra{\bm\alpha}$, and the measure is $\pi^{-n}\odif[order={2n}]{\bm\alpha}$.
The coherent-state resolution in Eq.~\eqref{eq:coherent-resolution} proves completeness.

\paragraph{Beam splitters and two-mode squeezers.}
A beam splitter and a two-mode squeezer coupling a system mode $a$ to environment modes $b_1, b_2$ are
\begin{equation}
  B_\theta = \exp[\theta(b_1^\dagger a-a^\dagger b_1)], \quad
  S_r = \exp[r(a^\dagger b_2^\dagger-a b_2)],
  \label{eq:beam-splitter-and-squeezer}
\end{equation}
with $\sin^2 \theta = \sinh^2 r = t$, $0 < t < 1$, $0 < \theta < \pi / 2$, and $r > 0$.
The number-basis formulas are
\begin{align}
  B_\theta \ket{m, 0}
  & = \sum_{j = 0}^m \sqrt{\binom m j} (1 - t)^{j / 2} t^{(m - j) / 2} \ket{j, m - j}, \label{eq:beam-splitter-expansion} \\
  S_r \ket{j, 0}
  & = \sum_{h \ge 0} \sqrt{\binom{j + h}{j}}
  \frac{t^{h / 2}}{(1 + t)^{(j + h + 1) / 2}} \ket{j + h, h}.
  \label{eq:two-mode-squeezer-expansion}
\end{align}
Eq.~\eqref{eq:beam-splitter-expansion} follows by expanding the transformed creation operator with the binomial theorem.
For Eq.~\eqref{eq:two-mode-squeezer-expansion}, solve the vacuum annihilation equations at $j = 0$, then apply $(S_r a^\dagger S_r^\dagger)^j / \sqrt{j!}$, using $S_r a^\dagger S_r^\dagger = a^\dagger\cosh r - b_2\sinh r$.
The generator has real matrix elements in the occupation basis; continuity from $r = 0$ then fixes the phase.

\paragraph{Environment states and Schmidt rank.}
For $n \ge 1$ and $\ell \ge 0$, consider a joint vector $\ket{\Psi} \in \F_n \otimes \F_\ell$.
The first factor is the system $S$ and the second is the environment $E$.
Its Schmidt rank is the least number of product terms in an expansion
\begin{equation}
  \ket{\Psi} = \sum_{j = 1}^R \ket{u_j}_S \otimes \ket{v_j}_E.
  \label{eq:schmidt-rank-definition}
\end{equation}
We write this rank as $\SR_{S : E} (\ket{\Psi})$, with value $ + \infty$ if no finite expansion exists and $\SR_{S : E} (0) = 0$.
For a fixed Gaussian isometry $V : \F_n \to \F_n \otimes \F_\ell$ and a normalized input $\ket{\phi}$, write
\begin{equation}
  \rho_\phi \coloneqq \Tr_S (V \ket{\phi} \bra{\phi}
  V^\dagger).
  \label{eq:environment-state-definition}
\end{equation}
This density operator acts on $\F_\ell$.

\paragraph{Photon cutoffs and moments.}
Expand a joint vector in environment occupations:
\begin{equation}
  \ket{\Psi} = \sum_{x \in \Z_{\ge 0}^\ell} \ket{\Psi_x} \otimes \ket{x}, \quad
  \ket{\Psi_x} \coloneqq (I \otimes \bra{x}) \ket{\Psi}.
  \label{eq:environment-coefficient-vectors}
\end{equation}
The environment photon-number operator satisfies
\begin{equation}
  N_E \ket{x} = \abs{x} \ket{x}, \quad \abs{x} = \sum_{j = 1}^\ell x_j.
  \label{eq:environment-number-operator}
\end{equation}
Its projection onto at most $p$ photons acts by
\begin{equation}
  \Pi_{\le p}^E \ket{\Psi} = \sum_{\abs{x} \le p} \ket{\Psi_x} \otimes \ket{x}, \quad p \in \Z_{\ge 0}.
  \label{eq:photon-cutoff-action}
\end{equation}
On joint vectors, the system identity in $I_S \otimes \Pi_{\le p}^E$ is implicit.
For normalized $\ket{\Psi}$, orthogonality gives the discarded probability
\begin{equation}
  \norm{(I - \Pi_{\le p}^E) \ket{\Psi}}_2^2
  = \sum_{\abs{x} > p} \norm{\ket{\Psi_x}}_2^2
  = \Tr \mleft[(I-\Pi_{\le p}^E)\Tr_S\ket{\Psi}\bra{\Psi}\mright].
  \label{eq:photon-cutoff-probability}
\end{equation}
For a density operator $\rho$ on the environment with finite second moment, $\Var (N_E;\rho) \coloneqq \Tr (N_E^2 \rho) - (\Tr (N_E \rho))^2$.

\paragraph{Thermal entropy.}
A thermal mode of mean photon number $y \ge 0$ has probabilities $y^j / (1 + y)^{j + 1}$, for $j \ge 0$, with the vacuum understood at $y = 0$.
For $n \ge 1$, summing the entropies of $n$ independent thermal modes, each of mean photon number $x / n$, gives
\begin{equation}
  d_n (x) \coloneqq (n + x) \log (n + x) - n \log n - x \log x, \quad x > 0, \quad d_n (0) = 0.
  \label{eq:dn-definition}
\end{equation}
In particular, $d_n (x) = n d_1 (x / n)$.
This function also gives the exponential growth rate of the photon-cutoff dimension in Proposition~\ref{prop:finite-copy-certificate}.

\subsection{Separations between Gaussian rank and border Gaussian rank}
\label{appendix:gaussian-border-gap}
\label{appendix:one-photon-separation}

Group the six fermionic modes into the pairs $(1, 2)$, $(3, 4)$, and $(5, 6)$.

\begin{proposition}[Gaussian rank and border Gaussian rank of $\ket{\Psi_6}$]
  \label{prop:six-mode-rank-gap}
  The even-parity state
  \begin{equation}
    \ket{\Psi_6} \coloneqq \frac{\ket{110000} + \ket{001100} + \ket{000011}}{\sqrt3}
    \label{eq:six-mode-gap-state}
  \end{equation}
  has border Gaussian rank two and exact Gaussian rank three:
  \begin{equation}
    \bchiG (\ket{\Psi_6}) = 2 < 3 = \chiG (\ket{\Psi_6}).
  \end{equation}
\end{proposition}

\begin{proof}
  For $\epsilon > 0$, define the unnormalized vector
  \begin{equation}
    \ket{\phi_\epsilon} \coloneqq \frac{\bigotimes_{b = 1}^3 \mleft(\ket{00} + \epsilon \ket{11}\mright)_{2 b - 1, 2 b} - \ket{0^6}}{\sqrt3\,\epsilon}.
  \end{equation}
  Both terms in the numerator are Gaussian: the first is a tensor product of two-mode Gaussian vectors, and the second is the vacuum.
  Consequently, $\chiG (\ket{\phi_\epsilon}) \le 2$ for every $\epsilon > 0$.
  Expanding gives
  \begin{equation}
    \ket{\phi_\epsilon} = \ket{\Psi_6} + \frac{\epsilon}{\sqrt3} \mleft(\ket{111100} + \ket{110011} + \ket{001111}\mright) + \frac{\epsilon^2}{\sqrt3} \ket{111111},
  \end{equation}
  so $\ket{\phi_\epsilon} \to \ket{\Psi_6}$ as $\epsilon \to 0$.
  Therefore $\bchiG (\ket{\Psi_6}) \le 2$.

  At fixed mode number, normalized pure fermionic Gaussian states form a finite union of compact Gaussian-unitary orbits.
  Their scalar cone together with zero is therefore norm closed, so a nonzero vector has border Gaussian rank one if and only if it is Gaussian.

  We use one elementary consequence of the Pfaffian relations for pure Gaussian states.
  For an even pure Gaussian vector $s$, let $s_0$ be its vacuum amplitude, let $Q_s$ be its antisymmetric two-particle amplitude matrix, and let $s_I$ be its four-particle amplitude on the mode set $I$.
  Then
  \begin{equation}
    \Pf \mleft((Q_s)_I\mright) = s_0 s_I,
    \quad \abs{I} = 4.
    \label{eq:six-mode-pfaffian-relations}
  \end{equation}
  For $s_0 \ne 0$, this follows from the Bardeen--Cooper--Schrieffer (BCS) representation~\cite{hackl2024bosonicfermionicgaussianstates}; since vectors with nonzero vacuum amplitude are dense in the even pure Gaussian cone and both sides are polynomial in the amplitudes, the identity extends to $s_0 = 0$.
  In particular, $s_0 = 0$ implies $\rank (Q_s) \le 2$, because all $4 \times 4$ principal Pfaffians vanish.

  Suppressing the normalization, define
  \begin{equation}
    D \coloneqq a_1^\dagger a_2^\dagger + a_3^\dagger a_4^\dagger + a_5^\dagger a_6^\dagger,
  \end{equation}
  so the target is $D \ket{0^6}$.
  Its two-particle amplitude matrix is
  \begin{equation}
    M_6 \coloneqq J \oplus J \oplus J,
    \quad
    J \coloneqq
    \begin{pmatrix}
      0   & 1 \\
      - 1 & 0
    \end{pmatrix},
    \label{eq:six-mode-target-matrix}
  \end{equation}
  which has rank six.
  Since the target has zero vacuum amplitude, Eq.~\eqref{eq:six-mode-pfaffian-relations} shows that it is not Gaussian.
  Closedness of the Gaussian cone together with zero therefore rules out border Gaussian rank one, and hence $\bchiG (\ket{\Psi_6}) = 2$.

  Each of the three Fock-basis states in Eq.~\eqref{eq:six-mode-gap-state} is a Slater determinant, giving $\chiG (\ket{\Psi_6}) \le 3$.
  It remains to rule out a sum of two pure Gaussian vectors.
  Suppose for contradiction that
  \begin{equation}
    D \ket{0^6} = g + h,
  \end{equation}
  where scalar coefficients have been absorbed into two nonzero pure Gaussian vectors.
  Parity forces both summands to be even.
  Let $g_0$, $h_0$ be their vacuum amplitudes and $Q_g$, $Q_h$ their two-particle amplitude matrices.
  Comparison of the vacuum and two-particle components gives
  \begin{equation}
    h_0 = -g_0,
    \quad
    Q_g + Q_h = M_6.
  \end{equation}

  First suppose $g_0 = h_0 = 0$.
  Equation~\eqref{eq:six-mode-pfaffian-relations} gives $\rank (Q_g)$, $\rank (Q_h) \le 2$.
  The relation $Q_g + Q_h = M_6$ would then imply
  \begin{equation}
    6 = \rank (M_6)
    \le \rank (Q_g) + \rank (Q_h)
    \le 4,
  \end{equation}
  a contradiction.

  It remains to consider $g_0 = \alpha$ and $h_0 = -\alpha$ with $\alpha \ne 0$.
  By the same BCS representation, there are quadratic pair-creation operators $X_1$, $X_2$ such that
  \begin{equation}
    g = \alpha e^{X_1} \ket{0^6},
    \quad
    h = -\alpha e^{X_2} \ket{0^6}.
  \end{equation}
  On six modes, the expansion terminates at
  \begin{equation}
    e^{X_\ell} \ket{0^6}
    = \mleft(1 + X_\ell + \frac12 X_\ell^2 + \frac16 X_\ell^3\mright) \ket{0^6}.
  \end{equation}
  Disjoint quadratic creation monomials commute after four anticommutations, while products with a repeated creation operator vanish; hence $X_1 X_2 = X_2 X_1$.
  Comparing the two-particle components gives
  \begin{equation}
    \alpha (X_1 - X_2) \ket{0^6} = D \ket{0^6}.
  \end{equation}
  Since the two-particle Fock basis is linearly independent,
  \begin{equation}
    R \coloneqq X_1 - X_2 = \lambda D, \quad \lambda \coloneqq \alpha^{-1} \ne 0.
  \end{equation}
  The target has no four-particle component, and hence
  \begin{equation}
    0 = (X_1^2 - X_2^2) \ket{0^6} = R S \ket{0^6},
    \quad
    S \coloneqq X_1 + X_2.
  \end{equation}
  The four-particle Fock basis is also linearly independent, so $R S = 0$ as a quartic creation operator.
  Finally, the absence of a six-particle component would require
  \begin{equation}
    0 = (X_1^3 - X_2^3) \ket{0^6}.
  \end{equation}
  Because $R$ and $S$ commute and $R S = 0$,
  \begin{equation}
    X_1^3 - X_2^3 = \frac14 \mleft(3 S^2 R + R^3\mright) = \frac14 R^3 = \frac{\lambda^3}{4} D^3.
  \end{equation}
  However,
  \begin{equation}
    D^3 \ket{0^6}
    = 6 a_1^\dagger a_2^\dagger a_3^\dagger a_4^\dagger a_5^\dagger a_6^\dagger\ket{0^6}
    = 6 \ket{1^6} \ne 0,
  \end{equation}
  which is a contradiction.
  Thus no exact two-Gaussian decomposition exists, and $\chiG (\ket{\Psi_6}) = 3$.
\end{proof}

\begin{proposition}[Gaussian rank and border Gaussian rank of $\ket{1}$]
  \label{prop:one-photon-rank-gap}
  For the one-photon state,
  \begin{equation}
    \bchiG (\ket{1}) = 2 < \chiG (\ket{1}) = +\infty.
    \label{eq:one-photon-rank-gap}
  \end{equation}
\end{proposition}

\begin{proof}
  As real $\varepsilon \to 0$, the coherent-state expansion in Eq.~\eqref{eq:coherent-occupation-expansion} gives
  \begin{equation}
    \frac{\ket{\varepsilon} - \ket{0}}{\varepsilon}
    = \ket{1} + O (\varepsilon)
    \longrightarrow \ket{1}.
    \label{eq:one-photon-coherent-border}
  \end{equation}
  Thus $\bchiG (\ket{1}) \le 2$.

  Gaussian density operators form a trace-norm-closed set~\cite[Appendix A, Lemma 1]{lami2018gaussianquantumresourcetheories}.
  Hence $\G_1 \cup \set{0}$ is closed in $\norm{\cdot }_2$.
  Since $\ket1$ is non-Gaussian, $\bchiG (\ket{1}) > 1$.

  The exact Gaussian rank is $\chiG (\ket{1}) = +\infty$~\cite[Appendix D.1]{hahn2026classicalsimulationquantumresource}.
\end{proof}

\section{Fermionic proofs}
\label{appendix:fermionic-proofs}

\subsection{Gaussian rank lower bounds for \texorpdfstring{\IfFontSeriesContextTF{bf}{\boldmath}{}$\ket{M}^{\otimes k}$}{the k-fold tensor power of ket M}}
\label{appendix:formal-main}

The finite-copy bound is governed by two combinatorial quantities.
For a polynomial $f (z)$, we write $[z^d] f (z)$ for the coefficient of $z^d$.
For integers $n \ge 1$ and $0 \le p \le n$, set
\begin{equation}
  C_{n, p} \coloneqq \sum_{q = 0}^{\floor{p / 2}} \binom{n}{p - 2 q}.
  \label{eq:rho-main}
\end{equation}
For integers $k \ge 1$ and $0 \le p \le 2 k$, set
\begin{equation}
  D_{k, p} \coloneqq [z^p] (1 + 8 z + 7 z^2)^k.
  \label{eq:D-shell-app}
\end{equation}

\begin{theorem}[Gaussian rank lower bounds for $\ket{M}^{\otimes k}$]
  \label{thm:main_app}
  For every $k \ge 1$, $0 \le p \le 2 k$, and $0 \le \delta < 1$,
  \begin{equation}
    \abchiG (\ket{M}^{\otimes k}) \ge (1 - \delta^2) \frac{D_{k, p}}{C_{4 k, p}}.
    \label{eq:finite-main}
  \end{equation}
  Consequently, there is an absolute constant $c > 0$ such that
  \begin{equation}
    \abchiG (\ket{M}^{\otimes k}) \ge c (1 - \delta^2) \mleft(\frac{343}{243}\mright)^k
    \label{eq:exponential-main}
  \end{equation}
  for every $k \ge 1$ and $0 \le \delta < 1$.
\end{theorem}

For comparison, the computational-basis expansion is
\begin{equation}
  \ket{M}^{\otimes k}
  = 2^{-k / 2} \sum_{x \in \set{0, 1}^k}
  \bigotimes_{b = 1}^k \ket{x_b x_b x_b x_b},
  \label{eq:M-fock-expansion-main}
\end{equation}
whose $2^k$ terms are Gaussian Slater determinants.
Set $c_\delta \coloneqq c (1 - \delta^2)$.
For every $0 < \delta < 1$, Theorem~\ref{thm:main_app}, Eq.~\eqref{eq:gaussian-rank-comparison}, and the occupation-basis expansion in Eq.~\eqref{eq:M-fock-expansion-main} give
\begin{equation}
  c_\delta (343 / 243)^k \le \abchiG (\ket{M}^{\otimes k}) \le \achiG (\ket{M}^{\otimes k}) \le \bchiG (\ket{M}^{\otimes k}) \le \chiG (\ket{M}^{\otimes k}) \le 2^k.
  \label{eq:gap-main}
\end{equation}
At $\delta = 0$, the approximate ranks reduce to their exact counterparts, and the corresponding comparison is
\begin{equation}
  c_0 (343 / 243)^k \le \bchiG (\ket{M}^{\otimes k}) \le \chiG (\ket{M}^{\otimes k}) \le 2^k.
\end{equation}

The proof compares matrices whose columns are obtained by applying degree-$p$ Majorana monomials to a state.
Appendix~\ref{appendix:flattenings} shows that the full matrix has rank $C_{n, p}$ on every pure Gaussian state and rank at most $r C_{n, p}$ on every state of border Gaussian rank $r$.
Appendix~\ref{appendix:short-proof} then selects a restricted matrix with $D_{k, p}$ orthonormal columns on $\ket{M}^{\otimes k}$.
In Appendix~\ref{appendix:spectral}, we use Lemma~\ref{lem:eckart-young-app} with the Gaussian rank bound, orthonormality, and the Frobenius-norm error estimate to obtain the finite-copy lower bound in Eq.~\eqref{eq:finite-main}.
Appendix~\ref{appendix:exponential-estimate} completes the argument by taking $p = \floor{8 k / 9}$ and estimating the resulting ratio, which gives Eq.~\eqref{eq:exponential-main}.

\subsubsection{Matrices generated by Majorana monomials}
\label{appendix:flattenings}

For an index set $J = \set{j_1 < \cdots < j_p}$, subscripts will denote products in increasing order.
Thus $\gamma_J \coloneqq \gamma_{j_1} \cdots \gamma_{j_p}$, with the same convention for other indexed operators.
For $\ket{\psi} \in \mc H_n$ and $0 \le p \le 2 n$, let $F_p (\ket{\psi})$ be the matrix whose columns are indexed by the $p$-element subsets $J \subset [2 n]$, with column $J$ equal to $\gamma_J \ket{\psi}$.
The matrix depends linearly on $\ket{\psi}$ and has $\binom{2 n}{p}$ columns.
For a nonempty collection $\mc L$ of $p$-element subsets of $[2 n]$, retain only the columns indexed by $\mc L$:
\begin{equation}
  F_{\mc L} (\ket{\psi}) e_J \coloneqq \gamma_J \ket{\psi}, \quad J \in \mc L.
  \label{eq:restricted-majorana-action}
\end{equation}
Here $e_J \in \C^{\mc L}$ is the standard coordinate vector indexed by $J$, so $F_{\mc L} (\ket{\psi}) : \C^{\mc L} \to \mc H_n$.
Restricting the columns cannot increase rank, so $\rank F_{\mc L} (\ket{\psi}) \le \rank F_p (\ket{\psi})$.
The space spanned by the degree-$p$ Majorana monomials will also be useful:
\begin{equation}
  \mc A_p \coloneqq \spn \set*{\gamma_J : \abs{J} = p}, \quad \mc A_p \ket{\psi} \coloneqq \spn \set*{A \ket{\psi} : A \in \mc A_p}.
  \label{eq:majorana-operator-space}
\end{equation}
The Majorana monomials are linearly independent, so $\dim \mc A_p = \binom{2 n}{p}$.
By construction,
\begin{equation}
  \im F_p (\ket{\psi}) = \mc A_p \ket{\psi}, \quad \rank F_p (\ket{\psi}) = \dim (\mc A_p \ket{\psi}).
  \label{eq:majorana-action-image}
\end{equation}

Eq.~\eqref{eq:pure-gaussian-majorana-rank} below expresses the dimension of the osculating filtration of the spinor variety in the Fock representation of the Clifford algebra~\cite{landsberg2002projectivegeometryhomogeneousspaces}.
For $p = 2 s$, the projectivization of $\im F_p (\ket{G})$ is the order-$s$ osculating space described explicitly in Proposition~3.24 and Corollary~3.25 of Ref.~\cite{freire2018projectiveaspectsgeometrylagrangian}; the case of odd $p$ follows from the same exterior-algebra model of Clifford multiplication~\cite{manivel2009spinorvarietiessecants,bhoja2023spinorspurespinors}.

\begin{lemma}[Rank of $F_p$ for pure Gaussian states]
  \label{lem:pure-rho-app}
  For every nonzero pure Gaussian state $\ket{G}$ on $n$ modes and every $0 \le p \le n$,
  \begin{equation}
    \rank F_p (\ket{G}) = C_{n, p}.
    \label{eq:pure-gaussian-majorana-rank}
  \end{equation}
\end{lemma}

\begin{proof}
  We begin with an even-parity state $\ket{G} = U \ket{0^n}$, where $U$ is a Gaussian unitary.
  Conjugation by $U$ rotates the Majoranas by a real orthogonal transformation:
  \begin{equation}
    U^\dagger \gamma_i U = \sum_{j = 1}^{2 n} R_{i j} \gamma_j, \quad R \in \mathrm{SO} (2 n).
    \label{eq:majorana-rotation}
  \end{equation}
  For $p$-element subsets $J$, $J' \subset [2 n]$, let $R_{J, J'}$ denote the corresponding $p \times p$ submatrix.
  Expanding the rotated Majoranas and using their anticommutation relations gives
  \begin{equation}
    U^\dagger \gamma_J U = \sum_{\abs{J'} = p} \det (R_{J, J'}) \gamma_{J'}.
    \label{eq:rotated-majorana-product}
  \end{equation}
  Thus conjugation $A \mapsto U^\dagger A U$ preserves $\mc A_p$.
  The inverse conjugation $A \mapsto U A U^\dagger$ also preserves $\mc A_p$, and hence
  \begin{equation}
    U^\dagger \mc A_p U = \mc A_p.
    \label{eq:majorana-space-invariance}
  \end{equation}
  Consequently,
  \begin{equation}
    \mc A_p \ket{G} = \mc A_p U \ket{0^n} = U (U^\dagger \mc A_p U) \ket{0^n} = U \mc A_p \ket{0^n}.
    \label{eq:gaussian-majorana-image}
  \end{equation}
  Since multiplication by $U$ preserves dimension,
  \begin{equation}
    \rank F_p (\ket{G}) = \dim (\mc A_p \ket{G}) = \dim (\mc A_p \ket{0^n}) = \rank F_p (\ket{0^n}).
    \label{eq:gaussian-vacuum-rank}
  \end{equation}

  It remains to evaluate the rank on the vacuum.
  Fix a $p$-element subset $J \subset [2 n]$.
  For each mode $j$, the set $J$ contains zero, one, or both of the indices $2 j - 1$ and $2 j$.
  Let $K$ contain the modes for which exactly one index occurs, and let $Q$ contain those for which both occur.
  It follows that $\abs{K} = p - 2 \abs{Q}$, and $\gamma_J \ket{0^n}$ differs from $a_K^\dagger \ket{0^n}$ only by a phase.
  Thus every column of $F_p (\ket{0^n})$ lies in the sector with $p - 2 q$ particles for some $0 \le q \le \floor{p / 2}$.
  Denoting the $q$-particle sector by $\mc H_{n, q}$, we find
  \begin{equation}
    \im F_p (\ket{0^n}) \subset \bigoplus_{q = 0}^{\floor{p / 2}} \mc H_{n, p - 2 q}.
    \label{eq:vacuum-majorana-image}
  \end{equation}
  For the reverse inclusion, choose a basis state $\ket{K} = a_K^\dagger \ket{0^n}$ with $\abs{K} = p - 2 q$.
  Since $p \le n$, there is a $q$-element subset $Q \subset [n] \setminus K$.
  For each $j \in Q$, the quadratic Majorana operator $Z_j \coloneqq - i \gamma_{2 j - 1} \gamma_{2 j}$ fixes the vacuum.
  The operator $a_K^\dagger Z_Q$ belongs to $\mc A_p$ and maps $\ket{0^n}$ to $\ket{K}$.
  Every basis state on the right-hand side of Eq.~\eqref{eq:vacuum-majorana-image} is therefore attained, proving the reverse inclusion and the rank formula in Eq.~\eqref{eq:pure-gaussian-majorana-rank}.

  The odd-parity case differs only by a reflection.
  Indeed, after a suitable choice of Majorana basis, every odd-parity pure Gaussian state takes the form $U \gamma_1 \ket{0^n}$ for some Gaussian unitary $U$.
  Conjugation by $U \gamma_1$ again acts orthogonally on the Majoranas, with $\gamma_1$ supplying the reflection.
  Hence it preserves $\mc A_p$, and the argument leading to Eq.~\eqref{eq:gaussian-vacuum-rank} applies without change.
\end{proof}

Lemma~\ref{lem:pure-rho-app} immediately gives a useful rank obstruction.
If $\chiG (\ket{\phi}) \le r$, linearity and subadditivity of matrix rank give $\rank F_p (\ket{\phi}) \le r C_{n, p}$.
The same bound holds for border Gaussian rank.
To see this, suppose that $\bchiG (\ket{\phi}) \le r$, and choose a sequence $\ket{\phi_h} \to \ket{\phi}$ with $\chiG (\ket{\phi_h}) \le r$.
Since $F_p (\ket{\phi_h}) \to F_p (\ket{\phi})$ and matrix rank is lower semicontinuous,
\begin{equation}
  \rank F_p (\ket{\phi}) \le \liminf_{h \to \infty} \rank F_p (\ket{\phi_h}) \le r C_{n, p}.
  \label{eq:majorana-rank-cap-app}
\end{equation}

\subsubsection{Orthonormal states generated by Majorana monomials}
\label{appendix:short-proof}

We now select a family of Majorana monomials that produces orthonormal states from a single copy of $\ket{M}$.
On one four-mode block, let
\begin{align}
  \mc C_0 & \coloneqq \set{\varnothing}, \label{eq:short-degree-zero-columns} \\
  \mc C_1 & \coloneqq \set{\set{1}, \set{2}, \set{3}, \set{4}, \set{5}, \set{6}, \set{7}, \set{8}}, \label{eq:short-degree-one-columns} \\
  \mc C_2 & \coloneqq \set{\set{1, 2}, \set{1, 3}, \set{2, 3}, \set{1, 5}, \set{2, 5}, \set{1, 7}, \set{2, 7}}, \label{eq:short-degree-two-columns} \\
  \mc C   & \coloneqq \mc C_0 \cup \mc C_1 \cup \mc C_2,
  \label{eq:short-local-columns-app}
\end{align}
where $\gamma_\varnothing$ is the identity operator.

\begin{lemma}[Orthonormality of the states $\gamma_J \ket{M}$, $J \in \mc C$]
  \label{lem:short-local-frame-app}
  The states $\set*{\gamma_J \ket{M} : J \in \mc C}$ are orthonormal.
\end{lemma}

\begin{proof}
  Consider first the eight states of odd parity.
  Since $\gamma_i^2 = 1$, each $\gamma_i \ket{M}$ has unit norm.
  If $i$ and $j$ belong to different modes, $\gamma_i \gamma_j$ maps both $\ket{0000}$ and $\ket{1111}$ outside their span, and hence $\bra{M} \gamma_i \gamma_j \ket{M} = 0$.
  If they are the two Majoranas of the same mode, their product has opposite expectation values on $\ket{0000}$ and $\ket{1111}$, so the expectation value again vanishes.
  Thus $\bra{M} \gamma_i \gamma_j \ket{M} = \delta_{i j}$, and the odd-parity states are orthonormal.

  For the even-parity states, we compute the action of the quadratic monomials explicitly.
  Acting on $\ket{M}$, the seven quadratic monomials in Eq.~\eqref{eq:short-degree-two-columns} give, up to individual phases, the following state and three pairs of states:
  \begin{equation}
    \frac{\ket{0000} - \ket{1111}}{\sqrt{2}}, \quad \frac{\ket{1100} \pm \ket{0011}}{\sqrt{2}}, \quad \frac{\ket{1010} \pm \ket{0101}}{\sqrt{2}}, \quad \frac{\ket{1001} \pm \ket{0110}}{\sqrt{2}}.
    \label{eq:even-local-majorana-frame}
  \end{equation}
  These seven states are orthogonal to $\ket{M}$ and to one another.
  Finally, states of opposite parity are orthogonal, so the two families together are orthonormal.
\end{proof}

Apply the construction in Lemma~\ref{lem:short-local-frame-app} to each tensor factor.
For a word $\bm J = (J_1, \cdots, J_k) \in \mc C^k$, let $\abs{\bm J}$ denote the sum of the cardinalities of its entries.
We write $\gamma_{\bm J}$ for the corresponding blockwise Majorana product in increasing block order.
Using the same block order, we identify $\bm J$ with the corresponding subset of $[8 k]$.
For $0 \le p \le 2 k$, define the set of words of total degree $p$ in $\mc C^k$ by
\begin{equation}
  \mc L_{k, p} \coloneqq \set*{\bm J \in \mc C^k : \abs{\bm J} = p}.
  \label{eq:short-shell-size-app}
\end{equation}
There is one monomial of degree zero, eight of degree one, and seven of degree two.
The generating polynomial that counts these monomials by degree is therefore $1 + 8 z + 7 z^2$, and Eq.~\eqref{eq:D-shell-app} gives $\abs{\mc L_{k, p}} = D_{k, p}$.
Each state $\gamma_{\bm J} \ket{M}^{\otimes k}$ is, up to an overall fermionic sign, the tensor product of the local states $\gamma_{J_i} \ket{M}$.
These states are orthonormal by Lemma~\ref{lem:short-local-frame-app}, since the signs do not affect their inner products.

\subsubsection{Finite-copy lower bound via matrix approximation}
\label{appendix:spectral}

\begin{lemma}[Eckart--Young theorem~\cite{eckart1936approximationonematrixanotherlowerrank,mirsky1960symmetricgaugefunctionsunitarilyinvariantnorms}]
  \label{lem:eckart-young-app}
  Let $A \in \C^{d_1 \times d_2}$ have singular values $\sigma_1 (A) \ge \cdots \ge \sigma_{\min \set{d_1, d_2}} (A) \ge 0$.
  Let $B$ range over $\C^{d_1 \times d_2}$.
  For every integer $R \ge 0$,
  \begin{equation}
    \sum_{j = R + 1}^{\min \set{d_1, d_2}} \sigma_j (A)^2
    = \min_{\rank B \le R} \norm{A - B}_{\mathrm F}^2.
    \label{eq:eckart-young-app}
  \end{equation}
  The minimum is attained by truncating a singular-value decomposition of $A$ after its first $\min \set{R, d_1, d_2}$ terms.
  In particular, every matrix $B$ with $\rank B \le R$ obeys
  \begin{equation}
    \sum_{j = R + 1}^{\min \set{d_1, d_2}} \sigma_j (A)^2
    \le \norm{A - B}_{\mathrm F}^2.
    \label{eq:eckart-young-bound}
  \end{equation}
\end{lemma}

\begin{proof}[Proof of the first assertion of Theorem~\ref{thm:main_app}]
  Fix $0 \le p \le 2 k$, and let $\ket{\eta}$ satisfy
  \begin{equation}
    \norm{\ket{\eta} - \ket{M}^{\otimes k}}_2 \le \delta.
    \label{eq:gaussian-rank-approximant}
  \end{equation}
  Write $r \coloneqq \bchiG (\ket{\eta})$ and abbreviate $\mc L \coloneqq \mc L_{k, p}$.
  The orthonormal family constructed in Appendix~\ref{appendix:short-proof} gives
  \begin{equation}
    F_{\mc L} (\ket{M}^{\otimes k})^\dagger F_{\mc L} (\ket{M}^{\otimes k}) = I_{D_{k, p}}.
    \label{eq:tensor-shell-isometry}
  \end{equation}
  Restriction to a subset of columns and Eq.~\eqref{eq:majorana-rank-cap-app} give
  \begin{equation}
    \rank F_{\mc L} (\ket{\eta}) \le r C_{4 k, p}.
    \label{eq:approximant-shell-rank}
  \end{equation}
  Moreover, every Majorana monomial is unitary, and hence
  \begin{equation}
    \norm[\big]{F_{\mc L} (\ket{M}^{\otimes k}) - F_{\mc L} (\ket{\eta})}_{\mathrm F}^2
    = D_{k, p} \norm[\big]{\ket{M}^{\otimes k} - \ket{\eta}}_2^2
    \le D_{k, p} \delta^2.
    \label{eq:tensor-shell-distance}
  \end{equation}
  All $D_{k, p}$ singular values of $F_{\mc L} (\ket{M}^{\otimes k})$ are one.
  Applying Eq.~\eqref{eq:eckart-young-bound} with $R = \rank F_{\mc L} (\ket{\eta})$, and then using Eqs.~\eqref{eq:approximant-shell-rank} and \eqref{eq:tensor-shell-distance}, gives
  \begin{equation}
    D_{k, p} - r C_{4 k, p} \le D_{k, p} - \rank F_{\mc L} (\ket{\eta}) \le \norm[\big]{F_{\mc L} (\ket{M}^{\otimes k}) - F_{\mc L} (\ket{\eta})}_{\mathrm F}^2 \le D_{k, p} \delta^2.
    \label{eq:tensor-shell-rank-comparison}
  \end{equation}
  Rearranging and minimizing over all $\ket{\eta}$ satisfying Eq.~\eqref{eq:gaussian-rank-approximant} proves Eq.~\eqref{eq:finite-main}.
\end{proof}

Combining the selected monomials of all degrees up to $p$ with the same parity gives the stronger finite-copy bound underlying Table~\ref{tab:canonical-resource-scaling}.
For integers $k \ge 1$ and $0 \le p \le 2 k$ and every $0 \le \delta < 1$,
\begin{equation}
  \abchiG (\ket{M}^{\otimes k}) \ge (1 - \delta^2) \frac{\sum_{j = 0}^{\floor{p / 2}} D_{k, p - 2 j}}{C_{4 k, p}}.
  \label{eq:short-certificate-main}
\end{equation}

\begin{proof}[Proof of Eq.~\eqref{eq:short-certificate-main}]
  Fix integers $k \ge 1$ and $0 \le p \le 2 k$.
  Let $F (\ket{\phi})$ have columns $\gamma_{\bm J} \ket{\phi}$ indexed by $\bm J \in \bigcup_{j = 0}^{\floor{p / 2}} \mc{L}_{k, p - 2 j}$, and write $N \coloneqq \sum_{j = 0}^{\floor{p / 2}} D_{k, p - 2 j}$ for their number.
  Lemma~\ref{lem:short-local-frame-app} and its tensor-product construction show that these $N$ columns are orthonormal on $\ket{M}^{\otimes k}$, including columns of different total degrees.

  For a pure Gaussian state $\ket{G}$, the vacuum-sector calculation and Gaussian conjugation in the proof of Lemma~\ref{lem:pure-rho-app} give
  \begin{equation}
    \mc{A}_{p - 2 j} \ket{G} \subset \mc{A}_p \ket{G}, \quad 0 \le j \le \floor{p / 2}.
  \end{equation}
  Thus $\rank F (\ket{G}) \le C_{4 k, p}$.
  Linearity and lower semicontinuity of matrix rank imply $\rank F (\ket{\eta}) \le r C_{4 k, p}$ whenever $\bchiG (\ket{\eta}) \le r$.
  If $\norm{\ket{\eta} - \ket{M}^{\otimes k}}_2 \le \delta$, unitarity of each Majorana monomial gives
  \begin{equation}
    \norm[\big]{F (\ket{M}^{\otimes k}) - F (\ket{\eta})}_{\mathrm F}^2
    = N \norm[\big]{\ket{M}^{\otimes k} - \ket{\eta}}_2^2
    \le N \delta^2.
  \end{equation}
  Since all $N$ target singular values equal one, Lemma~\ref{lem:eckart-young-app} yields $N - r C_{4 k, p} \le N \delta^2$.
  Rearranging and minimizing over $\ket{\eta}$ proves Eq.~\eqref{eq:short-certificate-main}.
\end{proof}

\subsubsection{Exponential lower bound via coefficient asymptotics}
\label{appendix:exponential-estimate}

\begin{theorem}[Saddle-point estimate for coefficients of large powers~\cite{flajolet2009analyticcombinatorics}]
  \label{thm:large-powers-saddle}
  Let $B (z) = \sum_{j \ge 0} b_j z^j$ have nonnegative coefficients and radius of convergence $R > 0$, with $B (0) > 0$ and $\gcd \set*{j : b_j > 0} = 1$.
  Let $A$ be analytic in $\abs{z} < R$ with nonnegative Taylor coefficients, and set $T \coloneqq \lim_{z \to R^-} z B' (z) / B (z)$.
  For $\lambda \in (0, T)$, let $\zeta = \zeta (\lambda) \in (0, R)$ be the unique saddle and $\xi = \xi (\lambda)$ its curvature, defined by
  \begin{equation}
    \frac{\zeta B' (\zeta)}{B (\zeta)} = \lambda, \quad \xi \coloneqq \evaluate*{\frac{d^2}{d z^2} \bigl(\log B (z) - \lambda \log z\bigr)}_{z = \zeta}.
    \label{eq:large-powers-saddle-point}
  \end{equation}
  As $n \to \infty$ through positive integers,
  \begin{equation}
    [z^{\lambda n}] A (z) B (z)^n = \frac{A (\zeta) B (\zeta)^n}{\zeta^{\lambda n + 1} \sqrt{2 \pi n \xi}} \bigl(1 + o (1)\bigr),
    \label{eq:large-powers-saddle}
  \end{equation}
  uniformly for $\lambda n \in \Z$ and $\lambda$ in any compact subset of $(0, T)$.
\end{theorem}

\begin{proof}[Proof of the second assertion of Theorem~\ref{thm:main_app}]
  It remains to estimate the ratio in Eq.~\eqref{eq:finite-main}.
  We take $p_k \coloneqq \floor{8 k / 9}$.
  This value maximizes the exponential base, as calculated in Eqs.~\eqref{eq:weighted-balanced-derivative}--\eqref{eq:weighted-balanced-base}.

  The degree enumerator of the selected one-block monomials is
  \begin{equation}
    P (z) \coloneqq 1 + 8 z + 7 z^2.
    \label{eq:local-shell-main}
  \end{equation}
  Put $\lambda_k \coloneqq p_k / k$.
  The polynomial $P$ satisfies the hypotheses of Theorem~\ref{thm:large-powers-saddle} with $T = \lim_{z \to \infty} z P' (z) / P (z) = 2$.
  Choose a compact interval in $(0, 2)$ containing $8 / 9$ in its interior.
  Since $\lambda_k \to 8 / 9$ and $k \lambda_k = p_k \in \Z$, the uniform estimate applies for all sufficiently large $k$.

  Apply Theorem~\ref{thm:large-powers-saddle} with $A (z) = 1$, $B (z) = P (z)$, and $\lambda = \lambda_k$.
  Write $\zeta_k \coloneqq \zeta (\lambda_k)$ and $\xi_k \coloneqq \xi (\lambda_k)$.
  The saddle equation becomes
  \begin{equation}
    \frac{\zeta_k}{P (\zeta_k)}
    P' (\zeta_k) = \lambda_k.
    \label{eq:tensor-shell-moving-saddle}
  \end{equation}
  The left-hand side of Eq.~\eqref{eq:tensor-shell-moving-saddle}, viewed as a function of $\zeta_k$, is strictly increasing from zero to two on the positive real axis.
  At $\zeta = 2 / 7$, it has value $8 / 9$ and derivative $112 / 81$.
  The implicit-function theorem and $\abs{\lambda_k - 8 / 9} < 1 / k$ therefore give
  \begin{equation}
    \zeta_k = \frac{2}{7} + O (1 / k), \quad \xi_k \to \frac{392}{81}.
    \label{eq:tensor-shell-saddle-limits}
  \end{equation}

  By Eqs.~\eqref{eq:D-shell-app} and \eqref{eq:large-powers-saddle},
  \begin{equation}
    U_k \coloneqq \frac{D_{k, p_k} \zeta_k^{p_k + 1} \sqrt{2 \pi k \xi_k}} {P (\zeta_k)^k} \longrightarrow 1.
    \label{eq:tensor-shell-saddle-normalization}
  \end{equation}

  We turn next to $C_{4 k, p_k}$.
  For every integer $q \ge 0$, define the normalized summand
  \begin{equation}
    r_{k, q} \coloneqq
    \begin{cases*}
      \binom{4 k}{p_k - 2 q} \binom{4 k}{p_k}^{-1}
      & if $0 \le 2 q \le p_k$, \\
      0
      & if $2 q > p_k$.
    \end{cases*}
    \label{eq:pure-gaussian-rank-cap-normalized-summand}
  \end{equation}
  These summands admit the product representation
  \begin{equation}
    r_{k, q} = \prod_{j = 1}^{q} \frac{(p_k - 2 j + 2) (p_k - 2 j + 1)} {(4 k - p_k + 2 j - 1) (4 k - p_k + 2 j)}.
    \label{eq:pure-gaussian-rank-cap-summand-product}
  \end{equation}
  The empty product is understood to be one.
  For each fixed $q$, Eq.~\eqref{eq:pure-gaussian-rank-cap-summand-product} gives
  \begin{equation}
    r_{k, q} \longrightarrow \mleft(\frac{4}{49}\mright)^q.
    \label{eq:pure-gaussian-rank-cap-summand-limit}
  \end{equation}

  To control the tail uniformly in $q$, observe that whenever $r_{k, q}$ is nonzero, each factor in Eq.~\eqref{eq:pure-gaussian-rank-cap-summand-product} is at most its value at $j = 1$.
  For the remaining values of $q$, the following bound is immediate:
  \begin{equation}
    0 \le r_{k, q} \le \mleft( \frac{p_k (p_k - 1)} {(4 k - p_k + 1) (4 k - p_k + 2)} \mright)^q.
    \label{eq:pure-gaussian-rank-cap-summand-bound}
  \end{equation}
  The base in Eq.~\eqref{eq:pure-gaussian-rank-cap-summand-bound} tends to $4 / 49$.
  Hence, for all sufficiently large $k$, the bound $r_{k, q} \le 2^{-q}$ holds for every integer $q \ge 0$.
  For any integer $N \ge 0$, it follows that
  \begin{equation}
    \begin{split}
      \abs*{
        \sum_{q = 0}^{\infty} r_{k, q}
        - \sum_{q = 0}^{\infty} \mleft(\frac{4}{49}\mright)^q
      }
      & \le \sum_{q = 0}^{N} \abs*{r_{k, q} - \mleft(\frac{4}{49}\mright)^q}
      + \sum_{q = N + 1}^{\infty} r_{k, q}
      + \sum_{q = N + 1}^{\infty} \mleft(\frac{4}{49}\mright)^q \\
      & \le \sum_{q = 0}^{N} \abs*{r_{k, q} - \mleft(\frac{4}{49}\mright)^q}
      + 2^{-N}
      + \frac{(4 / 49)^{N + 1}}{1 - 4 / 49}.
    \end{split}
    \label{eq:pure-gaussian-rank-cap-truncation-bound}
  \end{equation}
  For fixed $N$, the finite sum tends to zero as $k \to \infty$ by Eq.~\eqref{eq:pure-gaussian-rank-cap-summand-limit}.
  Taking the limit superior in Eq.~\eqref{eq:pure-gaussian-rank-cap-truncation-bound}, and then sending $N$ to infinity, proves that
  \begin{equation}
    R_k \coloneqq C_{4 k, p_k} \binom{4 k}{p_k}^{-1} = \sum_{q = 0}^{\infty} r_{k, q} \longrightarrow \sum_{q = 0}^{\infty} \mleft(\frac{4}{49}\mright)^q = \frac{49}{45}.
    \label{eq:pure-gaussian-rank-cap-binomial-ratio}
  \end{equation}

  The remaining binomial coefficient is estimated by Stirling's formula.
  Let $\mu_k \coloneqq p_k / (4 k)$, so that $\mu_k = 2 / 9 + O (1 / k)$.
  We obtain
  \begin{equation}
    S_k \coloneqq \binom{4 k}{p_k} \sqrt{8 \pi k \mu_k (1 - \mu_k)} \mu_k^{p_k} (1 - \mu_k)^{4 k - p_k} \longrightarrow 1.
    \label{eq:binomial-stirling-limit}
  \end{equation}

  We now combine these estimates.
  Put $\theta_k \coloneqq 8 k / 9 - p_k$, so that $0 \le \theta_k < 1$.
  Eqs.~\eqref{eq:tensor-shell-saddle-normalization}, \eqref{eq:pure-gaussian-rank-cap-binomial-ratio}, and \eqref{eq:binomial-stirling-limit} give the exact factorization
  \begin{equation}
    \begin{split}
      \mleft(\frac{243}{343}\mright)^k \frac{D_{k, p_k}}{C_{4 k, p_k}}
      & = \frac{U_k}{R_k S_k}
      \frac{2}{\zeta_k}
      \sqrt{\frac{\mu_k (1 - \mu_k)}{\xi_k}}
      \mleft(\frac{\zeta_k (1 - \mu_k)}{\mu_k}\mright)^{\theta_k} \\
      & \quad \times \mleft(\frac{7 P (\zeta_k)}{27}\mright)^k
      \mleft(\frac{2}{7 \zeta_k}\mright)^{8 k / 9}
      \mleft(\frac{9 \mu_k}{2}\mright)^{8 k / 9}
      \mleft(\frac{9 (1 - \mu_k)}{7}\mright)^{28 k / 9}.
    \end{split}
    \label{eq:coefficient-ratio-factorization}
  \end{equation}

  Every factor on the first line of Eq.~\eqref{eq:coefficient-ratio-factorization} converges to a positive limit.
  Each base on the second line is positive and differs from one by $O (1 / k)$; its logarithm is therefore $O (1 / k)$.
  Since the exponents are proportional to $k$, the logarithms of the corresponding powers remain bounded.
  These powers are therefore bounded above and away from zero, giving
  \begin{equation}
    \liminf_{k \to \infty} \mleft(\frac{243}{343}\mright)^k \frac{D_{k, p_k}}{C_{4 k, p_k}} > 0.
    \label{eq:coefficient-ratio-asymptotic}
  \end{equation}
  Eq.~\eqref{eq:coefficient-ratio-asymptotic} gives the required estimate for all sufficiently large $k$.
  The remaining finitely many ratios $D_{k, p_k} / C_{4 k, p_k}$ are positive, so decreasing the constant extends the estimate to every $k \ge 1$.
  Combining this estimate with Eq.~\eqref{eq:finite-main} proves Eq.~\eqref{eq:exponential-main}.
  The upper bound on the same factors also gives, with $p = \floor{8 k / 9}$,
  \begin{equation}
    \frac{D_{k, p}}{C_{4 k, p}}
    = \Theta \mleft(\mleft(\frac{343}{243}\mright)^k\mright).
    \label{eq:short-rate-main}
  \end{equation}
\end{proof}

\subsection{Proof of Theorem~\ref{thm:fermion-universal-main}}
\label{appendix:theorem-two-proof}

We first show how a lower bound on the approximate border Gaussian rank of a postselected output transfers to tensor powers of the target state.
We then construct a four-mode output of the form $\ket*{M_t}$ and obtain its rank bound from a weighted Gram matrix.
Together, these steps prove Theorem~\ref{thm:fermion-universal-main}.

\subsubsection{Gaussian postselection and Gaussian rank}
\label{appendix:gaussian-postselection-rank-lifting}

We reduce the number of modes by selecting a coefficient vector in an occupation-basis expansion.
For $\ket*{\eta} \in \mc H_{n + \ell}$, expansion in the occupation basis of the last $\ell$ modes gives
\begin{equation}
  \ket*{\eta} = \sum_{x \in \set*{0, 1}^\ell} \ket*{\eta_x} \otimes \ket*{x},
  \quad
  \ket*{\eta_x} = (I \otimes \bra*{x}) \ket*{\eta} \in \mc H_n.
  \label{eq:occupation-coefficient-vectors}
\end{equation}
Postselection selects a nonzero coefficient vector and normalizes it.
Orthogonality of the occupation basis gives
\begin{equation}
  \sum_{x \in \set*{0, 1}^\ell} \norm*{\ket*{\eta_x}}_2^2
  = \norm*{\ket*{\eta}}_2^2.
  \label{eq:occupation-coefficient-norms}
\end{equation}

\begin{lemma}[Preservation of Gaussianity under occupation measurement]
  \label{lem:gaussian-occupation-projection}
  Let $n$, $\ell$ be nonnegative integers and $\ket*{G} \in \mc H_{n + \ell}$ a pure Gaussian vector.
  For every $x \in \set*{0, 1}^\ell$, the coefficient vector $\ket*{G_x}$ in Eq.~\eqref{eq:occupation-coefficient-vectors} is either zero or a pure Gaussian vector in $\mc H_n$.
\end{lemma}

\begin{proof}
    The single-mode projectors $a_j a_j^\dagger$ and $a_j^\dagger a_j$ select an empty and an occupied mode, respectively.
    Both are Gaussian operators, as noted immediately after Definition~2 in Ref.~\cite[p.~224]{bravyi2004lagrangianrepresentationfermioniclinear}.
    For each of the last $\ell$ modes, choose the projector specified by $x$.
    Their product is $I \otimes \ket*{x} \bra*{x}$, which is Gaussian by Corollary~1 of the same reference.
    Acting on $\ket*{G}$ gives $\ket*{G_x} \otimes \ket*{x}$, which is Gaussian or zero by Corollary~2 of the same reference~\cite[p.~225]{bravyi2004lagrangianrepresentationfermioniclinear}.
    In the nonzero case, the reduced state on the first $n$ modes is pure and Gaussian, so $\ket*{G_x}$ is a Gaussian vector.
\end{proof}

A rank bound for an output obtained with positive probability also constrains tensor powers of the target state.

\begin{proposition}[Gaussian rank inequality under postselection]
  \label{prop:typical-branch-lifting}
  Let $\ket*{\psi} \in \mc{H}_n$ and $\ket*{\phi} \in \mc{H}_{n'}$ be normalized pure states.
  Let $K_\omega : \mc H_n \to \mc H_{n_\omega}$, indexed by a finite set $\Omega$, be linear maps formed by composing Gaussian unitaries and occupation-coefficient selections.
  Suppose that $n_{\omega_*} = n'$ and
  \begin{equation}
    \sum_{\omega \in \Omega}
    K_\omega^\dagger K_\omega = I, \quad K_{\omega_*} \ket*{\psi} = \sqrt{p} \ket*{\phi}, \quad \omega_* \in \Omega, \quad 0 < p \le 1.
    \label{eq:completed-protocol-maps}
  \end{equation}
  For every fixed $0 \le \delta < \delta' < 1$, there are positive integers $m_k$, defined for all sufficiently large $k$, such that $m_k / k \to p$ and
  \begin{equation}
    \abchiG (\ket*{\psi}^{\otimes k})
    \ge \bchiG^{(\delta')} (\ket*{\phi}^{\otimes m_k}).
    \label{eq:typical-branch-rank-transfer}
  \end{equation}
\end{proposition}

\begin{proof}
  Keep the successful map $K_{\omega_*}$ unchanged and measure every other output in the occupation basis.
  Denote the resulting maps by
  \begin{equation}
    L_* \coloneqq K_{\omega_*}, \quad L_{(\omega, y)} \coloneqq \bra*{y} K_\omega, \quad \omega \ne \omega_*, \quad y \in \set*{0, 1}^{n_\omega}.
    \label{eq:refined-outcome-maps}
  \end{equation}
  Let $\Lambda$ be their index set.
  The map $L_*$ takes values in $\mc H_{n'}$; all other maps take values in $\C$.
  Orthonormality of the occupation basis gives
  \begin{equation}
    \sum_{\lambda \in \Lambda}
    L_\lambda^\dagger L_\lambda = K_{\omega_*}^\dagger K_{\omega_*} + \sum_{\omega \ne \omega_*} K_\omega^\dagger \mleft(\sum_{y \in \set*{0, 1}^{n_\omega}} \ket*{y} \bra*{y}\mright) K_\omega = I.
    \label{eq:refined-outcome-completeness}
  \end{equation}

  For $\bm\lambda = (\lambda_1, \ldots, \lambda_k) \in \Lambda^k$, define
  \begin{equation}
    T_{\bm\lambda} \coloneqq L_{\lambda_1} \otimes \cdots \otimes L_{\lambda_k}, \quad m (\bm\lambda) \coloneqq \abs*{\set*{j \in [k] : \lambda_j = *}}.
    \label{eq:tensor-coefficient-maps}
  \end{equation}
  Thus $T_{\bm\lambda} \colon \mc H_{n k} \to \mc H_{n' m (\bm\lambda)}$, with the surviving blocks in their original order.
  Tensoring Eq.~\eqref{eq:refined-outcome-completeness} and evaluating on any $\ket*{\xi} \in \mc H_{n k}$ gives
  \begin{equation}
    \sum_{\bm\lambda \in \Lambda^k} \norm*{T_{\bm\lambda} \ket*{\xi}}_2^2
    = \norm*{\ket*{\xi}}_2^2.
    \label{eq:tensor-outcome-norm-identity}
  \end{equation}
  Each $T_{\bm\lambda}$ maps Gaussian vectors to Gaussian vectors or zero, even when the input is entangled across copies.
  Indeed, Lemma~\ref{lem:gaussian-occupation-projection} applies to each coefficient selection with all other modes retained; the Gaussian unitaries and any parity operators introduced by the fixed fermionic ordering also preserve Gaussianity.
  By linearity, each map sends a sum of $r$ Gaussian vectors to a sum of at most $r$ Gaussian vectors.
  Continuity extends this statement to limits of such sums, yielding the border-rank bound
  \begin{equation}
    \bchiG (T_{\bm\lambda} \ket*{\xi})
    \le \bchiG (\ket*{\xi}).
    \label{eq:tensor-outcome-rank-monotonicity}
  \end{equation}

  On the target state, the outcome probabilities factor as
  \begin{equation}
    P_{\bm\lambda} \coloneqq \norm*{T_{\bm\lambda} \ket*{\psi}^{\otimes k}}_2^2 = \prod_{j = 1}^k \norm*{L_{\lambda_j} \ket*{\psi}}_2^2, \quad \sum_{\bm\lambda \in \Lambda^k} P_{\bm\lambda} = 1.
    \label{eq:target-branch-probability}
  \end{equation}
  Let $\mb P$ be the probability measure on $\Lambda$ with $\mb P (\{\lambda\}) = \norm{L_\lambda \ket{\psi}}_2^2$, so $P_{\bm\lambda}$ is the probability of $\bm\lambda$ under $\mb P^{\otimes k}$.
  The outcomes are therefore independent, and each equals $*$ with probability $p$.
  Consequently, $m (\bm\lambda)$ is binomial with mean $k p$ and variance $k p (1 - p)$.
  Chebyshev's inequality gives
  \begin{equation}
    \mb P^{\otimes k} \mleft(\abs*{m (\bm\lambda) / k - p} > k^{-1 / 4}\mright)
    \le \frac{p (1 - p)}{\sqrt{k}}.
    \label{eq:typical-success-count}
  \end{equation}

  Choose $\ket*{\eta_k}$ attaining the approximate border rank, so that $\norm*{\ket*{\psi}^{\otimes k} - \ket*{\eta_k}}_2 \le \delta$ and $\bchiG (\ket*{\eta_k}) = \abchiG (\ket*{\psi}^{\otimes k})$.
  Write $E_{\bm\lambda} \coloneqq \norm*{T_{\bm\lambda} (\ket*{\psi}^{\otimes k} - \ket*{\eta_k})}_2^2$ for the squared error associated with outcome $\bm\lambda$.
  Summing the squared errors using Eq.~\eqref{eq:tensor-outcome-norm-identity} gives
  \begin{equation}
    \mb P^{\otimes k} (E_{\bm\lambda} > \delta'^2 P_{\bm\lambda})
    = \sum_{E_{\bm\lambda} > \delta'^2 P_{\bm\lambda}}
    P_{\bm\lambda} \le \frac{1}{\delta'^2} \sum_{\bm\lambda \in \Lambda^k} E_{\bm\lambda} \le \frac{\delta^2}{\delta'^2}.
    \label{eq:branch-error-budget}
  \end{equation}
  Combining Eqs.~\eqref{eq:branch-error-budget} and~\eqref{eq:typical-success-count}, we obtain
  \begin{equation}
    \mb P^{\otimes k} \mleft(E_{\bm\lambda} \le \delta'^2 P_{\bm\lambda}, \abs*{m (\bm\lambda) / k - p} \le k^{-1 / 4}\mright)
    \ge 1 - \frac{\delta^2}{\delta'^2} - \frac{p (1 - p)}{\sqrt{k}}.
    \label{eq:simultaneous-outcome-bounds}
  \end{equation}
  Since $\delta < \delta'$, the right-hand side is positive for all sufficiently large $k$.
  Thus at least one outcome satisfies both inequalities and has positive probability.
  Choose one, denote it by $\bm\lambda_k$, and set $m_k \coloneqq m (\bm\lambda_k)$.
  Then $m_k / k \to p > 0$, so $m_k$ is eventually positive.

  Each successful factor contributes $\sqrt{p} \ket*{\phi}$, while every other factor is a scalar.
  Thus the selected output has the form
  \begin{equation}
    T_{\bm\lambda_k} \ket*{\psi}^{\otimes k} = c_k \ket*{\phi}^{\otimes m_k}, \quad \abs*{c_k}^2 = P_{\bm\lambda_k} > 0.
    \label{eq:fermionic-typical-branch-target}
  \end{equation}
  Rescaling the corresponding approximant by $c_k^{-1}$ gives
  \begin{equation}
    \norm*{\ket*{\phi}^{\otimes m_k} - c_k^{-1}
    T_{\bm\lambda_k} \ket*{\eta_k}}_2^2 = \frac{E_{\bm\lambda_k}}{P_{\bm\lambda_k}} \le \delta'^2.
    \label{eq:typical-branch-distance}
  \end{equation}
  Finally, Eq.~\eqref{eq:tensor-outcome-rank-monotonicity} and invariance under nonzero scalar multiplication imply
  \begin{equation}
    \bchiG^{(\delta')} (\ket*{\phi}^{\otimes m_k})
    \le \bchiG (c_k^{-1} T_{\bm\lambda_k} \ket*{\eta_k})
    \le \bchiG (\ket*{\eta_k})
    = \abchiG (\ket*{\psi}^{\otimes k}),
    \label{eq:typical-branch-rank}
  \end{equation}
  as required.
\end{proof}

\subsubsection{Reduction to a four-mode normal form}
\label{appendix:four-mode-reduction}

We now construct a postselected output in the four-mode family
\begin{equation}
  \ket*{M_t}
  \coloneqq
  \frac{\ket*{0000} + t \ket*{1111}}
  {\sqrt{1 + t^2}},
  \quad 0 < t \le 1.
  \label{eq:four-mode-normal-form-family}
\end{equation}

An equivalent four-mode normal form appears in Ref.~\cite{gigena2015entanglementfermionsystems}.
The proof of Theorem~1 in the Supplemental Material of Ref.~\cite[pp.~3--4]{hebenstreit2019purefermionicnongaussianstates} reduces the number of modes once at a time.
Proposition~\ref{prop:postselected-four-mode-reduction} gives a direct construction with output $\ket*{M_t}$, up to a phase and a parity change.

\begin{proposition}[Four-mode normal form under Gaussian postselection]
  \label{prop:postselected-four-mode-reduction}
  Let $\ket*{\psi} \in \mc{H}_n$ be a normalized pure non-Gaussian state, with $n \ge 4$.
  There exist a Gaussian unitary $U$, real parameters $0 < p_\psi \le 1$ and $0 < t \le 1$, and a phase $\theta \in \R$ such that
  \begin{equation}
    (I \otimes \bra*{0^{n - 4}}) U \ket*{\psi}
    = e^{i \theta} \sqrt{p_\psi} \gamma_1^\epsilon \ket*{M_t},
    \label{eq:postselected-four-mode-output}
  \end{equation}
  where $\epsilon = 0$ for even $\ket*{\psi}$ and $\epsilon = 1$ for odd $\ket*{\psi}$.
  The operator $\gamma_1$ on the right acts on the first retained mode.
\end{proposition}

\begin{proof}
  Suppose first that $\ket*{\psi}$ has even parity.
  Choose a Gaussian unitary $V$ that maximizes $\abs*{\bra*{0^n} V \ket*{\psi}}$.
  Compactness of the Gaussian unitary group ensures that a maximizer exists.
  The maximum is positive because every even occupation vector is a Gaussian image of the vacuum, up to a phase.
  For $J \subset [n]$, write $\ket*{J}$ for the occupation vector with precisely the modes in $J$ occupied, and expand
  \begin{equation}
    V \ket*{\psi} = \sum_{J \subset [n]} c_J \ket*{J}, \quad c_{\varnothing} \ne 0.
    \label{eq:reduction-maximal-vacuum-overlap}
  \end{equation}
  Maximality forces every two-particle coefficient to vanish.
  To see this, fix $1 \le r < s \le n$ and a real phase $\varphi$, and put $K_\varphi \coloneqq e^{i \varphi} a_r^\dagger a_s^\dagger - e^{-i \varphi} a_s a_r$.
  Since $K_\varphi$ is quadratic and anti-Hermitian, $G_\varphi (u) \coloneqq e^{u K_\varphi}$ is a Gaussian unitary for real $u$.
  The identities $K_\varphi \ket*{0^n} = e^{i \varphi} \ket*{\set*{r, s}}$ and $K_\varphi \ket*{\set*{r, s}} = -e^{-i \varphi} \ket*{0^n}$ give
  \begin{equation}
    G_\varphi (u) \ket*{0^n} = \cos u \ket*{0^n} + e^{i \varphi} \sin u \ket*{\set*{r, s}}.
    \label{eq:reduction-gaussian-pair-rotation}
  \end{equation}
  Applying maximality of $V$ to the Gaussian unitary $G_\varphi (u)^\dagger V$ gives
  \begin{equation}
    \abs*{c_{\varnothing} \cos u + e^{-i \varphi} c_{\set*{r, s}} \sin u}^2
    = \abs*{\bra*{0^n}
    G_\varphi (u)^\dagger V \ket*{\psi}}^2 \le \abs*{c_{\varnothing}}^2.
    \label{eq:reduction-pair-overlap-bound}
  \end{equation}
  Equality holds at $u = 0$, where the derivative must vanish:
  \begin{equation}
    0
    = \evaluate*{\odv{}{u} \abs*{c_{\varnothing} \cos u + e^{-i \varphi} c_{\set*{r, s}} \sin u}^2}_{u = 0}
    = 2 \Re (\overline{c_{\varnothing}} e^{-i \varphi} c_{\set*{r, s}}).
    \label{eq:reduction-vanishing-pair-coefficients}
  \end{equation}
  Taking $\varphi = 0$ and $\varphi = \pi / 2$ gives $c_{\set*{r, s}} = 0$, since $c_{\varnothing} \ne 0$.

  Since $V$ preserves Gaussianity, $V \ket*{\psi}$ is not a multiple of the vacuum.
  Choose a nonempty $S \subset [n]$ of smallest cardinality with $c_S \ne 0$, and set $d \coloneqq \abs*{S}$.
  Parity and Eq.~\eqref{eq:reduction-vanishing-pair-coefficients} imply that $d$ is even and $d \ge 4$.
  Write $S = \set*{s_1 < \cdots < s_d}$.
  Choose a passive Gaussian unitary $P$ implementing a mode permutation that sends $s_j$ to $j$ for $1 \le j \le d$.
  Since $P$ preserves particle number, it fixes the vacuum up to a phase $e^{i \kappa}$:
  \begin{equation}
    P \ket*{0^n} = e^{i \kappa} \ket*{0^n}, \quad P a_{s_j}^\dagger P^\dagger = a_j^\dagger, \quad 1 \le j \le d.
    \label{eq:reduction-permutation-action}
  \end{equation}
  The increasing order of the occupied modes gives $P \ket*{S} = e^{i \kappa} \ket*{1^d} \otimes \ket*{0^{n - d}}$ with the same phase and no additional fermionic sign.
  Selecting the vacuum on the last $n - d$ modes now removes every term with $J \not \subset S$, while minimality of $S$ gives $c_J = 0$ for every proper nonempty subset $J \subset S$.
  Only the vacuum and $S$ terms remain:
  \begin{equation}
    (I \otimes \bra*{0^{n - d}}) P V \ket*{\psi}
    = e^{i \kappa} (c_{\varnothing} \ket*{0^d} + c_S \ket*{1^d}).
    \label{eq:reduction-minimal-support}
  \end{equation}

  It remains to remove $d - 4$ modes without changing the ratio of the two coefficients.
  All unitaries act on the original $n$-mode space until the final vacuum selection; $I_{\mc H_q}$ denotes the identity on $\mc H_q$.
  For each adjacent pair $(r, s) = (5 + 2 j, 6 + 2 j)$, where $0 \le j < (d - 4) / 2$, use the Gaussian unitary $Q_{r s} \colon \mc H_n \to \mc H_n$ given by
  \begin{equation}
    Q_{r s} \coloneqq \exp \mleft[\frac{\pi}{4} (a_s a_r - a_r^\dagger a_s^\dagger)\mright].
    \label{eq:reduction-pair-unitary}
  \end{equation}
  Applying $Q_{r s}$ and then selecting the vacuum on modes $r$, $s$ gives the following identity of maps from $\mc H_n$ to $\mc H_{n - 2}$:
  \begin{equation}
    (I_{\mc H_{r - 1}} \otimes \bra*{00} \otimes I_{\mc H_{n - s}}) Q_{r s}
    = I_{\mc H_{r - 1}} \otimes \frac{\bra*{00} + \bra*{11}}{\sqrt{2}} \otimes I_{\mc H_{n - s}}.
    \label{eq:reduction-pair-selection}
  \end{equation}
  The unselected modes retain their original order, and each pair contributes $1 / \sqrt{2}$ to both terms in Eq.~\eqref{eq:reduction-minimal-support}.
  The operators $Q_{r s}$ commute because their quadratic generators act on disjoint pairs.
  Their product has the form $W = W_d \otimes I_{\mc H_{n - d}}$; for $d = 4$, take $W = I_{\mc H_n}$.
  Set $U \coloneqq W P V$.
  The final vacuum selection gives
  \begin{equation}
    \begin{split}
      (I_{\mc H_4} \otimes \bra*{0^{n - 4}}) U \ket*{\psi}
      & = (I_{\mc H_4} \otimes \bra*{0^{d - 4}}) W_d
      (I_{\mc H_d} \otimes \bra*{0^{n - d}}) P V \ket*{\psi} \\
      & = e^{i \kappa} (I_{\mc H_4} \otimes \bra*{0^{d - 4}}) W_d
      (c_{\varnothing} \ket*{0^d} + c_S \ket*{1^d}) \\
      & = e^{i \kappa} 2^{-(d - 4) / 4}
      (c_{\varnothing} \ket*{0000} + c_S \ket*{1111}).
    \end{split}
    \label{eq:reduction-four-mode-coefficient}
  \end{equation}
  The success probability is $p_\psi = 2^{-(d - 4) / 2} (\abs*{c_{\varnothing}}^2 + \abs*{c_S}^2)$.
  It is positive because both coefficients are nonzero, and at most one by Eq.~\eqref{eq:occupation-coefficient-norms}.

  Since $\ket*{S}$ is Gaussian, the maximal vacuum overlap also gives $\abs*{c_S} \le \abs*{c_{\varnothing}}$.
  Hence $t \coloneqq \abs*{c_S} / \abs*{c_{\varnothing}} \in (0, 1]$.
  Write $c_{\varnothing} = \abs*{c_{\varnothing}} e^{i \alpha}$ and $c_S = \abs*{c_S} e^{i \beta}$, with real $\alpha$, $\beta$.
  On the retained four modes, the Gaussian unitary
  \begin{equation}
    R \coloneqq \exp \mleft[i (\alpha - \beta) \mleft(a_1^\dagger a_1 - \frac{1}{2}\mright)\mright]
    \label{eq:reduction-relative-phase-rotation}
  \end{equation}
  removes the relative phase: its factors on the vacuum and fully occupied vectors are $e^{-i (\alpha - \beta) / 2}$ and $e^{i (\alpha - \beta) / 2}$, respectively.
  The rotation $R$ acts only on the retained modes, so it may be applied before or after the vacuum selection.
  Absorb the permutation phase into $\theta \coloneqq \kappa + (\alpha + \beta) / 2$.
  Eq.~\eqref{eq:reduction-four-mode-coefficient} then becomes
  \begin{equation}
    \begin{split}
      (I_{\mc H_4} \otimes \bra*{0^{n - 4}})
      (R \otimes I_{\mc H_{n - 4}}) U \ket*{\psi}
      & = R (I_{\mc H_4} \otimes \bra*{0^{n - 4}}) U \ket*{\psi} \\
      & = e^{i \theta} 2^{-(d - 4) / 4}
      (\abs*{c_{\varnothing}} \ket*{0000} + \abs*{c_S} \ket*{1111}) \\
      & = e^{i \theta} \sqrt{p_\psi} \ket*{M_t}.
    \end{split}
    \label{eq:reduction-postselected-normal-form}
  \end{equation}
  Absorbing $R \otimes I_{\mc H_{n - 4}}$ into $U$ proves the even case.
  For $n = 4$, necessarily $d = 4$, so no modes are removed and $p_\psi = 1$.

  For odd $\ket*{\psi}$, apply the even construction to $\gamma_1 \ket*{\psi}$, obtaining $U_{\mathrm e}$.
  Conjugation by $\gamma_1$ preserves quadratic Majorana generators, so $U \coloneqq \gamma_1 U_{\mathrm e} \gamma_1$ is Gaussian.
  Since $\gamma_1$ acts on the first retained mode, Eq.~\eqref{eq:reduction-postselected-normal-form} gives
  \begin{equation}
    (I_{\mc H_4} \otimes \bra*{0^{n - 4}}) U \ket*{\psi} = \gamma_1 (I_{\mc H_4} \otimes \bra*{0^{n - 4}}) U_{\mathrm e} (\gamma_1 \ket*{\psi}) = e^{i \theta} \sqrt{p_\psi} \gamma_1 \ket*{M_t}.
    \label{eq:reduction-odd-parity-output}
  \end{equation}
  This is Eq.~\eqref{eq:postselected-four-mode-output} with $\epsilon = 1$, completing the proof.
\end{proof}

We can now apply Proposition~\ref{prop:typical-branch-lifting} to the reduction in Proposition~\ref{prop:postselected-four-mode-reduction}.
Choose $U$, $p_\psi$, $t$, $\theta$, $\epsilon$ from that reduction and include every occupation outcome on the last $n - 4$ modes:
\begin{equation}
  K_x \coloneqq (I \otimes \bra*{x}) U \colon \mc H_n \to \mc H_4, \quad x \in \set*{0, 1}^{n - 4}.
  \label{eq:normal-form-complete-outcomes}
\end{equation}
Orthonormality of the occupation basis gives
\begin{equation}
  \sum_{x \in \set*{0, 1}^{n - 4}}
  K_x^\dagger K_x = U^\dagger \mleft(I \otimes \sum_{x \in \set*{0, 1}^{n - 4}} \ket*{x} \bra*{x}\mright) U = I.
  \label{eq:normal-form-outcome-completeness}
\end{equation}
The vacuum outcome satisfies $K_{0^{n - 4}} \ket*{\psi} = e^{i \theta} \sqrt{p_\psi} \gamma_1^\epsilon \ket*{M_t}$.
Multiplication by $e^{i \theta} \gamma_1^\epsilon$ preserves $\norm{\cdot }_2$ and maps pure Gaussian states bijectively onto pure Gaussian states.
The same holds for its tensor powers, so removing the phase and Majorana factor from the normalized output does not change its approximate border Gaussian rank.
Proposition~\ref{prop:typical-branch-lifting} therefore gives, for every fixed $0 \le \delta < \delta' < 1$, positive integers $m_k$ with $m_k / k \to p_\psi$ such that, for all sufficiently large $k$,
\begin{equation}
  \abchiG (\ket*{\psi}^{\otimes k})
  \ge \bchiG^{(\delta')} (\ket*{M_t}^{\otimes m_k}).
  \label{eq:general-state-normal-form-rank-bound}
\end{equation}
Here $p_\psi > 0$ and $t \in (0, 1]$ depend on $\ket*{\psi}$, not on the error tolerances.

\subsubsection{Gaussian rank lower bounds for \texorpdfstring{\IfFontSeriesContextTF{bf}{\boldmath}{}$\ket{M_t}^{\otimes k}$}{the k-fold tensor power of ket M with parameter t}}
\label{appendix:tunable}
\label{appendix:weighted-majorana-entropy}

The reduction in Proposition~\ref{prop:postselected-four-mode-reduction} leaves us with $\ket*{M_t}$, where $0 < t \le 1$.
For each fixed $t$, the spectrum of a weighted local Gram matrix will give an exponential lower bound.
The base is determined by an entropy maximization over three real parameters and equals $343 / 243$ at $t = 1$.

Set
\begin{equation}
  \theta_t \coloneqq \frac{t^2}{1 + t^2}, \quad
  N_t \coloneqq \exp (h (\theta_t)).
  \label{eq:weighted-entropy-parameters}
\end{equation}
Since $0 < \theta_t \le 1 / 2$, we have $1 < N_t \le 2$.

We first record an elementary upper bound for comparison with the exponential lower bound.
Fix $0 < t \le 1$.
For $x \in \set{0, 1}^k$, let $\ket*{b (x)}$ be the occupation vector whose $j$-th four-mode block is empty when $x_j = 0$ and fully occupied when $x_j = 1$.
The tensor expansion gives
\begin{equation}
  \ket*{M_t}^{\otimes k}
  = \sum_{x \in \set{0, 1}^k} \sqrt{q_k (x)} \ket*{b (x)},
  \quad q_k (x) \coloneqq \theta_t^{\abs{x}} (1 - \theta_t)^{k - \abs{x}}.
  \label{eq:four-mode-occupation-expansion}
\end{equation}
Under the probability distribution $q_k$, the coordinates are independent Bernoulli variables with parameter $\theta_t$.
Hence $ - \log q_k (x)$ is a sum of $k$ independent variables, each with mean $h (\theta_t)$ and variance $v_t \coloneqq \theta_t (1 - \theta_t) (\log ((1 - \theta_t) / \theta_t))^2$.
Retain the terms indexed by
\begin{equation}
  \mc T_k \coloneqq \set*{x \in \set{0, 1}^k : -\log q_k (x) \le k h (\theta_t) + k^{3 / 4}},
  \label{eq:upper-bound-typical-set}
\end{equation}
and denote the resulting truncation by $\ket*{\eta_k}$.
Orthogonality of the occupation vectors and Chebyshev's inequality give
\begin{equation}
  \norm*{\ket*{M_t}^{\otimes k} - \ket*{\eta_k}}_2^2
  = \sum_{x \notin \mc T_k} q_k (x)
  \le \frac{k v_t}{(k^{3 / 4})^2}
  = \frac{v_t}{\sqrt{k}} \longrightarrow 0.
  \label{eq:occupation-truncation-error}
\end{equation}
By Eq.~\eqref{eq:upper-bound-typical-set}, every retained squared coefficient is at least $\exp (-k h (\theta_t) - k^{3 / 4})$.
Their sum is at most one, so $\abs{\mc T_k} \le \exp (k^{3 / 4}) N_t^k$.
Each occupation vector is Gaussian, and therefore $\chiG (\ket*{\eta_k}) \le \abs{\mc T_k}$.
For every fixed $0 < \delta < 1$, Eq.~\eqref{eq:occupation-truncation-error} makes $\ket*{\eta_k}$ an admissible approximant for all sufficiently large $k$, yielding
\begin{equation}
  \abchiG (\ket*{M_t}^{\otimes k})
  \le \achiG (\ket*{M_t}^{\otimes k})
  \le \exp (k^{3 / 4}) N_t^k.
  \label{eq:weighted-approximate-upper-bound}
\end{equation}

We now turn to the lower bound.
On one four-mode block, take the identity, the eight single Majoranas, $\gamma_1 \gamma_2$, and the six quadratic monomials
\begin{equation}
  \gamma_1 \gamma_3, \quad \gamma_2 \gamma_3, \quad
  \gamma_1 \gamma_5, \quad \gamma_2 \gamma_5, \quad
  \gamma_1 \gamma_7, \quad \gamma_2 \gamma_7.
  \label{eq:local-quadratic-pairs}
\end{equation}
Acting on $\ket*{M_t}$, these give one degree-zero, eight degree-one, and seven degree-two normalized columns.
To compute their overlaps, let $e_j$ denote the four-bit string with only mode $j$ occupied and let $\overline{x}$ denote the bitwise complement of $x$.
Ordering the creation operators by mode gives, for $1 \le j \le 4$,
\begin{equation}
  \gamma_{2 j - 1} \ket*{M_t} = \frac{\ket*{e_j} + (-1)^{j - 1} t \ket*{\overline{e}_j}}{\sqrt{1 + t^2}}, \quad - i \gamma_{2 j} \ket*{M_t} = \frac{\ket*{e_j} - (-1)^{j - 1} t \ket*{\overline{e}_j}}{\sqrt{1 + t^2}}.
  \label{eq:tunable-odd-pairs-app}
\end{equation}
For $2 \le j \le 4$, let $x_j$ have precisely modes $1$ and $j$ occupied.
The six quadratic monomials in Eq.~\eqref{eq:local-quadratic-pairs} give
\begin{equation}
  \gamma_1 \gamma_{2 j - 1} \ket*{M_t} = \frac{\ket*{x_j} + (-1)^{j - 1} t \ket*{\overline{x}_j}}{\sqrt{1 + t^2}}, \quad - i \gamma_2 \gamma_{2 j - 1} \ket*{M_t} = \frac{\ket*{x_j} - (-1)^{j - 1} t \ket*{\overline{x}_j}}{\sqrt{1 + t^2}}.
  \label{eq:tunable-even-pairs-app}
\end{equation}
In each pair, orthogonality of the two occupation vectors gives the overlap
\begin{equation}
  \sigma_t \coloneqq \frac{1 - t^2}{1 + t^2}.
  \label{eq:weighted-pair-overlap}
\end{equation}
The four odd pairs have disjoint supports $\set*{e_j, \overline{e}_j}$ in the one- and three-particle sectors.
The three even pairs have disjoint supports $\set*{1100, 0011}$, $\set*{1010, 0101}$, and $\set*{1001, 0110}$ in the two-particle sector.
These supports contain neither the vacuum nor the fully occupied configuration, so every pair is orthogonal to the identity column $\ket*{M_t}$.
The remaining quadratic monomial satisfies
\begin{equation}
  - i \gamma_1 \gamma_2 \ket*{M_t}
  = \frac{\ket*{0000} - t \ket*{1111}}{\sqrt{1 + t^2}}.
  \label{eq:local-vacuum-full-pair}
\end{equation}
This column is orthogonal to all seven pairs, and its overlap with the identity column is $\sigma_t$.
Together with the identity column, it therefore gives an eighth pair with the same Gram matrix.
After grouping and rephasing the columns as in Eqs.~\eqref{eq:tunable-odd-pairs-app}, \eqref{eq:tunable-even-pairs-app}, and \eqref{eq:local-vacuum-full-pair}, the full Gram matrix is
\begin{equation}
  \begin{bmatrix}
    1        & \sigma_t \\
    \sigma_t & 1
  \end{bmatrix}^{\oplus 8}.
  \label{eq:local-gram-main}
\end{equation}
Each $2 \times 2$ block has eigenvalues
\begin{equation}
  1 + \sigma_t = \frac{2}{1 + t^2}, \quad 1 - \sigma_t = \frac{2 t^2}{1 + t^2}.
  \label{eq:weighted-pair-eigenvalues}
\end{equation}
Eq.~\eqref{eq:local-gram-main} shows that the Gram matrix is the identity precisely when $t = 1$.

For tensor powers of $\ket*{M_t}$, blockwise factorization holds up to fermionic signs, which multiply individual columns by $ \pm 1$.
Changing these signs conjugates the Gram matrix by a diagonal unitary and leaves its spectrum unchanged.

Divide the 16 Majorana monomials into four groups: the identity, the eight single Majoranas, the six quadratic monomials in Eq.~\eqref{eq:local-quadratic-pairs}, and $\gamma_1 \gamma_2$.
Assign total weights $w_0$, $w_1$, $w_2$, $w_3$ to these groups, distributed equally among the monomials within each group.
The allowed weights form the simplex
\begin{equation}
  \Delta_3 \coloneqq \set*{\bm w = (w_0, w_1, w_2, w_3) \in \R^4 : w_j \ge 0, \quad w_0 + w_1 + w_2 + w_3 = 1}.
  \label{eq:weighted-action-simplex}
\end{equation}
Multiply each column by the square root of its weight.
After rephasing and grouping the columns, their Gram matrix is
\begin{equation}
  \rho_t (\bm w)
  \coloneqq
  \begin{bmatrix}
    w_0                     & \sigma_t \sqrt{w_0 w_3} \\
    \sigma_t \sqrt{w_0 w_3} & w_3
  \end{bmatrix}
  \oplus
  \mleft(\frac{w_1}{8}
    \begin{bmatrix}
      1        & \sigma_t \\
      \sigma_t & 1
  \end{bmatrix}\mright)^{\oplus 4}
  \oplus
  \mleft(\frac{w_2}{6}
    \begin{bmatrix}
      1        & \sigma_t \\
      \sigma_t & 1
  \end{bmatrix}\mright)^{\oplus 3}.
  \label{eq:weighted-local-gram}
\end{equation}
This matrix is positive semidefinite and has trace one.
Its rows and columns are indexed by the selected Majorana monomials.
The first block has eigenvalues
\begin{equation}
  \frac{w_0 + w_3 \pm \sqrt{(w_0 - w_3)^2 + 4 \sigma_t^2 w_0 w_3}}{2}.
  \label{eq:weighted-local-spectrum}
\end{equation}
The remaining blocks contribute $w_1 (1 \pm \sigma_t) / 8$, each with multiplicity four, and $w_2 (1 \pm \sigma_t) / 6$, each with multiplicity three.
The four groups have degrees zero, one, two, and two, respectively, so their weighted mean degree is $\mu (\bm w) = w_1 + 2 w_2 + 2 w_3$.
Define the base $B_t$ by
\begin{equation}
  \log B_t
  \coloneqq \max_{\bm w \in \Delta_3}
  \set*{H (\rho_t (\bm w)) - 4 h (\mu (\bm w) / 4)}.
  \label{eq:weighted-entropy-base}
\end{equation}
The objective is continuous on the whole simplex, including its boundary, so compactness ensures that the maximum is attained.

\begin{theorem}[Gaussian rank lower bound from a weighted Gram matrix]
  \label{thm:weighted-spectral-bound}
  For every $0 < t \le 1$, there is a sequence $f_t (k) > 0$, independent of the error tolerance, with $\log f_t (k) / k \to 0$ as $k \to \infty$.
  For every fixed $0 \le \delta < 1$ and all sufficiently large $k$,
  \begin{equation}
    \abchiG (\ket*{M_t}^{\otimes k})
    \ge (1 - \delta^2) f_t (k) B_t^k.
    \label{eq:weighted-robust-rank-bound}
  \end{equation}
  The base $B_t$ is continuous on $(0, 1]$ and satisfies
  \begin{equation}
    1 < B_t \le N_t,
    \quad
    B_1 = \frac{343}{243},
    \quad
    \lim_{t \to 0}
    B_t = 1.
    \label{eq:weighted-base-properties}
  \end{equation}
\end{theorem}

\begin{proof}
  We first show that $B_t > 1$ and that every maximizing weight in Eq.~\eqref{eq:weighted-entropy-base} has mean degree strictly between zero and two.
  At $\bm w = (1 - u, u, 0, 0)$, the objective is
  \begin{equation}
    H (\rho_t (1 - u, u, 0, 0)) - 4 h (u / 4) = h (u) + u \log (4 N_t) - 4 h (u / 4).
    \label{eq:weighted-positive-rate}
  \end{equation}
  This vanishes at $u = 0$, and its derivative for $0 < u < 1$ is $\log (N_t (1 - u) / (1 - u / 4))$, which tends to $\log N_t > 0$ as $u \to 0$.
  For sufficiently small positive $u$, the objective is positive.
  Hence $\log B_t > 0$.
  At mean degree zero, only the identity has positive weight and the objective is zero.
  At mean degree two, only the seven quadratic monomials can have positive weight, so the objective is at most $\log 7 - 4 \log 2 < 0$.

  Fix a maximizing weight $\bm w$, and write $\rho = \rho_t (\bm w)$, $H = H (\rho)$, and $\mu = \mu (\bm w) \in (0, 2)$.
  The positive eigenvalues of $\rho$, counted with multiplicity, sum to one and define a probability distribution.
  Assigning the value $ - \log \lambda$ to each eigenvalue $\lambda$, with probability $\lambda$, gives mean $H$ and finite variance $v$.

  Let $D_k$ agree with $\rho^{\otimes k}$ on the eigenspaces whose eigenvalues exceed $\exp (-k H + k^{3 / 4})$, and vanish on the remaining eigenspaces.
  Set $A_k = \rho^{\otimes k} - D_k$.
  Under the product eigenvalue distribution, the negative logarithm has mean $k H$ and variance $k v$.
  The trace $\Tr D_k$ is the probability that this random variable falls more than $k^{3 / 4}$ below its mean, which Chebyshev's inequality bounds by $k v / k^{3 / 2} = v / \sqrt{k}$.
  Together with the construction of $A_k$ and $D_k$, this gives
  \begin{equation}
    0 \le A_k \le \exp (-k H + k^{3 / 4})
    I, \quad D_k \ge 0, \quad \Tr D_k \le \frac{v}{\sqrt{k}}.
    \label{eq:weighted-spectral-truncation}
  \end{equation}
  Matrix inequalities are in the positive semidefinite order: $A \le B$ means that $B - A$ is positive semidefinite.

  To apply the matrix rank bound in Eq.~\eqref{eq:majorana-rank-cap-app}, we must select columns corresponding to monomials of one degree.
  The standard coordinate basis for $\rho^{\otimes k}$ is indexed by words of $k$ Majorana monomials, one per block.
  For each tuple $\bm n = (n_0, n_1, n_2, n_3)$ of nonnegative integers summing to $k$, define $\Pi_{\bm n}$ as the diagonal matrix in this basis whose entry is one when the word contains exactly $n_j$ monomials from each group $j$, and zero otherwise.
  There are at most $(k + 1)^4$ such projections, and they are mutually orthogonal and sum to the identity.
  The probability of drawing a word with group counts $\bm n$ is
  \begin{equation}
    q_{\bm n} \coloneqq \Tr (\Pi_{\bm n} \rho^{\otimes k} \Pi_{\bm n})
    = \binom{k}{n_0, n_1, n_2, n_3}
    w_0^{n_0} w_1^{n_1} w_2^{n_2} w_3^{n_3},
    \label{eq:weighted-type-probability}
  \end{equation}
  where every factor with exponent zero is one.
  This is the multinomial law obtained by choosing the group independently on each block.
  The local degrees are therefore independent and equal to zero, one, or two with probabilities $w_0$, $w_1$, and $w_2 + w_3$, respectively.
  Each has mean $\mu$ and variance at most one, since its values lie in $[0, 2]$.
  Their sum $p_{\bm n} = n_1 + 2 n_2 + 2 n_3$ therefore has mean $k \mu$ and variance at most $k$.
  Put $b_{\bm n} = \Tr (\Pi_{\bm n} D_k \Pi_{\bm n})$; these numbers sum to $\Tr D_k$.
  Chebyshev's inequality for the degree and Markov's inequality for the remainder, followed by Eq.~\eqref{eq:weighted-spectral-truncation}, give
  \begin{equation}
    \sum_{\abs*{p_{\bm n} / k - \mu} > k^{-1 / 4}} q_{\bm n} + \sum_{b_{\bm n} > k^{-1 / 8} q_{\bm n}} q_{\bm n} \le k^{-1 / 2} + k^{1 / 8} \Tr D_k \le k^{-1 / 2} + v k^{-3 / 8} \longrightarrow 0.
    \label{eq:weighted-type-concentration}
  \end{equation}
  For all sufficiently large $k$, tuples satisfying both $\abs*{p_{\bm n} / k - \mu} \le k^{-1 / 4}$ and $b_{\bm n} \le k^{-1 / 8} q_{\bm n}$ have total probability at least $1 / 2$.
  Among these, choose a tuple $\bm n_k$ maximizing $q_{\bm n}$.
  Since there are at most $(k + 1)^4$ tuples, this choice satisfies
  \begin{equation}
    \abs*{p_{\bm n_k} / k - \mu} \le k^{-1 / 4}, \quad
    q_{\bm n_k} \ge \frac{1}{2 (k + 1)^4}, \quad
    b_{\bm n_k} \le k^{-1 / 8} q_{\bm n_k}.
    \label{eq:weighted-selected-type}
  \end{equation}

  Let $\mc L_k$ index all words with group counts $\bm n_k$, and put $p_k = p_{\bm n_k}$.
  Write $F = F_{\mc L_k} (\ket*{M_t}^{\otimes k})$ for the matrix of unweighted columns.
  Since the group counts are fixed, every selected word has the same weight, which is positive because $q_{\bm n_k} > 0$.
  Multiplying each column of $F$ by the square root of the common weight multiplies its Gram matrix by that weight.
  Restricted to the selected coordinates, $\Pi_{\bm n_k} \rho^{\otimes k} \Pi_{\bm n_k}$ is therefore proportional to $F^\dagger F$, up to diagonal unitary conjugation accounting for fermionic signs.
  The two Gram matrices have traces $q_{\bm n_k}$ and $\abs*{\mc L_k}$, respectively, because the unweighted columns have norm one.
  On the selected coordinate subspace, define
  \begin{equation}
    \widehat\rho_k = q_{\bm n_k}^{-1} \Pi_{\bm n_k} \rho^{\otimes k} \Pi_{\bm n_k}.
    \label{eq:weighted-normalized-shell-gram}
  \end{equation}
  Normalization removes the common weight, so $\widehat\rho_k$ is unitarily conjugate to $F^\dagger F / \abs*{\mc L_k}$.

  Let $\Pi$ be a rank-$R$ orthogonal projection on this subspace, extended by zero on its orthogonal complement, so that $0 \le \Pi \le \Pi_{\bm n_k}$ and $\Tr \Pi = R$.
  The bound on $A_k$ in Eq.~\eqref{eq:weighted-spectral-truncation} gives $\Tr (\Pi A_k) \le R \exp (-k H + k^{3 / 4})$.
  Since $D_k \ge 0$, we also have $\Tr (\Pi D_k) \le \Tr (\Pi_{\bm n_k} D_k \Pi_{\bm n_k}) = b_{\bm n_k}$.
  Using Eq.~\eqref{eq:weighted-normalized-shell-gram} and then Eq.~\eqref{eq:weighted-selected-type}, we obtain
  \begin{equation}
    \Tr (\Pi \widehat\rho_k) = \frac{\Tr (\Pi A_k) + \Tr (\Pi D_k)}{q_{\bm n_k}} \le \frac{R \exp (-k H + k^{3 / 4}) + b_{\bm n_k}}{q_{\bm n_k}} \le \frac{R \exp (-k H + k^{3 / 4})}{q_{\bm n_k}} + k^{-1 / 8}.
    \label{eq:weighted-shell-spectral-mass}
  \end{equation}

  If $\ket*{\eta}$ approximates $\ket*{M_t}^{\otimes k}$ within distance $\delta$, unitarity of the Majorana monomials gives
  \begin{equation}
    \norm*{F_{\mc L_k} (\ket*{M_t}^{\otimes k}) - F_{\mc L_k} (\ket*{\eta})}_{\mathrm F}^2
    = \abs*{\mc L_k} \norm*{\ket*{M_t}^{\otimes k} - \ket*{\eta}}_2^2
    \le \abs*{\mc L_k} \delta^2.
    \label{eq:weighted-action-approximation-error}
  \end{equation}
  Put $R = \rank F_{\mc L_k} (\ket*{\eta})$.
  By Eqs.~\eqref{eq:eckart-young-bound} and~\eqref{eq:weighted-action-approximation-error}, the eigenvalues of $F^\dagger F$ beyond the largest $R$ sum to at most $\abs*{\mc L_k} \delta^2$.
  The eigenvalues of $\widehat\rho_k$ are those of $F^\dagger F$ divided by $\abs*{\mc L_k}$, so the eigenvalues beyond the largest $R$ sum to at most $\delta^2$.
  Since $\Tr \widehat\rho_k = 1$, its largest $R$ eigenvalues sum to at least $1 - \delta^2$.
  Choose $\Pi$ to project onto the corresponding $R$ eigenvectors of $\widehat\rho_k$, so that $\Tr (\Pi \widehat\rho_k)$ is the sum of these eigenvalues.
  Eq.~\eqref{eq:weighted-shell-spectral-mass} then gives
  \begin{equation}
    1 - \delta^2 \le \Tr (\Pi \widehat\rho_k)
    \le \frac{R \exp (-k H + k^{3 / 4})}{q_{\bm n_k}} + k^{-1 / 8}.
    \label{eq:weighted-approximation-spectral-mass}
  \end{equation}
  Using $R \le C_{4 k, p_k} \bchiG (\ket*{\eta})$ from Eq.~\eqref{eq:majorana-rank-cap-app} and minimizing over $\ket*{\eta}$ gives
  \begin{equation}
    \abchiG (\ket*{M_t}^{\otimes k})
    \ge \frac{(1 - \delta^2 - k^{-1 / 8}) q_{\bm n_k}
    \exp (k H - k^{3 / 4})}{C_{4 k, p_k}}.
    \label{eq:weighted-finite-rank-bound}
  \end{equation}

  We now extract the exponential base.
  Since $0 \le p_k \le 2 k$, each binomial coefficient in $C_{4 k, p_k}$ is bounded by the first:
  \begin{equation}
    \binom{4 k}{p_k} \le C_{4 k, p_k} \le (k + 1) \binom{4 k}{p_k}.
    \label{eq:rank-cap-binomial-comparison}
  \end{equation}
  Since $p_k / k \to \mu \in (0, 2)$, the factorial arguments in the binomial coefficient tend to infinity.
  Stirling's formula, applied as in Eq.~\eqref{eq:binomial-stirling-limit}, therefore gives $\log C_{4 k, p_k} / k \to 4 h (\mu / 4)$.
  Define, for all sufficiently large $k$,
  \begin{equation}
    f_t (k)
    \coloneqq \frac{q_{\bm n_k} \exp (k H - k^{3 / 4})}
    {2 C_{4 k, p_k} B_t^k},
    \label{eq:weighted-prefactor-definition}
  \end{equation}
  and set the remaining values to one.
  The selected weights and tuples are independent of $\delta$.
  Eq.~\eqref{eq:weighted-selected-type} gives $\log q_{\bm n_k} / k \to 0$, and maximality gives $H - 4 h (\mu / 4) = \log B_t$, so $\log f_t (k) / k \to 0$.
  For fixed $\delta < 1$, we eventually have $1 - \delta^2 - k^{-1 / 8} \ge (1 - \delta^2) / 2$; Eq.~\eqref{eq:weighted-finite-rank-bound} then proves Eq.~\eqref{eq:weighted-robust-rank-bound}.

  It remains to calculate $B_1$.
  At $t = 1$, the eigenvalues of $\rho_1 (\bm w)$ are the weights of the individual monomials.
  Put $P (z) = 1 + 8 z + 7 z^2$.
  For $z > 0$, the weights
  \begin{equation}
    \bm u = \frac{(1, 8 z, 6 z^2, z^2)}{P (z)}, \quad
    \mu = \frac{z P' (z)}{P (z)}
    \label{eq:weighted-balanced-maximizer}
  \end{equation}
  maximize the entropy at mean degree $\mu$.
  Indeed, for any $\bm w$ with this mean degree,
  \begin{equation}
    H (\rho_1 (\bm w)) = \log P (z) - \mu \log z - \sum_{w_j > 0} w_j \log \frac{w_j}{u_j}.
    \label{eq:weighted-balanced-entropy-comparison}
  \end{equation}
  The sum is nonnegative by convexity of $x \log x$, extended continuously at zero, and vanishes at $\bm w = \bm u$.
  The mean degree increases from zero to two, since
  \begin{equation}
    \odv{}{z} \mleft(\frac{z P' (z)}{P (z)}\mright)
    = \frac{8 + 28 z + 56 z^2}{P (z)^2} > 0.
    \label{eq:weighted-balanced-mean}
  \end{equation}
  Thus each $0 < \mu < 2$ determines a unique $z > 0$, depending smoothly on $\mu$.
  At this mean degree, the maximum of the objective in Eq.~\eqref{eq:weighted-entropy-base} is $\Phi (\mu) = \log P (z) - \mu \log z - 4 h (\mu / 4)$.
  Upon differentiation, the terms involving the derivative of $z$ cancel because $z P' (z) / P (z) = \mu$, giving
  \begin{equation}
    \Phi' (\mu)
    = \log \frac{\mu}{(4 - \mu) z}
    = \log \frac{P' (z)}{4 P (z) - z P' (z)}.
    \label{eq:weighted-balanced-derivative}
  \end{equation}
  The denominator is positive, and the derivative has the sign of
  \begin{equation}
    (1 + z) P' (z) - 4 P (z) = 2 (1 + z) (2 - 7 z).
    \label{eq:weighted-balanced-stationary-equation}
  \end{equation}
  Hence the maximum occurs at $z = 2 / 7$, where $\mu = 4 z / (1 + z) = 8 / 9$.
  At this point, $4 h (\mu / 4) = 4 \log (1 + z) - \mu \log z$, so Eq.~\eqref{eq:weighted-balanced-entropy-comparison} at $\bm w = \bm u$ gives
  \begin{equation}
    B_1 = \frac{P (2 / 7)}{(1 + 2 / 7)^4} = \frac{343}{243}.
    \label{eq:weighted-balanced-base}
  \end{equation}

  The objective in Eq.~\eqref{eq:weighted-entropy-base} is jointly continuous in $t$ and $\bm w$, so its maximum over the compact simplex is continuous in $t$.

  It remains to prove $B_t \le N_t$ and the limit as $t \to 0$.
  For fixed $0 < \delta < 1$, combining Eqs.~\eqref{eq:weighted-robust-rank-bound} and~\eqref{eq:weighted-approximate-upper-bound}, taking logarithms, and dividing by $k$, we obtain
  \begin{equation}
    \log B_t
    \le \log N_t + k^{-1 / 4}
    - \frac{\log (1 - \delta^2) + \log f_t (k)}{k}.
    \label{eq:rate-comparison-logarithmic-bound}
  \end{equation}
  Since $\log f_t (k) / k \to 0$, Eq.~\eqref{eq:rate-comparison-logarithmic-bound} gives $B_t \le N_t$ upon letting $k \to \infty$.
  Finally, $N_t \to 1$ as $t \to 0$, so $1 < B_t \le N_t$ implies $B_t \to 1$, completing the proof of Eq.~\eqref{eq:weighted-base-properties}.
\end{proof}

\begin{wgrblock}
  At $t = 1$, the same Gram-matrix argument allows the sharper choice $f_1 (k) = 1$ in Theorem~\ref{thm:weighted-spectral-bound}.
  This replaces the generic prefactor in Eq.~\eqref{eq:weighted-prefactor-definition}, rather than bounding that particular sequence away from zero.

  Put $P (z) = 1 + 8 z + 7 z^2$, $\zeta = 2 / 7$, and $p_k = \floor*{8 k / 9}$.
  At the maximizing weights in Eq.~\eqref{eq:weighted-balanced-maximizer}, evaluated at $z = \zeta$, each individual monomial of degree $d$ has weight $\zeta^d / P (\zeta)$.
  Instead of fixing all four group counts, let $\mc L_k$ contain all words of total degree $p_k$.
  Every such word has weight $\zeta^{p_k} / P (\zeta)^k$.
  Since $\sigma_1 = 0$, their unweighted Gram matrix is the identity, and $\widehat\rho_k = I / \abs*{\mc L_k}$.
  The sum of its largest $R$ eigenvalues is exactly $R / \abs*{\mc L_k}$.
  Eqs.~\eqref{eq:eckart-young-bound}, \eqref{eq:weighted-action-approximation-error}, and~\eqref{eq:majorana-rank-cap-app} therefore give
  \begin{equation}
    \abchiG (\ket*{M_1}^{\otimes k})
    \ge (1 - \delta^2) \frac{\abs*{\mc L_k}}{C_{4 k, p_k}},
    \quad
    \abs*{\mc L_k} = [z^{p_k}] P (z)^k.
    \label{eq:balanced-full-shell-rank-bound}
  \end{equation}

  The polynomial $P$ has positive coefficients in degrees zero, one, and two, infinite radius of convergence, and $T = \lim_{z \to \infty} z P' (z) / P (z) = 2$, so Theorem~\ref{thm:large-powers-saddle} applies.
  The saddle for $p_k / k \to 8 / 9$ tends to $\zeta$, with curvature $392 / 81$.
  Since $\abs*{p_k / k - 8 / 9} < 1 / k$ and the derivative in Eq.~\eqref{eq:weighted-balanced-mean} is positive at $\zeta$, the moving saddle lies within a constant times $1 / k$ of $\zeta$.
  Taylor's theorem at the moving saddle bounds the change in the exponent per copy by a constant times $1 / k^2$, so replacing it by $\zeta$ does not change the normalized coefficient limit.
  The same argument applies to the binomial saddle in Eq.~\eqref{eq:binomial-stirling-limit}.
  The variances entering the two normalization factors are $\zeta^2 (392 / 81) = 32 / 81$ and $4 (2 / 9) (7 / 9) = 56 / 81$, respectively; the common factor $\zeta^{-p_k}$ cancels.
  Using Eq.~\eqref{eq:weighted-balanced-base}, we obtain
  \begin{equation}
    B_1^{-k} \abs*{\mc L_k} \binom{4 k}{p_k}^{-1} \longrightarrow \sqrt{\frac{56 / 81}{32 / 81}} = \frac{\sqrt{7}}{2}.
    \label{eq:balanced-shell-binomial-limit}
  \end{equation}

  To recover $C_{4 k, p_k}$, divide its summands by $\binom{4 k}{p_k}$ and extend them by zero beyond their range.
  For each fixed $h$, the normalized $h$-th summand tends to $(4 / 49)^h$.
  Each successive nonzero term is at most the preceding term times
  \begin{equation}
    \frac{p_k (p_k - 1)}{(4 k - p_k + 1) (4 k - p_k + 2)}
    \longrightarrow \frac{4}{49} < \frac{1}{2}.
    \label{eq:balanced-rank-cap-geometric-ratio}
  \end{equation}
  Thus the normalized summands are bounded by $2^{-h}$ for all sufficiently large $k$, and dominated convergence gives
  \begin{equation}
    C_{4 k, p_k} \binom{4 k}{p_k}^{-1} \longrightarrow \frac{49}{45}.
    \label{eq:balanced-rank-cap-binomial-limit}
  \end{equation}
  Combining Eqs.~\eqref{eq:balanced-shell-binomial-limit} and~\eqref{eq:balanced-rank-cap-binomial-limit}, we obtain
  \begin{equation}
    B_1^{-k} \frac{\abs*{\mc L_k}}{C_{4 k, p_k}} \longrightarrow \frac{45 \sqrt{7}}{98} > 1.
    \label{eq:balanced-constant-prefactor-limit}
  \end{equation}
  Eq.~\eqref{eq:balanced-full-shell-rank-bound} therefore proves the theorem with $f_1 (k) = 1$ for all sufficiently large $k$, without a spectral truncation loss.
\end{wgrblock}

\begin{wgrblock}
  The base in Eq.~\eqref{eq:weighted-entropy-base} also admits a variational characterization in terms of arbitrary selections of tensor words from the same 16 Majorana monomials.
  More precisely, for each $k$, let $\mc L$ range over nonempty selections of words of a common degree $p$, where $0 \le p \le 2 k$, and let $\rho_{\mc L}$ be their Gram matrix on $\ket*{M_t}^{\otimes k}$, normalized to trace one.
  For $0 < a < 1$, write $R_a (\rho)$ for the smallest number of eigenvalues of $\rho$ whose sum is at least $a$.
  Then
  \begin{equation}
    \inf_{0 < a < 1} \liminf_{k \to \infty} \frac{1}{k}
    \log \max_{p, \mc L} \frac{R_a (\rho_{\mc L})}{C_{4 k, p}}
    = \log B_t.
    \label{eq:variational-spectral-base-equivalence}
  \end{equation}
  The selected group counts in Eq.~\eqref{eq:weighted-selected-type}, together with Eq.~\eqref{eq:weighted-shell-spectral-mass}, give the lower bound for every fixed $a > 0$.

  For the reverse bound, average the outer products of the selected normalized columns.
  The resulting density matrix has the same nonzero eigenvalues as $\rho_{\mc L}$.
  Entropy subadditivity and concavity bound its entropy by $k$ times the entropy of its average one-block marginal.
  Permuting the four odd pairs or the three two-particle pairs, and exchanging the two vectors within each pair, are unitary symmetries of the local column ensemble that preserve degree.
  Averaging over these symmetries makes the weights equal within groups without decreasing entropy; these unitaries are used only for the entropy comparison and need not be Gaussian.
  Consequently, $H (\rho_{\mc L}) \le k H (\rho_t (\bm w))$ for some $\bm w \in \Delta_3$ with $\mu (\bm w) = p / k$.
  Markov's inequality applied to the random variable $ - \log \lambda$, with an eigenvalue sampled with probability $\lambda$, gives $\log R_a (\rho) \le H (\rho) / (1 - a)$; the zero-entropy case is immediate.
  Combining this with $C_{4 k, p} \ge \binom{4 k}{p} \ge \exp (4 k h (p / (4 k))) / (4 k + 1)$ yields the uniform estimate
  \begin{equation}
    \frac{1}{k} \log \frac{R_a (\rho_{\mc L})}{C_{4 k, p}}
    \le \log B_t
    + \frac{a \log 16}{1 - a}
    + \frac{\log (4 k + 1)}{k}.
    \label{eq:variational-spectral-base-upper-bound}
  \end{equation}
  First let $k \to \infty$, then let $a \to 0$, to obtain the reverse inequality in Eq.~\eqref{eq:variational-spectral-base-equivalence}.
  This identifies the asymptotic base, not the finite-copy lower bounds or their prefactors, and asserts no optimality beyond this fixed collection of Majorana monomials.
\end{wgrblock}

\subsubsection{Completion of the proof}
\label{appendix:fermionic-proof-completion}

We can now combine the reduction and the four-mode estimate, keeping the success probability and error tolerance explicit.

\begin{theorem}[Exponential lower bound on fermionic Gaussian rank]
  \label{thm:fermionic-exponential-transfer}
  Let $\ket*{\psi}$ be a normalized pure non-Gaussian state on finitely many fermionic modes.
  There exist $p_\psi \in (0, 1]$ and $t \in (0, 1]$, depending only on $\ket*{\psi}$, such that, for every fixed $0 \le \delta < \delta' < 1$, there are positive integers $m_k$ with $m_k / k \to p_\psi$ for which
  \begin{equation}
    \abchiG (\ket*{\psi}^{\otimes k})
    \ge \bchiG^{(\delta')} (\ket*{M_t}^{\otimes m_k})
    \ge (1 - \delta'^2) f_t (m_k) B_t^{m_k}
    \label{eq:universal-fermionic-rank-transfer}
  \end{equation}
  for all sufficiently large $k$.
  Here $f_t$ and $B_t$ are as in Theorem~\ref{thm:weighted-spectral-bound}; neither depends on the error tolerances.
  In particular, for every fixed $0 \le \delta < 1$,
  \begin{equation}
    \liminf_{k \to \infty} \frac{1}{k} \log \abchiG (\ket*{\psi}^{\otimes k})
    \ge p_\psi \log B_t > 0.
    \label{eq:universal-fermionic-rate}
  \end{equation}
  For four-mode target states, one may take $p_\psi = 1$, $m_k = k$, and $\delta' = \delta$.
  The first inequality in Eq.~\eqref{eq:universal-fermionic-rank-transfer} is then an equality for every $k \ge 1$; the lower bound holds for all sufficiently large $k$.
\end{theorem}

\begin{proof}
  Every pure fermionic state on at most three modes is Gaussian, so the target state has at least four modes.
  Eq.~\eqref{eq:general-state-normal-form-rank-bound} supplies $p_\psi$, $t$ and, for each fixed $\delta < \delta'$, the integers $m_k$.
  Since $m_k / k \to p_\psi > 0$, the number of output copies tends to infinity.
  Theorem~\ref{thm:weighted-spectral-bound} therefore applies for all sufficiently large $k$ and gives Eq.~\eqref{eq:universal-fermionic-rank-transfer}.
  The factor $1 - \delta'^2$ is positive and independent of $k$.
  Since $\log f_t (m) / m \to 0$,
  \begin{equation}
    \lim_{k \to \infty} \frac{\log f_t (m_k)}{k}
    = \lim_{k \to \infty} \frac{m_k}{k}
    \frac{\log f_t (m_k)}{m_k} = 0.
    \label{eq:transferred-subexponential-factor}
  \end{equation}
  Taking logarithms in Eq.~\eqref{eq:universal-fermionic-rank-transfer} and dividing by $k$ now proves Eq.~\eqref{eq:universal-fermionic-rate}.

  For four-mode target states, Proposition~\ref{prop:postselected-four-mode-reduction} has success probability one and gives $U \ket*{\psi} = e^{i \theta} \gamma_1^\epsilon \ket*{M_t}$.
  These transformations and their tensor powers preserve distance and map Gaussian states bijectively onto Gaussian states.
  Hence, for every $k \ge 1$ and $0 \le \delta < 1$,
  \begin{equation}
    \abchiG (\ket*{\psi}^{\otimes k})
    = \abchiG (\ket*{M_t}^{\otimes k}).
    \label{eq:metric-normal-form-app}
  \end{equation}
  Theorem~\ref{thm:weighted-spectral-bound} completes the proof without any loss in copy number or error tolerance.
\end{proof}

Theorem~\ref{thm:fermionic-exponential-transfer} proves Theorem~\ref{thm:fermion-universal-main}.
For Gaussian target states, every tensor power has approximate border Gaussian rank one whenever $0 \le \delta < 1$.
For non-Gaussian target states, Eq.~\eqref{eq:universal-fermionic-rate}, with $\delta' = (1 + \delta) / 2$, gives the eventual lower bound $e^{c_\psi k}$ with $c_\psi = p_\psi \log B_t / 2 > 0$, proving the fermionic part of Theorem~\ref{thm:unified-dichotomy}.
This constant is independent of $\delta$.

\section{Bosonic proofs}
\label{appendix:bosonic-proofs}

\subsection{Proof of Theorem~\ref{thm:boson-universal-main}}
\label{appendix:main-proof}

We couple the target state to a vacuum environment by a Gaussian isometry.
After a photon cutoff, each Gaussian term has bounded Schmidt rank, whereas the environment entropy gives a lower bound on the Schmidt rank needed to approximate the target.
Comparing these bounds yields a Gaussian rank lower bound.
For product states, the entropy grows additively, and the difference between its rate and that of the cutoff bound determines the exponential base.

To apply this argument to an arbitrary non-Gaussian target state, we use Gaussian postselection to obtain states close to a number state in the norm of Eq.~\eqref{eq:photon-moment-norm}.
We then construct an isometry for which all these successful outputs satisfy the same moment and entropy bounds.
Finally, completeness of the instrument ensures a measurement record with both a linear number of successes and sufficiently small approximation error.

\subsubsection{Gaussian rank lower bounds from environment entropy}
\label{appendix:certificate}

\begin{lemma}[Schmidt rank upper bound under a photon cutoff]
  \label{lem:photon-cap}
  Let $n \in \Z_{> 0}$, $\ell$, $p \in \Z_{\ge 0}$, and let $V : \F_{n} \to \F_{n} \otimes \F_{\ell}$ be any Gaussian isometry.
  The two output factors define the split $S : E$, and $\Pi_{\le p}^E$ retains at most $p$ photons in $E$.
  For every pure Gaussian vector $\ket{G} \in \F_{n}$,
  \begin{equation}
    \SR_{S : E} (\Pi_{\le p}^E V \ket{G}) \le \binom{n + p}{p}.
    \label{eq:gaussian-schmidt-cap}
  \end{equation}
  More generally, for $\ket{\eta} \in \F_n$ and an integer $r \ge 0$, if $\bchiG (\ket{\eta}) \le r$, then $\SR_{S : E} (\Pi_{\le p}^E V \ket{\eta}) \le r \binom{n + p}{p}$.
\end{lemma}

\begin{proof}
  Set $\ket{\Psi} \coloneqq V \ket{G}$.
  The environment expansion in Eq.~\eqref{eq:environment-coefficient-vectors} gives
  \begin{equation}
    \SR_{S : E} (\Pi_{\le p}^E \ket{\Psi})
    = \dim \spn \set*{\ket{\Psi_{x}} : \abs{x} \le p}.
    \label{eq:cutoff-coefficient-span}
  \end{equation}
  Any $R$-term product decomposition places the coefficients in a common span of dimension at most $R$.
  Conversely, expanding the coefficients in a basis of their span gives one product term per basis vector.

  Let $F_\Psi (z, w)$ be the Bargmann function of $\ket{\Psi}$, with $z \in \C^{n}$ for the system and $w \in \C^{\ell}$ for the environment.
  Since $\ket{\Psi}$ is Gaussian, Eq.~\eqref{eq:pure-gaussian-bargmann} gives
  \begin{equation}
    F_\Psi (z, w) = g (z) \exp \mleft(z^{\mathsf T} B w + \frac12 w^{\mathsf T} C w + \bm d^{\mathsf T} w\mright),
    \label{eq:joint-gaussian-bargmann}
  \end{equation}
  where $g (z) \coloneqq F_\Psi (z, 0)$ is a nonzero Gaussian function, $B \in \C^{n \times \ell}$, $C \in \C^{\ell \times \ell}$ is symmetric, and $\bm d \in \C^\ell$.

  Define
  \begin{equation}
    P_{x} (z) \coloneqq \sqrt{x!} [w^{x}]
    \exp \mleft(z^{\mathsf T} B w + \frac12 w^{\mathsf T} C w + \bm d^{\mathsf T} w\mright),
    \label{eq:conditional-polynomial-definition}
  \end{equation}
  where the brackets extract the indicated coefficient.
  Each factor of $z$ comes from the bilinear term $z^{\mathsf T} B w$ and is accompanied by a factor of $w$.
  Extracting $w^{x}$ therefore leaves a polynomial $P_{x}$ of degree at most $\abs{x}$.
  By Eq.~\eqref{eq:bargmann-dictionary}, the system coefficient $\ket{\Psi_{x}}$ has Bargmann function $\sqrt{x!} [w^{x}] F_\Psi = g P_{x}$.

  Every retained coefficient function therefore lies in the span of $z^{y} g (z)$ with $\abs{y} \le p$.
  There are $\binom{n + p}{p}$ monomials of total degree at most $p$ in $n$ variables: introducing a slack exponent turns the count into the number of $n + 1$ nonnegative integers summing to $p$.
  Hence
  \begin{equation}
    \dim \spn \set*{z^{y} g (z) : y \in \Z_{\ge 0}^{n}, \abs{y} \le p}
    = \binom{n + p}{p}.
    \label{eq:gaussian-polynomial-span}
  \end{equation}
  Multiplication by the same nonzero Gaussian function preserves their linear independence.
  Equation~\eqref{eq:cutoff-coefficient-span} proves Eq.~\eqref{eq:gaussian-schmidt-cap}.
  The common Gaussian factor $g$ may have components at arbitrarily high occupations; only the dimension of the system span is bounded.

  For a sum of $r$ Gaussian vectors, the system spans together have dimension at most $R \coloneqq r \binom{n + p}{p}$.
  If $\bchiG (\ket{\eta}) \le r$, choose such sums $\ket{\eta_h}$ converging to $\ket{\eta}$.
  Since $V$ is an isometry and the cutoff is an orthogonal projection,
  \begin{equation}
    \norm*{\Pi_{\le p}^E V (\ket{\eta_h} - \ket{\eta})}_2
    \le \norm*{\ket{\eta_h} - \ket{\eta}}_2
    \longrightarrow 0.
    \label{eq:cutoff-border-convergence}
  \end{equation}
  Each retained coefficient converges in norm because coefficient extraction is bounded.
  Every Gram determinant of $R + 1$ coefficients therefore remains zero in the limit, so their span still has dimension at most $R$.
\end{proof}

Lemma~\ref{lem:photon-cap} bounds the Schmidt rank supplied by each Gaussian term.
The environment entropy will bound the rank needed to approximate the target.
For products, $d_n$ in Eq.~\eqref{eq:dn-definition} describes the growth rate of the binomial factor.
The moment assumptions allow a cutoff above the mean total environment photon number whose discarded probability tends to zero uniformly over the target states.

\begin{proposition}[Uniform Gaussian rank lower bound for product states]
  \label{prop:finite-copy-certificate}
  \label{cor:uniform-bosonic-rate}
  Fix $n \in \Z_{> 0}$, $\ell \in \Z_{\ge 0}$, $s > 1$, and a Gaussian isometry $V : \F_n \to \F_n \otimes \F_{\ell}$.
  Let $\K \subset \F_n$ be a nonempty family of normalized vectors, and set $\rho_\phi \coloneqq \Tr_S V \ket{\phi} \bra{\phi} V^\dagger$.
  Assume a common mean bound $\mu > 0$ and bounded variance:
  \begin{equation}
    \sup_{\ket{\phi} \in \K} \Tr (N_E \rho_\phi) \le \mu,
    \quad
    \sup_{\ket{\phi} \in \K} \Var (N_E; \rho_\phi) < \infty.
    \label{eq:uniform-family-data}
  \end{equation}
  Assume also that the R\'enyi entropy exceeds $d_n (\mu)$ by a uniform positive amount:
  \begin{equation}
    c_\K \coloneqq \inf_{\ket{\phi} \in \K}
    H_s (\rho_\phi) - d_n (\mu) > 0.
    \label{eq:compact-rate}
  \end{equation}
  Set $B_\K \coloneqq e^{c_\K}$.
  There is an explicit sequence $f_\K (k) > 0$ independent of $\delta$, with $\log f_\K (k) / k \to 0$.
  For every fixed $0 \le \delta < 1$ and all sufficiently large $k$, uniformly over $\ket{\phi_1}$, $\ldots$, $\ket{\phi_k} \in \K$,
  \begin{equation}
    \abchiG \mleft(\bigotimes_{i = 1}^k \ket{\phi_i}\mright)
    \ge (1 - \delta^2)^{s / (s - 1)} f_\K (k) B_\K^k.
    \label{eq:heterogeneous-exponent}
  \end{equation}
\end{proposition}

\begin{proof}
  We first derive a bound for one target state and an arbitrary cutoff.
  Let $\ket{\phi}$ be normalized, let $p \in \Z_{\ge 0}$, and suppose
  \begin{equation}
    \Tr \mleft[(\id - \Pi_{\le p}^E) \rho_\phi\mright] \le \varepsilon,
    \quad 0 \le \varepsilon < 1 - \delta^2.
    \label{eq:finite-copy-tail-assumption}
  \end{equation}
  Set $\ket{\Psi} \coloneqq V \ket{\phi}$.
  If $r \coloneqq \abchiG (\ket{\phi})$ is infinite, the lower bound in Eq.~\eqref{eq:finite-copy-certificate} is immediate.
  Otherwise choose $\ket{\eta}$ with $\norm*{\ket{\phi} - \ket{\eta}}_2 \le \delta$ and $\bchiG (\ket{\eta}) = r$.
  Lemma~\ref{lem:photon-cap} bounds the Schmidt rank of $\Pi_{\le p}^E V \ket{\eta}$ by $r \binom{n + p}{p}$.
  The discarded part of $\ket{\Psi}$ is orthogonal to the retained approximation error, hence
  \begin{equation}
    \norm*{\ket{\Psi} - \Pi_{\le p}^E V \ket{\eta}}_2^2 = \norm*{(\id - \Pi_{\le p}^E) \ket{\Psi}}_2^2 + \norm*{\Pi_{\le p}^E V (\ket{\phi} - \ket{\eta})}_2^2 \le \varepsilon + \delta^2.
    \label{eq:orthogonal-projection-error}
  \end{equation}
  Let $R \coloneqq \SR_{S : E} (\Pi_{\le p}^E V \ket{\eta})$, and list the eigenvalues of $\rho_\phi$ in decreasing order as $\lambda_j$.
  The best approximation to $\ket{\Psi}$ by a vector of Schmidt rank at most $R$ retains its largest $R$ Schmidt coefficients; the squared error is $\sum_{j > R} \lambda_j$.
  Thus Eq.~\eqref{eq:orthogonal-projection-error} implies $\sum_{j = 1}^R \lambda_j \ge 1 - \delta^2 - \varepsilon$.
  Applying H\"older's inequality to the vectors $(\lambda_1, \ldots, \lambda_R)$ and $(1, \ldots, 1)$, with exponents $s$ and $s / (s - 1)$, gives
  \begin{equation}
    1 - \delta^2 - \varepsilon
    \le \sum_{j = 1}^R \lambda_j
    \le R^{(s - 1) / s} \mleft(\sum_j \lambda_j^s\mright)^{1 / s}.
    \label{eq:schmidt-holder-bound}
  \end{equation}
  Substitute $\sum_j \lambda_j^s = \exp ((1 - s) H_s (\rho_\phi))$ and $R \le r \binom{n + p}{p}$ into Eq.~\eqref{eq:schmidt-holder-bound} and solve for $r$:
  \begin{equation}
    \abchiG (\ket{\phi})
    \ge (1 - \delta^2 - \varepsilon)^{s / (s - 1)}
    \exp (H_s (\rho_\phi)) \binom{n + p}{p}^{-1}.
    \label{eq:finite-copy-certificate}
  \end{equation}

  We now choose a cutoff whose discarded probability tends to zero and whose binomial factor has exponential growth rate $d_n (\mu)$.
  Choose $v > 0$ such that $\Var (N_E; \rho_\phi) \le v$ for every $\ket{\phi} \in \K$, as permitted by Eq.~\eqref{eq:uniform-family-data}.
  Apply $V^{\otimes k}$, grouping all $n k$ system outputs before the $\ell k$ environment outputs.
  For the product state $\bigotimes_{i = 1}^k \ket{\phi_i}$, the environment state is $\bigotimes_{i = 1}^k \rho_i$, where $\rho_i \coloneqq \rho_{\phi_i}$.
  The identity $\Tr [(\bigotimes_i \rho_i)^s] = \prod_i \Tr (\rho_i^s)$ makes R\'enyi entropy additive:
  \begin{equation}
    H_s \mleft(\bigotimes_{i = 1}^k \rho_i\mright) = \sum_{i = 1}^k H_s (\rho_i).
    \label{eq:product-environment-renyi-moment}
  \end{equation}
  The photon counts are independent, so their total has mean at most $k \mu$ and variance at most $k v$.
  We choose a cutoff margin $k^{2 / 3}$, which grows faster than the standard deviation but is $o (k)$.
  Set
  \begin{equation}
    p_k \coloneqq \ceil*{\mu k + k^{2 / 3}}.
    \label{eq:uniform-photon-cutoff}
  \end{equation}
  Exceeding this cutoff requires a deviation from the mean greater than $k^{2 / 3}$.
  Chebyshev's inequality gives
  \begin{equation}
    \Tr \mleft[(\id - \Pi_{\le p_k}^E) \bigotimes_{i = 1}^k \rho_i\mright]
    \le \frac{k v}{(k^{2 / 3})^2} = v k^{-1 / 3}.
    \label{eq:cutoff-tail-chebyshev}
  \end{equation}
  For all sufficiently large $k$, Eq.~\eqref{eq:finite-copy-certificate}, applied on $n k$ system modes, gives
  \begin{equation}
    \abchiG \mleft(\bigotimes_{i = 1}^k \ket{\phi_i}\mright)
    \ge (1 - \delta^2 - v k^{-1 / 3})^{s / (s - 1)}
    \exp \mleft(\sum_{i = 1}^k H_s (\rho_i)\mright) \binom{n k + p_k}{p_k}^{-1}.
    \label{eq:uniform-finite-copy-bound}
  \end{equation}

  The logarithmic form of Stirling's estimate, $\log (h!) = h \log h - h + O (\log (h + 1))$, applied to the three factorials in the binomial coefficient gives $\log \binom{n k + p_k}{p_k} = k d_n (p_k / k) + O (\log k)$ when $p_k / k$ stays bounded.
  Since $p_k / k \to \mu > 0$, continuity of $d_n$ gives
  \begin{equation}
    \binom{n k + p_k}{p_k} = \exp (k d_n (\mu) + o (k)).
    \label{eq:binomial-entropy-rate}
  \end{equation}
  To write the lower bound as an exponential factor $B_\K^k$ times a subexponential factor, define
  \begin{equation}
    f_\K (k) \coloneqq
    2^{-s / (s - 1)}\,\exp (k d_n (\mu))
    \binom{n k + \ceil*{\mu k + k^{2 / 3}}}{\ceil*{\mu k + k^{2 / 3}}}^{-1}.
    \label{eq:bosonic-subexponential-factor}
  \end{equation}
  Equation~\eqref{eq:binomial-entropy-rate} gives $\log f_\K (k) / k \to 0$.
  For all sufficiently large $k$, the tail bound satisfies $v k^{-1 / 3} \le (1 - \delta^2) / 2$.
  Substituting this inequality and $\sum_i H_s (\rho_i) \ge k[c_\K + d_n (\mu)]$ into Eq.~\eqref{eq:uniform-finite-copy-bound} proves Eq.~\eqref{eq:heterogeneous-exponent}.
  The cutoff and tail bound do not depend on the chosen states, so the threshold in $k$ is uniform over all products from $\K$.
\end{proof}

\subsubsection{Gaussian postselection near a number state}
\label{appendix:regularization}

For integers $n, m \ge 1$, write $\ket{\xi_{n, m}} \coloneqq \ket{m} \otimes \ket{0}^{\otimes (n - 1)} \in \F_n$.
We use the fact that every non-Gaussian pure state is orthogonal to some coherent state.
Indeed, let $\ket{u} \in \F_n$ be normalized and suppose its Bargmann function has no zeros.
Since $F_u$ is entire and nowhere zero on the simply connected domain $\C^n$, it admits an entire logarithm $g : \C^n \to \C$, so $F_u = \exp g$.
Fix $v \in \C^n$ with $\abs{v} = 1$.
The function $w \mapsto g (w v)$ is entire on $\C$, so its Taylor expansion $g (w v) = \sum_{j \ge 0} \gamma_j w^j$ converges for every $w \in \C$; the coefficients $\gamma_j \in \C$ depend on $v$.
For every $R > 0$, Eq.~\eqref{eq:bargmann-growth} implies that on $\abs{w} = R$ the function $R^2 / 2 - \operatorname{Re} g (w v)$ is nonnegative.
Its average over this circle is $R^2 / 2 - \operatorname{Re} \gamma_0$, and its $j$th Fourier coefficient is $-\gamma_j R^j / 2$ for $j \ge 1$.
The triangle inequality gives
\begin{equation}
  \abs{\gamma_j}
  R^j / 2 \le R^2 / 2 - \operatorname{Re} \gamma_0 \quad (j \ge 1).
  \label{eq:zero-free-bargmann-coefficients}
\end{equation}
Letting $R \to \infty$ in Eq.~\eqref{eq:zero-free-bargmann-coefficients} gives $\gamma_j = 0$ for $j > 2$.
Each homogeneous Taylor term of $g$ of degree greater than two therefore vanishes on every unit vector, and hence identically.
Thus $g (z) = c + \bm b^{\mathsf T} z + z^{\mathsf T} A z / 2$, with $c \in \C$, $\bm b \in \C^n$, and $A = A^{\mathsf T} \in \C^{n \times n}$.
Normalizability and Eq.~\eqref{eq:pure-gaussian-bargmann} then imply that $\ket{u}$ is Gaussian; see also Ref.~\cite[Proposition~5.3]{gröchenig2018zeroswignerdistributionshorttime}.
For a non-Gaussian input, $F_u$ must therefore have a zero, which gives an orthogonal coherent state by Eq.~\eqref{eq:bargmann-function}.

\begin{lemma}[Number-state approximation by Gaussian postselection]
  \label{lem:number-state-postselection}
  Let $\ket{\psi} \in \F_n$ be a normalized non-Gaussian state, with $n \ge 1$.
  There is an integer $m \ge 1$ such that, for every $\varepsilon > 0$, there is a complete Gaussian instrument with a measurable success set $B$ of positive probability.
  Its Kraus operators $K_\omega$ take values in $\F_n$ for $\omega \in B$ and in $\C$ otherwise.
  For every $\omega \in B$, the vector $K_\omega \ket{\psi}$ is nonzero and its normalization satisfies
  \begin{equation}
    \ket{\phi_\omega} = \frac{K_\omega \ket{\psi}}{\norm{K_\omega \ket{\psi}}_2},
    \quad \norm{\ket{\phi_\omega} - \ket{\xi_{n, m}}}_D < \varepsilon.
    \label{eq:number-state-approximation}
  \end{equation}
\end{lemma}

\begin{proof}
  Choose a coherent state $\ket{\bm\alpha}$ orthogonal to $\ket{\psi}$.
  The displaced state $D (-\bm\alpha) \ket{\psi}$ has zero vacuum amplitude because $\bra{0^n} D (-\bm\alpha) \ket{\psi} = \braket{\bm\alpha}{\psi} = 0$.
  Let $m \ge 1$ be the lowest occupied photon number of the displaced state.
  The vectors $(\sum_j v_j a_j^\dagger)^m \ket{0^n}$, with $v \in \C^n$, span the $m$-photon sector.
  Indeed, their expansions are $m! \sum_{\abs{x} = m} v^x \ket{x} / \sqrt{x!}$; a vector orthogonal to all of them would give a polynomial in $v$ that vanishes identically, forcing all its occupation coefficients to vanish.
  Choose $v$ so that $(\sum_j v_j a_j^\dagger)^m \ket{0^n}$ has nonzero overlap with the displaced state.
  Normalize $v$ and choose a number-preserving Gaussian unitary sending $\sum_j v_j a_j^\dagger$ to $a_1^\dagger$.
  This unitary fixes the vacuum and makes the coefficient of $\ket{m, 0, \ldots, 0}$ nonzero.
  Combining these unitaries into $U_0$, choosing its overall phase so that $c_m > 0$, and projecting the other modes onto vacuum gives
  \begin{equation}
    \ket{\eta} = (\id \otimes \bra{0}^{\otimes (n - 1)}) U_0 \ket{\psi}
    = \sum_{j \ge m} c_j \ket{j}, \quad c_m > 0.
    \label{eq:lowest-number-coefficient}
  \end{equation}
  In the one-mode formulas below, $N_{\mathrm{tot}} = a^\dagger a$.
  Attenuation isolates the lowest occupation as $\kappa \to 0$.
  Since $\kappa^{N_{\mathrm{tot}}} \ket{j} = \kappa^j \ket{j}$, the definition of $\norm{\cdot }_D$ in Eq.~\eqref{eq:photon-moment-norm} gives, for $0 < \kappa \le 1 / 2$,
  \begin{equation}
    \norm{\kappa^{-m} \kappa^{N_{\mathrm{tot}}} \ket{\eta} - c_m \ket{m}}_D^2 = \sum_{j > m} (j + 1)^2 \abs{c_j}^2 \kappa^{2 (j - m)} \le (m + 2)^2 \kappa^2 \norm{\ket{\eta}}_2^2 \longrightarrow 0.
    \label{eq:attenuation-number-state-limit}
  \end{equation}
  For $j \ge m + 1$, the terms $(j + 1) \kappa^{j - m - 1}$ start at $m + 2$ and decrease, since their successive ratio is $\kappa (j + 2) / (j + 1) < 1$; this proves the inequality.
  Since $\norm{\cdot }_2 \le \norm{\cdot }_D$, the norm $\norm{\kappa^{-m} \kappa^{N_{\mathrm{tot}}} \ket{\eta}}_2$ tends to $c_m > 0$; dividing by it gives convergence of the normalized output to $\ket{m}$ in $\norm{\cdot }_D$.

  We implement attenuation by a beam splitter with a vacuum ancilla and heterodyne detection.
  Let $r = \sqrt{1 - \kappa^2}$, and choose its unitary $U_\kappa$ so that
  \begin{equation}
    U_\kappa \ket{j, 0} = \sum_{l = 0}^j \sqrt{\binom{j}{l}} \kappa^l r^{j - l} \ket{l, j - l}.
    \label{eq:attenuation-beam-splitter}
  \end{equation}
  The Kraus operator $T_\alpha : \F_1 \to \F_1$ for outcome $\alpha \in \C$ is
  \begin{equation}
    T_\alpha = (\id \otimes \bra{\alpha}) U_\kappa (\id \otimes \ket{0}), \quad T_0 = \kappa^{N_{\mathrm{tot}}}, \quad \int_\C T_\alpha^\dagger T_\alpha \frac{\odif[order=2]{\alpha}}{\pi} = \id.
    \label{eq:attenuation-instrument}
  \end{equation}
  Completeness follows from Eq.~\eqref{eq:coherent-resolution}.

  The attenuation operator can also be written in terms of displacements:
  \begin{equation}
    T_\alpha = D (\kappa \alpha / r) \kappa^{N_{\mathrm{tot}}} D (-\alpha / r).
    \label{eq:attenuation-displacement-identity}
  \end{equation}
  Indeed, the beam splitter transforms displacements by $U_\kappa (D (-\alpha / r) \otimes I) = (D (-\kappa \alpha / r) \otimes D (-\alpha)) U_\kappa$.
  Taking vacuum matrix elements on the environment and using $\bra{0} D (-\alpha) = \bra{\alpha}$ gives $\kappa^{N_{\mathrm{tot}}} D (-\alpha / r) = D (-\kappa \alpha / r) T_\alpha$, which proves Eq.~\eqref{eq:attenuation-displacement-identity}.
  The displacement identity
  \begin{equation}
    D (\alpha)^\dagger N_{\mathrm{tot}} D (\alpha)
    = N_{\mathrm{tot}} + \alpha a^\dagger + \overline\alpha a + \abs{\alpha}^2
    \label{eq:displaced-number-operator}
  \end{equation}
  and the bounds $\norm{a \ket{u}}_2, \norm{a^\dagger \ket{u}}_2 \le \norm{\ket{u}}_D$ give
  $\norm{D (\alpha) \ket{u}}_D \le (1 + \abs{\alpha})^2 \norm{\ket{u}}_D$.
  Displaced finite occupation sums have finite $\norm{\cdot }_D$ by Eq.~\eqref{eq:coherent-occupation-expansion}, so the calculation is valid on these vectors; occupation truncation and closedness of $N_{\mathrm{tot}}$ extend the bound to every vector of finite $\norm{\cdot }_D$.
  Multiplying Eq.~\eqref{eq:displaced-number-operator} by $D (\alpha)$ and using its strong continuity shows that $\alpha \mapsto D (\alpha) \ket{u}$ is continuous in $\norm{\cdot }_D$ for each such vector.

  For $\kappa \le 1 / 2$, the bound $(j + 1) \kappa^j \le 1$ gives $\norm{\kappa^{N_{\mathrm{tot}}} \ket{u}}_D \le \norm{\ket{u}}_2$ for every $\ket{u} \in \F_1$.
  Equation~\eqref{eq:attenuation-displacement-identity} therefore gives $\norm{T_\alpha \ket{u}}_D \le (1 + \kappa \abs{\alpha} / r)^2 \norm{\ket{u}}_2$.
  The attenuation bound, together with strong continuity and local boundedness of the displacements, implies that $\alpha \mapsto T_\alpha \ket{u}$ is continuous in $\norm{\cdot }_D$ for every $\ket{u} \in \F_1$.

  After measuring the other $n - 1$ modes, append vacuum modes to restore the output to $\F_n$; this leaves $\norm{\cdot }_D$ unchanged.
  For a measurement record $\omega = (\beta, \alpha) \in \C^{n - 1} \times \C$, the conditional operator $K_\omega : \F_n \to \F_n$ is
  \begin{equation}
    K_\omega \ket{u} = [T_\alpha (\id \otimes \bra{\beta}) U_0 \ket{u}] \otimes \ket{0}^{\otimes (n - 1)}.
    \label{eq:number-state-kraus-operators}
  \end{equation}
  For $n = 1$, $\C^0$ is a point of measure one and $\id \otimes \bra\beta = \id$.
  By Eq.~\eqref{eq:attenuation-number-state-limit}, choose $\kappa$ so that the normalized output at $\omega = 0$ is within $\varepsilon / 2$ of $\ket{\xi_{n, m}}$ in $\norm{\cdot }_D$.
  Equation~\eqref{eq:attenuation-instrument} and the resolution of the identity by coherent states, integrated successively, give
  \begin{equation}
    \int_{\C^n} \norm{K_\omega \ket{u}}_2^2 \frac{\odif[order={2 n}]{\omega}}{\pi^n} = \norm{\ket{u}}_2^2.
    \label{eq:conditional-completeness}
  \end{equation}
  Tonelli's theorem justifies the integrations; polarization gives completeness.
  Coherent-state bras are norm-continuous, and the local bound on $T_\alpha$ therefore makes $\omega \mapsto K_\omega \ket{\psi}$ jointly continuous in $\norm{\cdot }_D$.
  Since $K_0 \ket{\psi} \ne 0$, there is a closed ball $B$ of positive radius about zero on which $\norm{K_\omega \ket{\psi}}_2 \ge \norm{K_0 \ket{\psi}}_2 / 2$ and all normalized outputs lie within $\varepsilon$ of $\ket{\xi_{n, m}}$ in $\norm{\cdot }_D$.
  The lower bound on the output norm and the positive volume of $B$ give
  \begin{equation}
    \int_B \norm{K_\omega \ket{\psi}}_2^2 \frac{\odif[order={2 n}]{\omega}}{\pi^n} > 0.
    \label{eq:number-state-success-probability}
  \end{equation}
  For $\omega \notin B$, measure the remaining modes by heterodyne detection, with outcome $\zeta \in \C^n$.
  The complete instrument is indexed by
  \begin{equation}
    \Omega = B \sqcup ((\C^n \setminus B) \times \C^n), \quad
    K_{\omega, \zeta} = \bra{\zeta} K_\omega \quad (\omega \notin B),
    \label{eq:number-state-failure-operators}
  \end{equation}
  retaining $K_\omega$ for $\omega \in B$.
  The measures on the two components are $\pi^{-n}\odif[order={2 n}]{\omega}$ and $\pi^{-2 n}\odif[order={2 n}]{\omega}\,\odif[order={2 n}]{\zeta}$.
  The coherent-state resolution of the identity and Eq.~\eqref{eq:conditional-completeness} prove completeness.
  Every Kraus operator is Gaussian, bounded, and weakly measurable.
\end{proof}

\subsubsection{An entropy gap near a number state}
\label{appendix:entropy}

\begin{lemma}[Uniform entropy gap]
  \label{lem:number-state-entropy-gap}
  For all integers $n, m \ge 1$, there are a Gaussian isometry $V : \F_n \to \F_n \otimes \F_{n + 1}$ and constants $s > 1$, $\mu, c, C, \varepsilon > 0$ such that every normalized $\ket{\phi} \in \F_n$ with $\norm{\ket{\phi} - \ket{\xi_{n, m}}}_D < \varepsilon$ satisfies
  \begin{equation}
    \Tr (N_E \rho_\phi) < \mu, \quad
    \Tr (N_E^2 \rho_\phi) \le C, \quad
    H_s (\rho_\phi) > d_n (\mu) + c.
    \label{eq:number-state-uniform-gap}
  \end{equation}
\end{lemma}

\begin{proof}
  Couple the occupied mode to two vacuum environment modes by $S_r B_\theta$, as defined in Eq.~\eqref{eq:beam-splitter-and-squeezer}.
  We will choose $t > 0$ sufficiently small.

  For the input $\ket{m, 0, 0}$, Eqs.~\eqref{eq:beam-splitter-expansion} and~\eqref{eq:two-mode-squeezer-expansion} show that the environment vector conditional on system occupation $k$ lies in the span of $\ket{m - j, k - j}$, with $0 \le j \le \min (m, k)$.
  The difference between the two environment occupations is $m - k$, so the conditional vectors for distinct $k$ are orthogonal.
  Their squared norms are therefore the nonzero environment eigenvalues, indexed by $k \in \Z_{\ge 0}$:
  \begin{equation}
    \lambda_k (t) = \frac1{(1 + t)^{k + 1}}
    \sum_{j = 0}^{\min (m, k)}\binom m j\binom k j
    (1 - t)^j t^{m + k - 2 j}.
    \label{eq:number-state-environment-spectrum}
  \end{equation}
  Merge all but $\lambda_m$ and $\lambda_{m - 1}$ into one entry:
  \begin{equation}
    q_t = (\lambda_m, \lambda_{m - 1}, 1 - \lambda_m - \lambda_{m - 1}), \quad H_s (\rho_E) \ge H_s (q_t),
    \label{eq:coarse-renyi-bound}
  \end{equation}
  where $\rho_E = \Tr_S S_r B_\theta \ket{m, 0, 0} \bra{m, 0, 0}
  B_\theta^\dagger S_r^\dagger$.
  Merging decreases R\'enyi entropy for $s > 1$, since $\sum_j x_j^s \le (\sum_j x_j)^s$.
  The terms $j = m$ and $j = m - 1$ in Eq.~\eqref{eq:number-state-environment-spectrum} give $\lambda_m = 1 - (2 m + 1) t + O (t^2)$ and $\lambda_{m - 1} = m t + O (t^2)$, respectively; the remaining terms in these two sums contribute $O (t^2)$.
  Thus the third entry of $q_t$ is $(m + 1) t + O (t^2)$.
  Substitution into $H (q_t) - d_1 ((2 m + 1) t)$ cancels the terms in $t \log t$ and gives, as $t \to 0$,
  \begin{equation}
    \frac{H (q_t) - d_1 ((2 m + 1) t)} t
    = (2 m + 1) \log (2 m + 1) - m \log m - (m + 1) \log (m + 1) + O (t \abs{\log t}).
    \label{eq:number-state-entropy-gap-limit}
  \end{equation}
  The limit is positive because it equals $(2 m + 1)$ times the binary entropy of $m / (2 m + 1) \in (0, 1)$.

  The mean total environment photon number is $m t + t[m (1 - t) + 1] = (2 m + 1) t - m t^2$: the first term is the loss, and the second follows from $S_r^\dagger b_2 S_r = b_2\cosh r + a^\dagger\sinh r$.
  Couple each remaining vacuum system mode to a new environment mode by a squeezer whose thermal environment has mean photon number $(2 m + 1) t$.
  This state $\tau$ has entropy $d_1 ((2 m + 1) t)$, and $\Tr \tau^s = [(1 + (2 m + 1) t)^s - ((2 m + 1) t)^s]^{-1}$, by the geometric series in Eq.~\eqref{eq:two-mode-squeezer-expansion}.
  Thus the product isometry satisfies
  \begin{equation}
    \rho_{\xi_{n, m}} = \rho_E \otimes \tau^{\otimes (n - 1)}, \quad
    \Tr (N_E \rho_{\xi_{n, m}}) = \mu - m t^2, \quad \mu = n (2 m + 1) t.
    \label{eq:number-state-reference-moments}
  \end{equation}
  Fix $t > 0$ so small that Eq.~\eqref{eq:number-state-entropy-gap-limit} gives $H (q_t) > d_1 ((2 m + 1) t)$.
  Since $d_n (\mu) = n d_1 ((2 m + 1) t)$, choose $s > 1$ sufficiently close to one to obtain
  \begin{equation}
    H_s (\rho_{\xi_{n, m}}) \ge H_s (q_t) + (n - 1) H_s (\tau) > d_n (\mu).
    \label{eq:number-state-reference-gap}
  \end{equation}
  Here $H_s (q_t) \to H (q_t)$ and $H_s (\tau) \to d_1 ((2 m + 1) t)$ as $s \to 1$.

  Write $V \ket{\phi} = U (\ket{\phi} \otimes \ket{0}^{\otimes (n + 1)})$, where $U$ is Gaussian.
  Then $U^\dagger N_E U$ is quadratic in creation and annihilation operators.
  Each product of at most two ladder operators sends distinct occupation vectors to orthogonal vectors and has occupation-basis coefficient at most $k + 2$ on total occupation $k$.
  Since $k + 2 \le 2 (k + 1)$, its value on $\ket{\phi} \otimes \ket{0}^{\otimes (n + 1)}$ has norm at most $2 \norm{\ket{\phi}}_D$.
  Summing the finitely many terms gives
  \begin{equation}
    \norm{N_E V \ket{\phi}}_2 \le M \norm{\ket{\phi}}_D
    \label{eq:gaussian-isometry-moment-bound}
  \end{equation}
  for a constant $M > 0$ depending only on $V$.
  For finite occupation sums, Eqs.~\eqref{eq:beam-splitter-expansion} and~\eqref{eq:two-mode-squeezer-expansion} give geometrically decaying occupation tails, so the conjugation and moment calculations are valid.
  Finite occupation truncations converge in $\norm{\cdot }_D$; closedness of $N_E$ therefore extends Eq.~\eqref{eq:gaussian-isometry-moment-bound} to every vector of finite $\norm{\cdot }_D$.
  If $\norm{\ket{\phi}}_2 = 1$ and $\norm{\ket{\phi} - \ket{\xi_{n, m}}}_D < \varepsilon \le 1$, then $\norm{\ket{\phi}}_D \le m + 2$, since $\norm{\ket{\xi_{n, m}}}_D = m + 1$.
  Expanding the difference of the mean photon numbers and applying Cauchy--Schwarz and Eq.~\eqref{eq:gaussian-isometry-moment-bound} gives
  \begin{equation}
    \abs{\Tr (N_E \rho_\phi) - \Tr (N_E \rho_{\xi_{n, m}})} \le M (m + 3) \varepsilon, \quad
    \Tr (N_E^2 \rho_\phi) = \norm{N_E V \ket{\phi}}_2^2 \le M^2 (m + 2)^2.
    \label{eq:number-state-moment-continuity}
  \end{equation}
  These trace identities follow by finite spectral projections of $N_E$.
  The Schatten-norm triangle inequality and partial-trace contractivity give
  \begin{equation}
    \abs{\norm{\rho_\phi}_s - \norm{\rho_{\xi_{n, m}}}_s} \le \norm{\rho_\phi - \rho_{\xi_{n, m}}}_1 \le 2 \norm{\ket{\phi} - \ket{\xi_{n, m}}}_2.
    \label{eq:environment-schatten-continuity}
  \end{equation}
  Equation~\eqref{eq:environment-schatten-continuity} and the entropy formula in Eq.~\eqref{eq:bosonic-renyi-entropy} show that $H_s (\rho_\phi)$ is continuous.
  Choose $c > 0$ below the gap in Eq.~\eqref{eq:number-state-reference-gap}, then choose $0 < \varepsilon \le 1$ so that $M (m + 3) \varepsilon < m t^2$ and $H_s (\rho_\phi) > d_n (\mu) + c$ throughout the $\varepsilon$-neighborhood.
  Equations~\eqref{eq:number-state-reference-moments} and~\eqref{eq:number-state-moment-continuity} give the mean bound and the uniform second-moment bound with $C = M^2 (m + 2)^2$.
\end{proof}

\label{sec:local-core}

\begin{proposition}[Uniform entropy inequality on a positive-probability set]
  \label{prop:local-core}
  Let $\ket{\psi} \in \F_n$ be normalized and non-Gaussian, with $n \ge 1$.
  There is a complete Gaussian instrument on a $\sigma$-finite standard Borel space $(\Omega, \nu)$, with bounded Gaussian Kraus operators, and a measurable success set $\Omega_{\mathrm{suc}}$ of probability $p_\psi > 0$, such that
  \begin{equation}
    K_\omega \ket{\psi} =
    \begin{cases}
      c_\omega \ket{\phi_\omega} \in \F_n, & \omega \in \Omega_{\mathrm{suc}},    \\
      d_\omega \in \C,                     & \omega \notin \Omega_{\mathrm{suc}},
    \end{cases}
    \quad c_\omega \ne 0, \quad \norm{\ket{\phi_\omega}}_2 = 1.
    \label{eq:local-core-outputs}
  \end{equation}
  The set of normalized outputs $\K_\psi = \set{\ket{\phi_\omega} : \omega \in \Omega_{\mathrm{suc}}}$ admits a common Gaussian isometry $V : \F_n \to \F_n \otimes \F_{n + 1}$ and constants $s > 1$, $\mu, c, C > 0$ such that
  \begin{equation}
    \sup_{\ket{\phi} \in \K_\psi} \Tr (N_E \rho_\phi) \le \mu, \quad
    \sup_{\ket{\phi} \in \K_\psi} \Var (N_E; \rho_\phi) \le C, \quad
    \inf_{\ket{\phi} \in \K_\psi}
    H_s (\rho_\phi) \ge d_n (\mu) + c.
    \label{eq:local-core-uniform-bounds}
  \end{equation}
  The isometry and constants may depend on $\ket{\psi}$, but are common to all successful outcomes.
\end{proposition}

\begin{proof}
  Choose $m$ from Lemma~\ref{lem:number-state-postselection}.
  Lemma~\ref{lem:number-state-entropy-gap} supplies the isometry $V$ and constants $s$, $\mu$, $c$, $C$, and $\varepsilon$.
  Apply Lemma~\ref{lem:number-state-postselection} with this $\varepsilon$, and set $\Omega_{\mathrm{suc}} = B$.
  The resulting instrument has the outputs in Eq.~\eqref{eq:local-core-outputs}, with $c_\omega = \norm{K_\omega \ket{\psi}}_2 > 0$ on $B$.
  Equation~\eqref{eq:number-state-uniform-gap} gives Eq.~\eqref{eq:local-core-uniform-bounds}, since variance is at most the second moment.
\end{proof}

\subsubsection{Completion of the proof}
\label{appendix:lifting}

Start from an approximant with error at most $\delta$, where $0 \le \delta < 1$, and apply the instrument separately to each copy of the target and to the corresponding blocks of the approximant.
For each record with nonzero target output, rescale both outputs by the same scalar so that the target becomes the normalized product of successful states.
Although this may amplify the error, completeness bounds its mean square under the target's outcome distribution by $\delta^2$.
For any fixed $\delta < \delta' < 1$, records with rescaled error at most $\delta'$ therefore have probability at least $1 - \delta^2 / \delta'^2 > 0$, uniformly in $k$.
Since the success count concentrates near $p_\psi k$, at least one such record also has a typical success count for all sufficiently large $k$.

For a record $\bm\omega = (\omega_1, \ldots, \omega_k) \in \Omega^k$, define
\begin{equation}
  T_{\bm\omega} \coloneqq K_{\omega_1} \otimes \cdots \otimes K_{\omega_k}, \quad m (\bm\omega) \coloneqq \abs*{\set*{j \in [k] : \omega_j \in \Omega_{\mathrm{suc}}}}.
  \label{eq:bosonic-record-map}
\end{equation}
The map $T_{\bm\omega} : \F_{n k} \to \F_{n m (\bm\omega)}$ retains the successful blocks in their original order.
Completeness gives
\begin{equation}
  \int_{\Omega^k} \norm*{T_{\bm\omega} \ket{\xi}}_2^2 \odif{\nu^{\otimes k}} (\bm\omega)
  = \norm*{\ket{\xi}}_2^2
  \quad (\ket{\xi} \in \F_{n k}).
  \label{eq:continuous-product-completeness}
\end{equation}

Expand $\ket{\xi} = \sum_j \ket{\xi_j} \otimes \ket{e_j}$ in an orthonormal basis of all blocks except the first.
Orthogonality and one-copy completeness give
\begin{equation}
  \int_\Omega \norm*{(K_\omega \otimes \id) \ket{\xi}}_2^2 \odif{\nu} (\omega) = \sum_j \int_\Omega \norm*{K_\omega \ket{\xi_j}}_2^2 \odif{\nu} (\omega) = \sum_j \norm*{\ket{\xi_j}}_2^2 = \norm*{\ket{\xi}}_2^2.
  \label{eq:completeness-with-spectators}
\end{equation}
All summands are nonnegative, so Tonelli's theorem permits the sum and integral to be interchanged.
Applying the identity block by block proves Eq.~\eqref{eq:continuous-product-completeness}.
The integrands are measurable: finite sums of product vectors give measurable outputs, and boundedness of each outcome map passes measurability to their pointwise limits.

The tensor product $T_{\bm\omega}$ is a Gaussian operator, so it maps each Gaussian summand to a Gaussian vector or zero.
An $r$-term Gaussian decomposition therefore gives an output decomposition with at most $r$ terms.
For border rank at most $r$, approximate the input in norm by sums of at most $r$ Gaussian vectors.
Boundedness of $T_{\bm\omega}$ makes the corresponding outputs converge in norm, giving
\begin{equation}
  \bchiG (T_{\bm\omega} \ket{\xi}) \le \bchiG (\ket{\xi}).
  \label{eq:gaussian-map-rank}
\end{equation}

We can now transfer the bound from successful outputs to the original target state.
With the data of Proposition~\ref{prop:local-core}, define
\begin{equation}
  B_\psi \coloneqq B_{\K_\psi} = \exp \mleft(\inf_{\ket{\phi} \in \K_\psi} H_s (\rho_\phi) - d_n (\mu)\mright).
  \label{eq:bosonic-input-base}
\end{equation}
All quantities in this expression depend only on the target state, and the uniform entropy gap ensures $B_\psi > 1$.

\begin{theorem}[Exponential lower bound on bosonic Gaussian rank]
  \label{thm:bosonic-exponential-transfer}
  Let $\ket{\psi} \in \F_n$ be normalized and non-Gaussian.
  Choose the instrument and successful family $\K_\psi$ of Proposition~\ref{prop:local-core}, with success probability $p_\psi > 0$, and the associated Gaussian isometry $V$, R\'enyi order $s > 1$, and mean bound $\mu$ satisfying Proposition~\ref{prop:finite-copy-certificate}.
  These choices depend only on $\ket{\psi}$; in particular, the same order $s$ applies to every state in $\K_\psi$ and is independent of the copy number and error tolerances.
  For every fixed $0 \le \delta < \delta' < 1$, there are positive integers $m_k$ with $m_k / k \to p_\psi$ and states $\ket{\phi_{k, j}} \in \K_\psi$ such that, for all sufficiently large $k$,
  \begin{equation}
    \abchiG (\ket{\psi}^{\otimes k}) \ge \bchiG^{(\delta')} \mleft(\bigotimes_{j = 1}^{m_k} \ket{\phi_{k, j}}\mright) \ge (1 - \delta'^2)^{s / (s - 1)} f_{\K_\psi} (m_k) B_{\K_\psi}^{m_k}.
    \label{eq:bosonic-transfer-main}
  \end{equation}
  Here $B_{\K_\psi} = \exp \mleft(\inf_{\ket{\phi} \in \K_\psi} H_s (\rho_\phi) - d_n (\mu)\mright) > 1$, and $f_{\K_\psi}$ is given by Eq.~\eqref{eq:bosonic-subexponential-factor}.
  The sequence satisfies $\log f_{\K_\psi} (m) / m \to 0$.
  Consequently,
  \begin{equation}
    \liminf_{k \to \infty} \frac{1}{k} \log \abchiG (\ket{\psi}^{\otimes k})
    \ge p_\psi \log B_{\K_\psi} > 0.
    \label{eq:bosonic-final-rate}
  \end{equation}
\end{theorem}

\begin{proof}
  Fix $0 \le \delta < \delta' < 1$.
  The one-copy outcome law is the probability measure
  \begin{equation}
    \odif{\mb P} (\omega) \coloneqq \norm*{K_\omega \ket{\psi}}_2^2 \odif{\nu} (\omega).
    \label{eq:target-outcome-law}
  \end{equation}
  Put $r_k \coloneqq \abchiG (\ket{\psi}^{\otimes k})$.
  If $r_k = \infty$, set $\ket{\eta_k} = \ket{\psi}^{\otimes k}$ for the record-selection argument below; the rank inequality is then automatic.
  For finite $r_k$, choose an approximant satisfying
  \begin{equation}
    \norm*{\ket{\psi}^{\otimes k} - \ket{\eta_k}}_2 \le \delta,
    \quad \bchiG (\ket{\eta_k}) = r_k.
    \label{eq:lifting-approximant}
  \end{equation}
  For the tensor-power target state, define
  \begin{equation}
    Q_k (\bm\omega) \coloneqq \norm*{T_{\bm\omega} \ket{\psi}^{\otimes k}}_2^2 = \prod_{i = 1}^k \norm*{K_{\omega_i} \ket{\psi}}_2^2.
    \label{eq:record-density}
  \end{equation}
  Thus $Q_k$ is the density of $\mb P^{\otimes k}$ relative to $\nu^{\otimes k}$.
  Write the corresponding squared error as
  \begin{equation}
    E_k (\bm\omega) \coloneqq \norm*{T_{\bm\omega} (\ket{\psi}^{\otimes k} - \ket{\eta_k})}_2^2.
    \label{eq:record-density-error}
  \end{equation}
  Equation~\eqref{eq:continuous-product-completeness} gives the integrated squared-error bound
  \begin{equation}
    \int_{\Omega^k}
    E_k \odif{\nu^{\otimes k}} \le \delta^2.
    \label{eq:continuous-error-budget}
  \end{equation}
  On the set where $Q_k > 0$, the ratio $E_k / Q_k$ is the squared error after dividing both output vectors by the target's norm.
  On a record with $E_k > \delta'^2 Q_k$, we have $Q_k < E_k / \delta'^2$.
  Integrating this inequality over those records and using Eq.~\eqref{eq:continuous-error-budget} gives
  \begin{equation}
    \mb P^{\otimes k} (Q_k > 0 \text{ and }
    E_k > \delta'^2 Q_k) \le \delta'^{-2} \int_{\Omega^k} E_k \odif{\nu^{\otimes k}} \le \frac{\delta^2}{\delta'^2} < 1.
    \label{eq:bad-record-bound}
  \end{equation}
  Under $\mb P^{\otimes k}$, the success count $m$ has a binomial distribution with parameters $k$ and $p_\psi$, hence mean $k p_\psi$ and variance $k p_\psi (1 - p_\psi)$.
  Chebyshev's inequality gives
  \begin{equation}
    \mb P^{\otimes k} \mleft(\abs*{m / k - p_\psi} > k^{-1 / 4}\mright) \le \frac{p_\psi (1 - p_\psi)}{\sqrt{k}}.
    \label{eq:success-count-typical}
  \end{equation}
  The probability that $E_k > \delta'^2 Q_k$ with $Q_k > 0$ is bounded strictly below one, while the probability that $\abs*{m / k - p_\psi} > k^{-1 / 4}$ tends to zero.
  Their sum is therefore less than one for all sufficiently large $k$, even though the two events need not be independent.
  Also $\mb P^{\otimes k} (Q_k = 0) = \int_{\{Q_k = 0\}} Q_k\,\odif{\nu^{\otimes k}} = 0$.
  The complement of these two events thus contains a record $\bm\omega_k$ satisfying
  \begin{equation}
    Q_k (\bm\omega_k) > 0, \quad E_k (\bm\omega_k) \le \delta'^2 Q_k (\bm\omega_k), \quad \abs*{m (\bm\omega_k) / k - p_\psi} \le k^{-1 / 4}.
    \label{eq:good-typical-record}
  \end{equation}
  Its success count $m_k \coloneqq m (\bm\omega_k)$ satisfies
  \begin{equation}
    m_k / k \longrightarrow p_\psi,
    \quad m_k \longrightarrow \infty.
    \label{eq:selected-success-frequency}
  \end{equation}

  Denote the normalized successful outputs, in their original order, by $\ket{\phi_j} \in \K_\psi$, $1 \le j \le m_k$.
  The selected states can change with $k$; this is where uniformity over the whole family $\K_\psi$ is essential.
  Collecting the scalar factors from Eq.~\eqref{eq:local-core-outputs} gives
  \begin{equation}
    T_{\bm\omega_k} \ket{\psi}^{\otimes k} = c_k \bigotimes_{j = 1}^{m_k} \ket{\phi_j}, \quad \abs*{c_k}^2 = Q_k (\bm\omega_k) > 0.
    \label{eq:typical-branch-target}
  \end{equation}
  Rescale the corresponding approximant by the same scalar, setting $\ket{\zeta_k} \coloneqq c_k^{-1} T_{\bm\omega_k} \ket{\eta_k}$.
  Equation~\eqref{eq:good-typical-record} gives
  \begin{equation}
    \norm*{\bigotimes_{j = 1}^{m_k} \ket{\phi_j} - \ket{\zeta_k}}_2^2
    = \frac{E_k (\bm\omega_k)}{Q_k (\bm\omega_k)}
    \le \delta'^2.
    \label{eq:successful-product-error}
  \end{equation}
  Multiplication by a nonzero scalar preserves border Gaussian rank.
  Equation~\eqref{eq:gaussian-map-rank} therefore gives
  \begin{equation}
    \bchiG^{(\delta')} \mleft(\bigotimes_{j = 1}^{m_k} \ket{\phi_j}\mright)
    \le \bchiG (\ket{\zeta_k})
    \le \bchiG (\ket{\eta_k}) = r_k.
    \label{eq:selected-record-rank-transfer}
  \end{equation}
  Applying Eq.~\eqref{eq:heterogeneous-exponent} to the selected product in Eq.~\eqref{eq:selected-record-rank-transfer} gives Eq.~\eqref{eq:bosonic-transfer-main}.
  Since $m_k / k \to p_\psi > 0$,
  \begin{equation}
    \frac{\log f_{\K_\psi} (m_k)}{k}
    = \frac{m_k}{k}\,\frac{\log f_{\K_\psi} (m_k)}{m_k} \longrightarrow 0.
    \label{eq:bosonic-transferred-factor}
  \end{equation}
  The factor $(1 - \delta'^2)^{s / (s - 1)}$ is positive and independent of $k$.
  Taking logarithms in Eq.~\eqref{eq:bosonic-transfer-main}, dividing by $k$, and using $m_k / k \to p_\psi$ therefore proves Eq.~\eqref{eq:bosonic-final-rate}.
\end{proof}

For each $0 \le \delta < 1$, choose $\delta' = (1 + \delta) / 2$ in Theorem~\ref{thm:bosonic-exponential-transfer}.
Since $B_\psi = B_{\K_\psi}$, Eq.~\eqref{eq:bosonic-final-rate} proves Theorem~\ref{thm:boson-universal-main} with the rate $p_\psi \log B_\psi$.

\subsection{Gaussian rank lower bounds for \texorpdfstring{\IfFontSeriesContextTF{bf}{\boldmath}{}$\ket{1}^{\otimes k}$}{the k-fold tensor power of ket 1}}
\label{appendix:boson-single-photon-rate}

\begin{proof}[Proof of Corollary~\ref{thm:boson-single-photon-main}]
  Use the Gaussian isometry $V \ket{\phi} = S_r B_\theta (\ket{\phi} \otimes \ket{0^2}_E)$ with $t = 1 / 6$.
  For $\rho_E \coloneqq \Tr_S V \ket{1} \bra{1} V^\dagger$, Eq.~\eqref{eq:number-state-environment-spectrum} with $m = 1$ gives the nonzero eigenvalues
  \begin{equation}
    q_0 = \frac17, \quad q_h = \frac{30 h + 1}{7^{h + 1}} \quad (h \ge 1).
    \label{eq:single-photon-environment-spectrum}
  \end{equation}
  By Eqs.~\eqref{eq:beam-splitter-expansion} and~\eqref{eq:two-mode-squeezer-expansion}, the environment occupations $\ket{1, h}$ and $\ket{0, h}$ have probabilities $7^{-h - 1}$ and $30 (h + 1) 7^{-h - 2}$, respectively, for $h \ge 0$.
  Hence the photon-number generating function is
  \begin{equation}
    f (x) \coloneqq \Tr (x^{N_E} \rho_E)
    = \sum_{h \ge 0} \mleft(\frac{x^{h + 1}}{7^{h + 1}} + \frac{30 (h + 1) x^h}{7^{h + 2}}\mright)
    = \frac{x}{7 - x} + \frac{30}{(7 - x)^2}, \quad 0 < x < 7.
    \label{eq:single-photon-number-generating-function}
  \end{equation}
  Differentiation gives
  \begin{equation}
    \Tr (N_E \rho_E) = f' (1) = \frac{17}{36}, \quad
    \Var (N_E; \rho_E) = f'' (1) + f' (1) - f' (1)^2 = \frac{587}{1296}.
    \label{eq:single-photon-environment-moments}
  \end{equation}

  Set $s = 21 / 20$ and merge the eigenvalues with $h \ge 9$ into one entry of weight $\sum_{h \ge 9} q_h = 46 / 40353607$.
  Merging decreases R\'enyi entropy, so
  \begin{equation}
    H_s (\rho_E) \ge -20 \log \mleft[\sum_{h = 0}^{8} q_h^{21 / 20} + \mleft(\sum_{h \ge 9} q_h\mright)^{21 / 20}\mright] > 20 \log \frac{20}{19}.
    \label{eq:single-photon-entropy-certificate}
  \end{equation}
  The last inequality follows by direct calculation, as does
  \begin{equation}
    \mleft(\frac{19}{20}\mright)^{720}
    \mleft(\frac{53}{36}\mright)^{53}
    \mleft(\frac{11}{10}\mright)^{36}
    < \mleft(\frac{17}{36}\mright)^{17}.
    \label{eq:single-photon-rational-gap}
  \end{equation}
  Taking logarithms in Eq.~\eqref{eq:single-photon-rational-gap}, dividing by $36$, and using Eq.~\eqref{eq:single-photon-entropy-certificate} gives
  \begin{equation}
    H_s (\rho_E) - d_1 (17 / 36) > \log (11 / 10).
    \label{eq:single-photon-explicit-gap}
  \end{equation}

  Proposition~\ref{prop:finite-copy-certificate} applies to $\K = \{\ket{1}\}$ with $\mu = 17 / 36$ and $s = 21 / 20$.
  Since $\log f_\K (k) / k \to 0$, its bound gives, for each fixed $0 \le \delta < 1$,
  \begin{equation}
    \liminf_{k \to \infty} \frac1k \log \abchiG (\ket{1}^{\otimes k})
    \ge H_s (\rho_E) - d_1 (17 / 36) > \log (11 / 10).
    \label{eq:single-photon-asymptotic-rate}
  \end{equation}
  Therefore $\abchiG (\ket{1}^{\otimes k}) \ge (11 / 10)^k$ for all sufficiently large $k$.
\end{proof}

\begin{proof}[Proof of Eq.~\eqref{equation:single-photon-border-bound}]
  Use the same Gaussian isometry $V$ as in the proof of Corollary~\ref{thm:boson-single-photon-main}.
  By Eqs.~\eqref{eq:beam-splitter-expansion} and~\eqref{eq:two-mode-squeezer-expansion}, the environment vector conditional on system occupation $h$ has support exactly $\{\ket{1, 0}\}$ when $h = 0$, and $\{\ket{1, h}, \ket{0, h - 1}\}$ when $h \ge 1$.
  These supports are disjoint for distinct $h$; for $k$ copies, the environment supports for distinct system occupation tuples are likewise disjoint.
  A photon cutoff preserves this orthogonality, so the Schmidt rank equals the number of system occupation tuples whose environment vectors remain nonzero.

  Fix an integer cutoff $p \ge 0$.
  A zero system occupation requires at least one environment photon, whereas a positive occupation $h$ requires at least $h - 1$.
  If exactly $j$ system occupations are zero, the tuple therefore survives precisely when the excesses of the remaining $k - j$ occupations above one sum to at most $p - j$.
  For $0 \le j \le \min (k, p)$, there are $\binom{k}{j}$ choices of the zero positions and $\binom{k + p - 2 j}{k - j}$ choices of these nonnegative excesses, by the same counting argument as in Eq.~\eqref{eq:gaussian-polynomial-span}.
  Hence
  \begin{equation}
    \SR_{S : E} \mleft(\Pi_{\le p}^E V^{\otimes k} \ket{1}^{\otimes k}\mright)
    = N_{k, p} \coloneqq \sum_{j = 0}^{\min(k, p)} \binom{k}{j} \binom{k + p - 2 j}{k - j}.
  \end{equation}
  Lemma~\ref{lem:photon-cap}, including its extension to border rank, therefore implies
  \begin{equation}
    \bchiG (\ket{1}^{\otimes k}) \ge \frac{N_{k, p}}{\binom{k + p}{p}}.
    \label{equation:single-photon-border-count}
  \end{equation}

  The ratio is maximized at $p = k$.
  To see this, write its $j$th summand as
  \begin{equation}
    A_{k, j} (p) \coloneqq \binom{k}{j} \binom{k + p - 2 j}{k - j} \binom{k + p}{p}^{-1}.
  \end{equation}
  For $1 \le j \le \min(k, p)$,
  \begin{equation}
    \frac{A_{k, j} (p + 1)}{A_{k, j} (p)}
    = \frac{(p + 1)(k + p + 1 - 2 j)}{(p + 1 - j)(k + p + 1)}.
  \end{equation}
  This quotient is greater than, equal to, or less than one according as $p$ is less than, equal to, or greater than $k - 1$.
  The $j = 0$ summand is constant, and additional positive summands enter as $p$ increases up to $k$.
  Hence the full ratio is maximized at $p = k$.

  At this cutoff,
  \begin{equation}
    A_{k, j} (k) = \binom{k}{j} 4^{-j} \prod_{i = 0}^{j - 1} \frac{2 k - 2 i}{2 k - 2 i - 1}.
  \end{equation}
  For $j \ge 1$, every factor in the product exceeds one.
  Summing over $j$ gives
  \begin{equation}
    \frac{N_{k, k}}{\binom{2 k}{k}} > \sum_{j = 0}^k \binom{k}{j} 4^{-j} = \mleft(\frac54\mright)^k,
  \end{equation}
  which proves Eq.~\eqref{equation:single-photon-border-bound}.
\end{proof}


\begin{thebibliography}{50}%
\makeatletter
\providecommand \@ifxundefined [1]{%
 \@ifx{#1\undefined}
}%
\providecommand \@ifnum [1]{%
 \ifnum #1\expandafter \@firstoftwo
 \else \expandafter \@secondoftwo
 \fi
}%
\providecommand \@ifx [1]{%
 \ifx #1\expandafter \@firstoftwo
 \else \expandafter \@secondoftwo
 \fi
}%
\providecommand \natexlab [1]{#1}%
\providecommand \enquote  [1]{``#1''}%
\providecommand \bibnamefont  [1]{#1}%
\providecommand \bibfnamefont [1]{#1}%
\providecommand \citenamefont [1]{#1}%
\providecommand \href@noop [0]{\@secondoftwo}%
\providecommand \href [0]{\begingroup \@sanitize@url \@href}%
\providecommand \@href[1]{\@@startlink{#1}\@@href}%
\providecommand \@@href[1]{\endgroup#1\@@endlink}%
\providecommand \@sanitize@url [0]{\catcode `\\12\catcode `\$12\catcode `\&12\catcode `\#12\catcode `\^12\catcode `\_12\catcode `\%12\relax}%
\providecommand \@@startlink[1]{}%
\providecommand \@@endlink[0]{}%
\providecommand \url  [0]{\begingroup\@sanitize@url \@url }%
\providecommand \@url [1]{\endgroup\@href {#1}{\urlprefix }}%
\providecommand \urlprefix  [0]{URL }%
\providecommand \Eprint [0]{\href }%
\providecommand \doibase [0]{https://doi.org/}%
\providecommand \selectlanguage [0]{\@gobble}%
\providecommand \bibinfo  [0]{\@secondoftwo}%
\providecommand \bibfield  [0]{\@secondoftwo}%
\providecommand \translation [1]{[#1]}%
\providecommand \BibitemOpen [0]{}%
\providecommand \bibitemStop [0]{}%
\providecommand \bibitemNoStop [0]{.\EOS\space}%
\providecommand \EOS [0]{\spacefactor3000\relax}%
\providecommand \BibitemShut  [1]{\csname bibitem#1\endcsname}%
\let\auto@bib@innerbib\@empty
\bibitem [{\citenamefont {Bravyi}\ \emph {et~al.}(2016)\citenamefont {Bravyi}, \citenamefont {Smith},\ and\ \citenamefont {Smolin}}]{bravyi2015tradingclassicalquantumcomputational}%
  \BibitemOpen
  \bibfield  {author} {\bibinfo {author} {\bibfnamefont {S.}~\bibnamefont {Bravyi}}, \bibinfo {author} {\bibfnamefont {G.}~\bibnamefont {Smith}},\ and\ \bibinfo {author} {\bibfnamefont {J.~A.}\ \bibnamefont {Smolin}},\ }\bibfield  {title} {\bibinfo {title} {Trading classical and quantum computational resources},\ }\href {https://doi.org/10.1103/PhysRevX.6.021043} {\bibfield  {journal} {\bibinfo  {journal} {Phys. Rev. X}\ }\textbf {\bibinfo {volume} {6}},\ \bibinfo {pages} {021043} (\bibinfo {year} {2016})}\BibitemShut {NoStop}%
\bibitem [{\citenamefont {Bravyi}\ \emph {et~al.}(2019)\citenamefont {Bravyi}, \citenamefont {Browne}, \citenamefont {Calpin}, \citenamefont {Campbell}, \citenamefont {Gosset},\ and\ \citenamefont {Howard}}]{bravyi2019simulationquantumcircuitslowrank}%
  \BibitemOpen
  \bibfield  {author} {\bibinfo {author} {\bibfnamefont {S.}~\bibnamefont {Bravyi}}, \bibinfo {author} {\bibfnamefont {D.}~\bibnamefont {Browne}}, \bibinfo {author} {\bibfnamefont {P.}~\bibnamefont {Calpin}}, \bibinfo {author} {\bibfnamefont {E.}~\bibnamefont {Campbell}}, \bibinfo {author} {\bibfnamefont {D.}~\bibnamefont {Gosset}},\ and\ \bibinfo {author} {\bibfnamefont {M.}~\bibnamefont {Howard}},\ }\bibfield  {title} {\bibinfo {title} {Simulation of quantum circuits by low-rank stabilizer decompositions},\ }\href {https://doi.org/10.22331/q-2019-09-02-181} {\bibfield  {journal} {\bibinfo  {journal} {{Quantum}}\ }\textbf {\bibinfo {volume} {3}},\ \bibinfo {pages} {181} (\bibinfo {year} {2019})}\BibitemShut {NoStop}%
\bibitem [{\citenamefont {Dias}\ and\ \citenamefont {Koenig}(2025)}]{dias2025samplingcomplexitycoherentsuperpositions}%
  \BibitemOpen
  \bibfield  {author} {\bibinfo {author} {\bibfnamefont {B.}~\bibnamefont {Dias}}\ and\ \bibinfo {author} {\bibfnamefont {R.}~\bibnamefont {Koenig}},\ }\href@noop {} {\bibinfo {title} {On the sampling complexity of coherent superpositions}} (\bibinfo {year} {2025}),\ \Eprint {https://arxiv.org/abs/2501.17071} {arXiv:2501.17071 [quant-ph]} \BibitemShut {NoStop}%
\bibitem [{\citenamefont {Bravyi}\ and\ \citenamefont {Gosset}(2016)}]{bravyi2019improvedclassicalsimulationquantum}%
  \BibitemOpen
  \bibfield  {author} {\bibinfo {author} {\bibfnamefont {S.}~\bibnamefont {Bravyi}}\ and\ \bibinfo {author} {\bibfnamefont {D.}~\bibnamefont {Gosset}},\ }\bibfield  {title} {\bibinfo {title} {Improved classical simulation of quantum circuits dominated by {Clifford} gates},\ }\href {https://doi.org/10.1103/PhysRevLett.116.250501} {\bibfield  {journal} {\bibinfo  {journal} {Phys. Rev. Lett.}\ }\textbf {\bibinfo {volume} {116}},\ \bibinfo {pages} {250501} (\bibinfo {year} {2016})}\BibitemShut {NoStop}%
\bibitem [{\citenamefont {Huang}\ and\ \citenamefont {Love}(2019)}]{huang2018approximatestabilizerrankimproved}%
  \BibitemOpen
  \bibfield  {author} {\bibinfo {author} {\bibfnamefont {Y.}~\bibnamefont {Huang}}\ and\ \bibinfo {author} {\bibfnamefont {P.}~\bibnamefont {Love}},\ }\bibfield  {title} {\bibinfo {title} {Approximate stabilizer rank and improved weak simulation of clifford-dominated circuits for qudits},\ }\href {https://doi.org/10.1103/PhysRevA.99.052307} {\bibfield  {journal} {\bibinfo  {journal} {Phys. Rev. A}\ }\textbf {\bibinfo {volume} {99}},\ \bibinfo {pages} {052307} (\bibinfo {year} {2019})}\BibitemShut {NoStop}%
\bibitem [{\citenamefont {Qassim}\ \emph {et~al.}(2019)\citenamefont {Qassim}, \citenamefont {Wallman},\ and\ \citenamefont {Emerson}}]{qassim2019cliffordrecompilationfasterclassical}%
  \BibitemOpen
  \bibfield  {author} {\bibinfo {author} {\bibfnamefont {H.}~\bibnamefont {Qassim}}, \bibinfo {author} {\bibfnamefont {J.~J.}\ \bibnamefont {Wallman}},\ and\ \bibinfo {author} {\bibfnamefont {J.}~\bibnamefont {Emerson}},\ }\bibfield  {title} {\bibinfo {title} {Clifford recompilation for faster classical simulation of quantum circuits},\ }\href {https://doi.org/10.22331/q-2019-08-05-170} {\bibfield  {journal} {\bibinfo  {journal} {{Quantum}}\ }\textbf {\bibinfo {volume} {3}},\ \bibinfo {pages} {170} (\bibinfo {year} {2019})}\BibitemShut {NoStop}%
\bibitem [{\citenamefont {Kocia}(2022)}]{kocia2022improvedstrongsimulationuniversal}%
  \BibitemOpen
  \bibfield  {author} {\bibinfo {author} {\bibfnamefont {L.}~\bibnamefont {Kocia}},\ }\href@noop {} {\bibinfo {title} {Improved strong simulation of universal quantum circuits}} (\bibinfo {year} {2022}),\ \Eprint {https://arxiv.org/abs/2012.11739} {arXiv:2012.11739 [quant-ph]} \BibitemShut {NoStop}%
\bibitem [{\citenamefont {Huang}\ and\ \citenamefont {Love}(2021)}]{huang2021feynmanpathtypesimulationusing}%
  \BibitemOpen
  \bibfield  {author} {\bibinfo {author} {\bibfnamefont {Y.}~\bibnamefont {Huang}}\ and\ \bibinfo {author} {\bibfnamefont {P.}~\bibnamefont {Love}},\ }\bibfield  {title} {\bibinfo {title} {Feynman-path-type simulation using stabilizer projector decomposition of unitaries},\ }\href {https://doi.org/10.1103/PhysRevA.103.022428} {\bibfield  {journal} {\bibinfo  {journal} {Phys. Rev. A}\ }\textbf {\bibinfo {volume} {103}},\ \bibinfo {pages} {022428} (\bibinfo {year} {2021})}\BibitemShut {NoStop}%
\bibitem [{\citenamefont {Qassim}\ \emph {et~al.}(2021)\citenamefont {Qassim}, \citenamefont {Pashayan},\ and\ \citenamefont {Gosset}}]{qassim2021improvedupperboundsstabilizer}%
  \BibitemOpen
  \bibfield  {author} {\bibinfo {author} {\bibfnamefont {H.}~\bibnamefont {Qassim}}, \bibinfo {author} {\bibfnamefont {H.}~\bibnamefont {Pashayan}},\ and\ \bibinfo {author} {\bibfnamefont {D.}~\bibnamefont {Gosset}},\ }\bibfield  {title} {\bibinfo {title} {Improved upper bounds on the stabilizer rank of magic states},\ }\href {https://doi.org/10.22331/q-2021-12-20-606} {\bibfield  {journal} {\bibinfo  {journal} {{Quantum}}\ }\textbf {\bibinfo {volume} {5}},\ \bibinfo {pages} {606} (\bibinfo {year} {2021})}\BibitemShut {NoStop}%
\bibitem [{\citenamefont {Labib}(2022)}]{labib2022stabilizerrankhigherorderfourier}%
  \BibitemOpen
  \bibfield  {author} {\bibinfo {author} {\bibfnamefont {F.}~\bibnamefont {Labib}},\ }\bibfield  {title} {\bibinfo {title} {Stabilizer rank and higher-order {F}ourier analysis},\ }\href {https://doi.org/10.22331/q-2022-02-09-645} {\bibfield  {journal} {\bibinfo  {journal} {{Quantum}}\ }\textbf {\bibinfo {volume} {6}},\ \bibinfo {pages} {645} (\bibinfo {year} {2022})}\BibitemShut {NoStop}%
\bibitem [{\citenamefont {Peleg}\ \emph {et~al.}(2022)\citenamefont {Peleg}, \citenamefont {Shpilka},\ and\ \citenamefont {Volk}}]{peleg2022lowerboundsstabilizerrank}%
  \BibitemOpen
  \bibfield  {author} {\bibinfo {author} {\bibfnamefont {S.}~\bibnamefont {Peleg}}, \bibinfo {author} {\bibfnamefont {A.}~\bibnamefont {Shpilka}},\ and\ \bibinfo {author} {\bibfnamefont {B.~L.}\ \bibnamefont {Volk}},\ }\bibfield  {title} {\bibinfo {title} {Lower {B}ounds on {S}tabilizer {R}ank},\ }\href {https://doi.org/10.22331/q-2022-02-15-652} {\bibfield  {journal} {\bibinfo  {journal} {{Quantum}}\ }\textbf {\bibinfo {volume} {6}},\ \bibinfo {pages} {652} (\bibinfo {year} {2022})}\BibitemShut {NoStop}%
\bibitem [{\citenamefont {Lovitz}\ and\ \citenamefont {Steffan}(2022)}]{lovitz2022newtechniquesboundingstabilizer}%
  \BibitemOpen
  \bibfield  {author} {\bibinfo {author} {\bibfnamefont {B.}~\bibnamefont {Lovitz}}\ and\ \bibinfo {author} {\bibfnamefont {V.}~\bibnamefont {Steffan}},\ }\bibfield  {title} {\bibinfo {title} {New techniques for bounding stabilizer rank},\ }\href {https://doi.org/10.22331/q-2022-04-20-692} {\bibfield  {journal} {\bibinfo  {journal} {{Quantum}}\ }\textbf {\bibinfo {volume} {6}},\ \bibinfo {pages} {692} (\bibinfo {year} {2022})}\BibitemShut {NoStop}%
\bibitem [{\citenamefont {Kissinger}\ \emph {et~al.}(2022)\citenamefont {Kissinger}, \citenamefont {van~de Wetering},\ and\ \citenamefont {Vilmart}}]{kissinger2022classicalsimulationquantumcircuits}%
  \BibitemOpen
  \bibfield  {author} {\bibinfo {author} {\bibfnamefont {A.}~\bibnamefont {Kissinger}}, \bibinfo {author} {\bibfnamefont {J.}~\bibnamefont {van~de Wetering}},\ and\ \bibinfo {author} {\bibfnamefont {R.}~\bibnamefont {Vilmart}},\ }\bibfield  {title} {\bibinfo {title} {{Classical Simulation of Quantum Circuits with Partial and Graphical Stabiliser Decompositions}},\ }in\ \href {https://doi.org/10.4230/LIPIcs.TQC.2022.5} {\emph {\bibinfo {booktitle} {17th Conference on the Theory of Quantum Computation, Communication and Cryptography (TQC 2022)}}},\ \bibinfo {series} {Leibniz International Proceedings in Informatics (LIPIcs)}, Vol.\ \bibinfo {volume} {232},\ \bibinfo {editor} {edited by\ \bibinfo {editor} {\bibfnamefont {F.}~\bibnamefont {Le~Gall}}\ and\ \bibinfo {editor} {\bibfnamefont {T.}~\bibnamefont {Morimae}}}\ (\bibinfo  {publisher} {Schloss Dagstuhl -- Leibniz-Zentrum f{\"u}r Informatik},\ \bibinfo {address} {Dagstuhl, Germany},\ \bibinfo {year} {2022})\ pp.\ \bibinfo {pages} {5:1--5:13}\BibitemShut {NoStop}%
\bibitem [{\citenamefont {de~Silva}\ \emph {et~al.}(2024)\citenamefont {de~Silva}, \citenamefont {Yin},\ and\ \citenamefont {Strelchuk}}]{desilva2024basesoptimisingstabiliserdecompositions}%
  \BibitemOpen
  \bibfield  {author} {\bibinfo {author} {\bibfnamefont {N.}~\bibnamefont {de~Silva}}, \bibinfo {author} {\bibfnamefont {M.}~\bibnamefont {Yin}},\ and\ \bibinfo {author} {\bibfnamefont {S.}~\bibnamefont {Strelchuk}},\ }\bibfield  {title} {\bibinfo {title} {Bases for optimising stabiliser decompositions of quantum states},\ }\href {https://doi.org/10.1088/2058-9565/ad53fc} {\bibfield  {journal} {\bibinfo  {journal} {Quantum Science and Technology}\ }\textbf {\bibinfo {volume} {9}},\ \bibinfo {pages} {045004} (\bibinfo {year} {2024})}\BibitemShut {NoStop}%
\bibitem [{\citenamefont {Mehraban}\ and\ \citenamefont {Tahmasbi}(2024)}]{mehraban2024quadraticlowerboundsapproximate}%
  \BibitemOpen
  \bibfield  {author} {\bibinfo {author} {\bibfnamefont {S.}~\bibnamefont {Mehraban}}\ and\ \bibinfo {author} {\bibfnamefont {M.}~\bibnamefont {Tahmasbi}},\ }\bibfield  {title} {\bibinfo {title} {Quadratic lower bounds on the approximate stabilizer rank: A probabilistic approach},\ }in\ \href {https://doi.org/10.1145/3618260.3649733} {\emph {\bibinfo {booktitle} {Proceedings of the 56th Annual ACM Symposium on Theory of Computing}}},\ \bibinfo {series and number} {STOC ’24}\ (\bibinfo  {publisher} {ACM},\ \bibinfo {year} {2024})\ pp.\ \bibinfo {pages} {608--619}\BibitemShut {NoStop}%
\bibitem [{\citenamefont {Mehraban}\ and\ \citenamefont {Tahmasbi}(2025)}]{mehraban2024improvedboundstestinglow}%
  \BibitemOpen
  \bibfield  {author} {\bibinfo {author} {\bibfnamefont {S.}~\bibnamefont {Mehraban}}\ and\ \bibinfo {author} {\bibfnamefont {M.}~\bibnamefont {Tahmasbi}},\ }\bibfield  {title} {\bibinfo {title} {Improved bounds for testing low stabilizer complexity states},\ }in\ \href {https://doi.org/10.1145/3717823.3718228} {\emph {\bibinfo {booktitle} {Proceedings of the 57th Annual ACM Symposium on Theory of Computing}}},\ \bibinfo {series and number} {STOC ’25}\ (\bibinfo  {publisher} {ACM},\ \bibinfo {year} {2025})\ pp.\ \bibinfo {pages} {1222--1233}\BibitemShut {NoStop}%
\bibitem [{\citenamefont {Arunachalam}\ and\ \citenamefont {Dutt}(2025)}]{arunachalam2025learningstabilizerstructurequantum}%
  \BibitemOpen
  \bibfield  {author} {\bibinfo {author} {\bibfnamefont {S.}~\bibnamefont {Arunachalam}}\ and\ \bibinfo {author} {\bibfnamefont {A.}~\bibnamefont {Dutt}},\ }\href@noop {} {\bibinfo {title} {Learning stabilizer structure of quantum states}} (\bibinfo {year} {2025}),\ \Eprint {https://arxiv.org/abs/2510.05890} {arXiv:2510.05890 [quant-ph]} \BibitemShut {NoStop}%
\bibitem [{\citenamefont {Kjelstr{\o{}}m}\ \emph {et~al.}(2026)\citenamefont {Kjelstr{\o{}}m}, \citenamefont {Pavlogiannis},\ and\ \citenamefont {van~de Pol}}]{kjelstrøm2026efficientsimulationhighlevelquantum}%
  \BibitemOpen
  \bibfield  {author} {\bibinfo {author} {\bibfnamefont {A.~H.}\ \bibnamefont {Kjelstr{\o{}}m}}, \bibinfo {author} {\bibfnamefont {A.}~\bibnamefont {Pavlogiannis}},\ and\ \bibinfo {author} {\bibfnamefont {J.}~\bibnamefont {van~de Pol}},\ }\bibfield  {title} {\bibinfo {title} {Efficient {S}imulation of {H}igh-{L}evel {Q}uantum {G}ates},\ }\href {https://doi.org/10.22331/q-2026-05-05-2093} {\bibfield  {journal} {\bibinfo  {journal} {{Quantum}}\ }\textbf {\bibinfo {volume} {10}},\ \bibinfo {pages} {2093} (\bibinfo {year} {2026})}\BibitemShut {NoStop}%
\bibitem [{\citenamefont {Gosset}\ \emph {et~al.}(2026)\citenamefont {Gosset}, \citenamefont {Kothari},\ and\ \citenamefont {Wu}}]{gosset2026quantumstatepreparationoptimal}%
  \BibitemOpen
  \bibfield  {author} {\bibinfo {author} {\bibfnamefont {D.}~\bibnamefont {Gosset}}, \bibinfo {author} {\bibfnamefont {R.}~\bibnamefont {Kothari}},\ and\ \bibinfo {author} {\bibfnamefont {K.}~\bibnamefont {Wu}},\ }\bibfield  {title} {\bibinfo {title} {Quantum state preparation with optimal {T}-count},\ }\href {https://doi.org/10.22331/q-2026-07-22-2168} {\bibfield  {journal} {\bibinfo  {journal} {{Quantum}}\ }\textbf {\bibinfo {volume} {10}},\ \bibinfo {pages} {2168} (\bibinfo {year} {2026})}\BibitemShut {NoStop}%
\bibitem [{\citenamefont {Kalra}\ and\ \citenamefont {Sinha}(2026)}]{kalra2026stabilizerranksbarneswall}%
  \BibitemOpen
  \bibfield  {author} {\bibinfo {author} {\bibfnamefont {A.~R.}\ \bibnamefont {Kalra}}\ and\ \bibinfo {author} {\bibfnamefont {P.}~\bibnamefont {Sinha}},\ }\bibfield  {title} {\bibinfo {title} {Stabilizer {R}anks, {B}arnes {W}all {L}attices and {M}agic {M}onotones},\ }\href {https://doi.org/10.22331/q-2026-07-29-2179} {\bibfield  {journal} {\bibinfo  {journal} {{Quantum}}\ }\textbf {\bibinfo {volume} {10}},\ \bibinfo {pages} {2179} (\bibinfo {year} {2026})}\BibitemShut {NoStop}%
\bibitem [{\citenamefont {Labib}\ and\ \citenamefont {Russo}(2026)}]{labib2026stabilizerrankboundsmagicstate}%
  \BibitemOpen
  \bibfield  {author} {\bibinfo {author} {\bibfnamefont {F.}~\bibnamefont {Labib}}\ and\ \bibinfo {author} {\bibfnamefont {V.}~\bibnamefont {Russo}},\ }\href@noop {} {\bibinfo {title} {Stabilizer rank bounds for magic-state orbits}} (\bibinfo {year} {2026}),\ \Eprint {https://arxiv.org/abs/2605.28586} {arXiv:2605.28586 [quant-ph]} \BibitemShut {NoStop}%
\bibitem [{\citenamefont {Valiant}(2002)}]{valiant2002quantumcircuitssimulatedclassicallypolynomialtime}%
  \BibitemOpen
  \bibfield  {author} {\bibinfo {author} {\bibfnamefont {L.~G.}\ \bibnamefont {Valiant}},\ }\bibfield  {title} {\bibinfo {title} {Quantum circuits that can be simulated classically in polynomial time},\ }\href {https://doi.org/10.1137/s0097539700377025} {\bibfield  {journal} {\bibinfo  {journal} {SIAM Journal on Computing}\ }\textbf {\bibinfo {volume} {31}},\ \bibinfo {pages} {1229} (\bibinfo {year} {2002})}\BibitemShut {NoStop}%
\bibitem [{\citenamefont {Terhal}\ and\ \citenamefont {DiVincenzo}(2002)}]{terhal2001classicalsimulationnoninteractingfermionquantum}%
  \BibitemOpen
  \bibfield  {author} {\bibinfo {author} {\bibfnamefont {B.~M.}\ \bibnamefont {Terhal}}\ and\ \bibinfo {author} {\bibfnamefont {D.~P.}\ \bibnamefont {DiVincenzo}},\ }\bibfield  {title} {\bibinfo {title} {Classical simulation of noninteracting-fermion quantum circuits},\ }\href {https://doi.org/10.1103/PhysRevA.65.032325} {\bibfield  {journal} {\bibinfo  {journal} {Phys. Rev. A}\ }\textbf {\bibinfo {volume} {65}},\ \bibinfo {pages} {032325} (\bibinfo {year} {2002})}\BibitemShut {NoStop}%
\bibitem [{\citenamefont {Bravyi}(2005)}]{bravyi2004lagrangianrepresentationfermioniclinear}%
  \BibitemOpen
  \bibfield  {author} {\bibinfo {author} {\bibfnamefont {S.}~\bibnamefont {Bravyi}},\ }\bibfield  {title} {\bibinfo {title} {Lagrangian representation for fermionic linear optics},\ }\href {https://doi.org/10.26421/qic5.3-3} {\bibfield  {journal} {\bibinfo  {journal} {Quantum Information and Computation}\ }\textbf {\bibinfo {volume} {5}},\ \bibinfo {pages} {216} (\bibinfo {year} {2005})}\BibitemShut {NoStop}%
\bibitem [{\citenamefont {Jozsa}\ and\ \citenamefont {Miyake}(2008)}]{jozsa2008matchgatesclassicalsimulationquantum}%
  \BibitemOpen
  \bibfield  {author} {\bibinfo {author} {\bibfnamefont {R.}~\bibnamefont {Jozsa}}\ and\ \bibinfo {author} {\bibfnamefont {A.}~\bibnamefont {Miyake}},\ }\bibfield  {title} {\bibinfo {title} {Matchgates and classical simulation of quantum circuits},\ }\href {https://doi.org/10.1098/rspa.2008.0189} {\bibfield  {journal} {\bibinfo  {journal} {Proceedings of the Royal Society A: Mathematical, Physical and Engineering Sciences}\ }\textbf {\bibinfo {volume} {464}},\ \bibinfo {pages} {3089} (\bibinfo {year} {2008})}\BibitemShut {NoStop}%
\bibitem [{\citenamefont {Bartlett}\ \emph {et~al.}(2002)\citenamefont {Bartlett}, \citenamefont {Sanders}, \citenamefont {Braunstein},\ and\ \citenamefont {Nemoto}}]{bartlett2002efficientclassicalsimulationcontinuous}%
  \BibitemOpen
  \bibfield  {author} {\bibinfo {author} {\bibfnamefont {S.~D.}\ \bibnamefont {Bartlett}}, \bibinfo {author} {\bibfnamefont {B.~C.}\ \bibnamefont {Sanders}}, \bibinfo {author} {\bibfnamefont {S.~L.}\ \bibnamefont {Braunstein}},\ and\ \bibinfo {author} {\bibfnamefont {K.}~\bibnamefont {Nemoto}},\ }\bibfield  {title} {\bibinfo {title} {Efficient classical simulation of continuous variable quantum information processes},\ }\href {https://doi.org/10.1103/PhysRevLett.88.097904} {\bibfield  {journal} {\bibinfo  {journal} {Phys. Rev. Lett.}\ }\textbf {\bibinfo {volume} {88}},\ \bibinfo {pages} {097904} (\bibinfo {year} {2002})}\BibitemShut {NoStop}%
\bibitem [{\citenamefont {Hebenstreit}\ \emph {et~al.}(2019)\citenamefont {Hebenstreit}, \citenamefont {Jozsa}, \citenamefont {Kraus}, \citenamefont {Strelchuk},\ and\ \citenamefont {Yoganathan}}]{hebenstreit2019purefermionicnongaussianstates}%
  \BibitemOpen
  \bibfield  {author} {\bibinfo {author} {\bibfnamefont {M.}~\bibnamefont {Hebenstreit}}, \bibinfo {author} {\bibfnamefont {R.}~\bibnamefont {Jozsa}}, \bibinfo {author} {\bibfnamefont {B.}~\bibnamefont {Kraus}}, \bibinfo {author} {\bibfnamefont {S.}~\bibnamefont {Strelchuk}},\ and\ \bibinfo {author} {\bibfnamefont {M.}~\bibnamefont {Yoganathan}},\ }\bibfield  {title} {\bibinfo {title} {All pure fermionic non-{Gaussian} states are magic states for matchgate computations},\ }\href {https://doi.org/10.1103/PhysRevLett.123.080503} {\bibfield  {journal} {\bibinfo  {journal} {Phys. Rev. Lett.}\ }\textbf {\bibinfo {volume} {123}},\ \bibinfo {pages} {080503} (\bibinfo {year} {2019})}\BibitemShut {NoStop}%
\bibitem [{\citenamefont {Lloyd}\ and\ \citenamefont {Braunstein}(1999)}]{lloyd1998quantumcomputationcontinuousvariables}%
  \BibitemOpen
  \bibfield  {author} {\bibinfo {author} {\bibfnamefont {S.}~\bibnamefont {Lloyd}}\ and\ \bibinfo {author} {\bibfnamefont {S.~L.}\ \bibnamefont {Braunstein}},\ }\bibfield  {title} {\bibinfo {title} {Quantum computation over continuous variables},\ }\href {https://doi.org/10.1103/PhysRevLett.82.1784} {\bibfield  {journal} {\bibinfo  {journal} {Phys. Rev. Lett.}\ }\textbf {\bibinfo {volume} {82}},\ \bibinfo {pages} {1784} (\bibinfo {year} {1999})}\BibitemShut {NoStop}%
\bibitem [{\citenamefont {Gottesman}\ \emph {et~al.}(2001)\citenamefont {Gottesman}, \citenamefont {Kitaev},\ and\ \citenamefont {Preskill}}]{gottesman2001encodingqubitoscillator}%
  \BibitemOpen
  \bibfield  {author} {\bibinfo {author} {\bibfnamefont {D.}~\bibnamefont {Gottesman}}, \bibinfo {author} {\bibfnamefont {A.}~\bibnamefont {Kitaev}},\ and\ \bibinfo {author} {\bibfnamefont {J.}~\bibnamefont {Preskill}},\ }\bibfield  {title} {\bibinfo {title} {Encoding a qubit in an oscillator},\ }\href {https://doi.org/10.1103/PhysRevA.64.012310} {\bibfield  {journal} {\bibinfo  {journal} {Phys. Rev. A}\ }\textbf {\bibinfo {volume} {64}},\ \bibinfo {pages} {012310} (\bibinfo {year} {2001})}\BibitemShut {NoStop}%
\bibitem [{\citenamefont {Cudby}\ and\ \citenamefont {Strelchuk}(2025)}]{cudby2025gaussiandecompositionmagicstates}%
  \BibitemOpen
  \bibfield  {author} {\bibinfo {author} {\bibfnamefont {J.}~\bibnamefont {Cudby}}\ and\ \bibinfo {author} {\bibfnamefont {S.}~\bibnamefont {Strelchuk}},\ }\href@noop {} {\bibinfo {title} {Gaussian decomposition of magic states for matchgate computations}} (\bibinfo {year} {2025}),\ \Eprint {https://arxiv.org/abs/2307.12654} {arXiv:2307.12654 [quant-ph]} \BibitemShut {NoStop}%
\bibitem [{\citenamefont {Dias}\ and\ \citenamefont {Koenig}(2024)}]{dias2024classicalsimulationnongaussianfermionic}%
  \BibitemOpen
  \bibfield  {author} {\bibinfo {author} {\bibfnamefont {B.}~\bibnamefont {Dias}}\ and\ \bibinfo {author} {\bibfnamefont {R.}~\bibnamefont {Koenig}},\ }\bibfield  {title} {\bibinfo {title} {Classical simulation of non-{G}aussian fermionic circuits},\ }\href {https://doi.org/10.22331/q-2024-05-21-1350} {\bibfield  {journal} {\bibinfo  {journal} {{Quantum}}\ }\textbf {\bibinfo {volume} {8}},\ \bibinfo {pages} {1350} (\bibinfo {year} {2024})}\BibitemShut {NoStop}%
\bibitem [{\citenamefont {Reardon-Smith}\ \emph {et~al.}(2024)\citenamefont {Reardon-Smith}, \citenamefont {Oszmaniec},\ and\ \citenamefont {Korzekwa}}]{reardonsmith2024improvedsimulationquantumcircuits}%
  \BibitemOpen
  \bibfield  {author} {\bibinfo {author} {\bibfnamefont {O.}~\bibnamefont {Reardon-Smith}}, \bibinfo {author} {\bibfnamefont {M.}~\bibnamefont {Oszmaniec}},\ and\ \bibinfo {author} {\bibfnamefont {K.}~\bibnamefont {Korzekwa}},\ }\bibfield  {title} {\bibinfo {title} {Improved simulation of quantum circuits dominated by free fermionic operations},\ }\href {https://doi.org/10.22331/q-2024-12-04-1549} {\bibfield  {journal} {\bibinfo  {journal} {{Quantum}}\ }\textbf {\bibinfo {volume} {8}},\ \bibinfo {pages} {1549} (\bibinfo {year} {2024})}\BibitemShut {NoStop}%
\bibitem [{\citenamefont {Dias}\ and\ \citenamefont {K{\"o}nig}(2024{\natexlab{a}})}]{dias2025classicalsimulationnongaussianbosonic}%
  \BibitemOpen
  \bibfield  {author} {\bibinfo {author} {\bibfnamefont {B.}~\bibnamefont {Dias}}\ and\ \bibinfo {author} {\bibfnamefont {R.}~\bibnamefont {K{\"o}nig}},\ }\bibfield  {title} {\bibinfo {title} {Classical simulation of non-{Gaussian} bosonic circuits},\ }\href {https://doi.org/10.1103/PhysRevA.110.042402} {\bibfield  {journal} {\bibinfo  {journal} {Phys. Rev. A}\ }\textbf {\bibinfo {volume} {110}},\ \bibinfo {pages} {042402} (\bibinfo {year} {2024}{\natexlab{a}})},\ \bibinfo {note} {see also Erratum, Phys. Rev. A 113, 059901 (2026)}\BibitemShut {NoStop}%
\bibitem [{\citenamefont {Dias}\ and\ \citenamefont {K{\"o}nig}(2024{\natexlab{b}})}]{dias2026erratumclassicalsimulationnongaussianbosoniccircuits}%
  \BibitemOpen
  \bibfield  {author} {\bibinfo {author} {\bibfnamefont {B.}~\bibnamefont {Dias}}\ and\ \bibinfo {author} {\bibfnamefont {R.}~\bibnamefont {K{\"o}nig}},\ }\bibfield  {title} {\bibinfo {title} {Classical simulation of non-gaussian bosonic circuits},\ }\href {https://doi.org/10.1103/PhysRevA.110.042402} {\bibfield  {journal} {\bibinfo  {journal} {Phys. Rev. A}\ }\textbf {\bibinfo {volume} {110}},\ \bibinfo {pages} {042402} (\bibinfo {year} {2024}{\natexlab{b}})}\BibitemShut {NoStop}%
\bibitem [{\citenamefont {Hahn}\ \emph {et~al.}(2025)\citenamefont {Hahn}, \citenamefont {Takagi}, \citenamefont {Ferrini},\ and\ \citenamefont {Yamasaki}}]{hahn2026classicalsimulationquantumresource}%
  \BibitemOpen
  \bibfield  {author} {\bibinfo {author} {\bibfnamefont {O.}~\bibnamefont {Hahn}}, \bibinfo {author} {\bibfnamefont {R.}~\bibnamefont {Takagi}}, \bibinfo {author} {\bibfnamefont {G.}~\bibnamefont {Ferrini}},\ and\ \bibinfo {author} {\bibfnamefont {H.}~\bibnamefont {Yamasaki}},\ }\bibfield  {title} {\bibinfo {title} {Classical simulation and quantum resource theory of non-{G}aussian optics},\ }\href {https://doi.org/10.22331/q-2025-10-13-1881} {\bibfield  {journal} {\bibinfo  {journal} {{Quantum}}\ }\textbf {\bibinfo {volume} {9}},\ \bibinfo {pages} {1881} (\bibinfo {year} {2025})}\BibitemShut {NoStop}%
\bibitem [{\citenamefont {Dias}\ \emph {et~al.}(2026)\citenamefont {Dias}, \citenamefont {Bosse},\ and\ \citenamefont {Seddon}}]{dias2026optimalimprovedgatedecompositions}%
  \BibitemOpen
  \bibfield  {author} {\bibinfo {author} {\bibfnamefont {B.}~\bibnamefont {Dias}}, \bibinfo {author} {\bibfnamefont {J.~L.}\ \bibnamefont {Bosse}},\ and\ \bibinfo {author} {\bibfnamefont {J.~R.}\ \bibnamefont {Seddon}},\ }\href@noop {} {\bibinfo {title} {Optimal and improved gate decompositions for accelerated classical simulation of near-gaussian fermionic circuits}} (\bibinfo {year} {2026}),\ \Eprint {https://arxiv.org/abs/2603.18869} {arXiv:2603.18869 [quant-ph]} \BibitemShut {NoStop}%
\bibitem [{\citenamefont {Cottier}\ and\ \citenamefont {Chabaud}(2026)}]{cottier2026lowerboundscoherentstate}%
  \BibitemOpen
  \bibfield  {author} {\bibinfo {author} {\bibfnamefont {F.}~\bibnamefont {Cottier}}\ and\ \bibinfo {author} {\bibfnamefont {U.}~\bibnamefont {Chabaud}},\ }\href@noop {} {\bibinfo {title} {Lower bounds on coherent state rank}} (\bibinfo {year} {2026}),\ \Eprint {https://arxiv.org/abs/2604.00766} {arXiv:2604.00766 [quant-ph]} \BibitemShut {NoStop}%
\bibitem [{\citenamefont {Chabaud}\ \emph {et~al.}(2021)\citenamefont {Chabaud}, \citenamefont {Ferrini}, \citenamefont {Grosshans},\ and\ \citenamefont {Markham}}]{chabaud2021classicalsimulationgaussianquantum}%
  \BibitemOpen
  \bibfield  {author} {\bibinfo {author} {\bibfnamefont {U.}~\bibnamefont {Chabaud}}, \bibinfo {author} {\bibfnamefont {G.}~\bibnamefont {Ferrini}}, \bibinfo {author} {\bibfnamefont {F.}~\bibnamefont {Grosshans}},\ and\ \bibinfo {author} {\bibfnamefont {D.}~\bibnamefont {Markham}},\ }\bibfield  {title} {\bibinfo {title} {Classical simulation of {Gaussian} quantum circuits with non-{Gaussian} input states},\ }\href {https://doi.org/10.1103/PhysRevResearch.3.033018} {\bibfield  {journal} {\bibinfo  {journal} {Phys. Rev. Res.}\ }\textbf {\bibinfo {volume} {3}},\ \bibinfo {pages} {033018} (\bibinfo {year} {2021})}\BibitemShut {NoStop}%
\bibitem [{\citenamefont {Hackl}\ and\ \citenamefont {Bianchi}(2021)}]{hackl2024bosonicfermionicgaussianstates}%
  \BibitemOpen
  \bibfield  {author} {\bibinfo {author} {\bibfnamefont {L.}~\bibnamefont {Hackl}}\ and\ \bibinfo {author} {\bibfnamefont {E.}~\bibnamefont {Bianchi}},\ }\bibfield  {title} {\bibinfo {title} {Bosonic and fermionic gaussian states from kähler structures},\ }\href {https://doi.org/10.21468/SciPostPhysCore.4.3.025} {\bibfield  {journal} {\bibinfo  {journal} {SciPost Phys. Core}\ }\textbf {\bibinfo {volume} {4}},\ \bibinfo {pages} {025} (\bibinfo {year} {2021})}\BibitemShut {NoStop}%
\bibitem [{\citenamefont {Chabaud}\ and\ \citenamefont {Mehraban}(2022)}]{chabaud2022holomorphicrepresentationquantumcomputations}%
  \BibitemOpen
  \bibfield  {author} {\bibinfo {author} {\bibfnamefont {U.}~\bibnamefont {Chabaud}}\ and\ \bibinfo {author} {\bibfnamefont {S.}~\bibnamefont {Mehraban}},\ }\bibfield  {title} {\bibinfo {title} {Holomorphic representation of quantum computations},\ }\href {https://doi.org/10.22331/q-2022-10-06-831} {\bibfield  {journal} {\bibinfo  {journal} {{Quantum}}\ }\textbf {\bibinfo {volume} {6}},\ \bibinfo {pages} {831} (\bibinfo {year} {2022})}\BibitemShut {NoStop}%
\bibitem [{\citenamefont {Lami}\ \emph {et~al.}(2018)\citenamefont {Lami}, \citenamefont {Regula}, \citenamefont {Wang}, \citenamefont {Nichols}, \citenamefont {Winter},\ and\ \citenamefont {Adesso}}]{lami2018gaussianquantumresourcetheories}%
  \BibitemOpen
  \bibfield  {author} {\bibinfo {author} {\bibfnamefont {L.}~\bibnamefont {Lami}}, \bibinfo {author} {\bibfnamefont {B.}~\bibnamefont {Regula}}, \bibinfo {author} {\bibfnamefont {X.}~\bibnamefont {Wang}}, \bibinfo {author} {\bibfnamefont {R.}~\bibnamefont {Nichols}}, \bibinfo {author} {\bibfnamefont {A.}~\bibnamefont {Winter}},\ and\ \bibinfo {author} {\bibfnamefont {G.}~\bibnamefont {Adesso}},\ }\bibfield  {title} {\bibinfo {title} {Gaussian quantum resource theories},\ }\href {https://doi.org/10.1103/PhysRevA.98.022335} {\bibfield  {journal} {\bibinfo  {journal} {Phys. Rev. A}\ }\textbf {\bibinfo {volume} {98}},\ \bibinfo {pages} {022335} (\bibinfo {year} {2018})}\BibitemShut {NoStop}%
\bibitem [{\citenamefont {Landsberg}\ and\ \citenamefont {Manivel}(2003)}]{landsberg2002projectivegeometryhomogeneousspaces}%
  \BibitemOpen
  \bibfield  {author} {\bibinfo {author} {\bibfnamefont {J.}~\bibnamefont {Landsberg}}\ and\ \bibinfo {author} {\bibfnamefont {L.}~\bibnamefont {Manivel}},\ }\bibfield  {title} {\bibinfo {title} {On the projective geometry of rational homogeneous varieties},\ }\href {https://doi.org/10.1007/s000140300003} {\bibfield  {journal} {\bibinfo  {journal} {Commentarii Mathematici Helvetici}\ }\textbf {\bibinfo {volume} {78}},\ \bibinfo {pages} {65} (\bibinfo {year} {2003})}\BibitemShut {NoStop}%
\bibitem [{\citenamefont {Freire}\ \emph {et~al.}(2020)\citenamefont {Freire}, \citenamefont {Massarenti},\ and\ \citenamefont {Rischter}}]{freire2018projectiveaspectsgeometrylagrangian}%
  \BibitemOpen
  \bibfield  {author} {\bibinfo {author} {\bibfnamefont {A.~B.}\ \bibnamefont {Freire}}, \bibinfo {author} {\bibfnamefont {A.}~\bibnamefont {Massarenti}},\ and\ \bibinfo {author} {\bibfnamefont {R.}~\bibnamefont {Rischter}},\ }\bibfield  {title} {\bibinfo {title} {Projective aspects of the geometry of lagrangian grassmannians and spinor varieties},\ }\href {https://doi.org/10.1016/j.bulsci.2019.102829} {\bibfield  {journal} {\bibinfo  {journal} {Bulletin des Sciences Mathématiques}\ }\textbf {\bibinfo {volume} {159}},\ \bibinfo {pages} {102829} (\bibinfo {year} {2020})}\BibitemShut {NoStop}%
\bibitem [{\citenamefont {Manivel}(2009)}]{manivel2009spinorvarietiessecants}%
  \BibitemOpen
  \bibfield  {author} {\bibinfo {author} {\bibfnamefont {L.}~\bibnamefont {Manivel}},\ }\bibfield  {title} {\bibinfo {title} {On spinor varieties and their secants},\ }\href {https://doi.org/10.3842/sigma.2009.078} {\bibfield  {journal} {\bibinfo  {journal} {Symmetry, Integrability and Geometry: Methods and Applications}\ }\textbf {\bibinfo {volume} {5}},\ \bibinfo {pages} {078} (\bibinfo {year} {2009})}\BibitemShut {NoStop}%
\bibitem [{\citenamefont {Bhoja}\ and\ \citenamefont {Krasnov}(2025)}]{bhoja2023spinorspurespinors}%
  \BibitemOpen
  \bibfield  {author} {\bibinfo {author} {\bibfnamefont {N.}~\bibnamefont {Bhoja}}\ and\ \bibinfo {author} {\bibfnamefont {K.}~\bibnamefont {Krasnov}},\ }\bibfield  {title} {\bibinfo {title} {Spinors from pure spinors},\ }\href {https://doi.org/10.1093/qmath/haaf025} {\bibfield  {journal} {\bibinfo  {journal} {The Quarterly Journal of Mathematics}\ }\textbf {\bibinfo {volume} {76}},\ \bibinfo {pages} {659} (\bibinfo {year} {2025})}\BibitemShut {NoStop}%
\bibitem [{\citenamefont {Eckart}\ and\ \citenamefont {Young}(1936)}]{eckart1936approximationonematrixanotherlowerrank}%
  \BibitemOpen
  \bibfield  {author} {\bibinfo {author} {\bibfnamefont {C.}~\bibnamefont {Eckart}}\ and\ \bibinfo {author} {\bibfnamefont {G.}~\bibnamefont {Young}},\ }\bibfield  {title} {\bibinfo {title} {The approximation of one matrix by another of lower rank},\ }\href {https://doi.org/10.1007/BF02288367} {\bibfield  {journal} {\bibinfo  {journal} {Psychometrika}\ }\textbf {\bibinfo {volume} {1}},\ \bibinfo {pages} {211} (\bibinfo {year} {1936})}\BibitemShut {NoStop}%
\bibitem [{\citenamefont {Mirsky}(1960)}]{mirsky1960symmetricgaugefunctionsunitarilyinvariantnorms}%
  \BibitemOpen
  \bibfield  {author} {\bibinfo {author} {\bibfnamefont {L.}~\bibnamefont {Mirsky}},\ }\bibfield  {title} {\bibinfo {title} {Symmetric gauge functions and unitarily invariant norms},\ }\href {https://doi.org/10.1093/qmath/11.1.50} {\bibfield  {journal} {\bibinfo  {journal} {The Quarterly Journal of Mathematics}\ }\textbf {\bibinfo {volume} {11}},\ \bibinfo {pages} {50} (\bibinfo {year} {1960})}\BibitemShut {NoStop}%
\bibitem [{\citenamefont {Flajolet}\ and\ \citenamefont {Sedgewick}(2009)}]{flajolet2009analyticcombinatorics}%
  \BibitemOpen
  \bibfield  {author} {\bibinfo {author} {\bibfnamefont {P.}~\bibnamefont {Flajolet}}\ and\ \bibinfo {author} {\bibfnamefont {R.}~\bibnamefont {Sedgewick}},\ }\href@noop {} {\emph {\bibinfo {title} {Analytic Combinatorics}}}\ (\bibinfo  {publisher} {Cambridge University Press},\ \bibinfo {address} {Cambridge},\ \bibinfo {year} {2009})\BibitemShut {NoStop}%
\bibitem [{\citenamefont {Gigena}\ and\ \citenamefont {Rossignoli}(2015)}]{gigena2015entanglementfermionsystems}%
  \BibitemOpen
  \bibfield  {author} {\bibinfo {author} {\bibfnamefont {N.}~\bibnamefont {Gigena}}\ and\ \bibinfo {author} {\bibfnamefont {R.}~\bibnamefont {Rossignoli}},\ }\bibfield  {title} {\bibinfo {title} {Entanglement in fermion systems},\ }\href {https://doi.org/10.1103/PhysRevA.92.042326} {\bibfield  {journal} {\bibinfo  {journal} {Phys. Rev. A}\ }\textbf {\bibinfo {volume} {92}},\ \bibinfo {pages} {042326} (\bibinfo {year} {2015})}\BibitemShut {NoStop}%
\bibitem [{\citenamefont {Gr{\"o}chenig}\ \emph {et~al.}(2020)\citenamefont {Gr{\"o}chenig}, \citenamefont {Jaming},\ and\ \citenamefont {Malinnikova}}]{gröchenig2018zeroswignerdistributionshorttime}%
  \BibitemOpen
  \bibfield  {author} {\bibinfo {author} {\bibfnamefont {K.}~\bibnamefont {Gr{\"o}chenig}}, \bibinfo {author} {\bibfnamefont {P.}~\bibnamefont {Jaming}},\ and\ \bibinfo {author} {\bibfnamefont {E.}~\bibnamefont {Malinnikova}},\ }\bibfield  {title} {\bibinfo {title} {Zeros of the {Wigner} distribution and the short-time {Fourier} transform},\ }\href {https://doi.org/10.1007/s13163-019-00335-w} {\bibfield  {journal} {\bibinfo  {journal} {Revista Matem{\'a}tica Complutense}\ }\textbf {\bibinfo {volume} {33}},\ \bibinfo {pages} {723} (\bibinfo {year} {2020})}\BibitemShut {NoStop}%
\end{thebibliography}
\end{document}